\documentclass[10pt,letterpaper]{amsart}
\usepackage{amssymb}
\usepackage[dvips]{epsfig}
\usepackage{booktabs}
\usepackage{multirow}
\usepackage{xltabular}
\usepackage{tabularx}
\usepackage{makecell} 
\usepackage{ragged2e} 
\usepackage{graphicx}
\usepackage{color}
\usepackage{amsmath}
\allowdisplaybreaks
\usepackage{amssymb}
\usepackage{amsfonts}
\usepackage{geometry}
\usepackage{amsthm}
\usepackage[normalem]{ulem}    
\usepackage{cite}
\usepackage{tikz}

\newcommand{\Z}{\mathbb{Z}}

\newcommand{\N}{\mathbb{N}}
\newcommand{\R}{\mathbb{R}}
\newcommand{\T}{\mathbb{T}}
\newcommand{\E}{\mathbb{E}}

\newcommand{\wtg}{\widetilde{G_0}}

\newcommand{\admissible}{(*)}

\newcommand{\nm}{|\!|}

\newcommand{\mcO}{\mathcal{O}}
\newcommand{\mcG}{\mathcal{G}}

\newcommand{\supp}{{\rm supp}}

\newtheorem{thm}{Theorem}[section]
\newtheorem{lem}[thm]{Lemma}

\newtheorem{cor}{\bf Corollary}[section]

\newtheorem{rmk}{\bf Remark}[section]

\theoremstyle{definition}

\numberwithin{equation}{section}
\usepackage[colorlinks=true,pdfstartview=FitV,linkcolor=magenta,citecolor=cyan]{hyperref}
\usepackage{bm}
\begin{document}
\title[Delocalization  and  Bernoulli model]{Extended states  for the Random Schr\"odinger operator  on $\Z^d$ ($d\geq 5$)  with  decaying Bernoulli potential}

\author[Liu]{Shihe Liu}
\address[S. Liu] {School of Mathematical Sciences,
Peking University,
Beijing 100871,
China}
\email{2301110021@stu.pku.edu.cn}

\author[Shi]{Yunfeng Shi}
\address[Y. Shi] {School of Mathematics,
Sichuan University,
Chengdu 610064,
China}
\email{yunfengshi@scu.edu.cn}

\author[Zhang]{Zhifei Zhang}
\address[Z. Zhang] {School of Mathematical Sciences,
Peking University,
Beijing 100871,
China}
\email{zfzhang@math.pku.edu.cn}

\begin{abstract}
In this paper,  we investigate the delocalization property of the discrete  Schr\"odinger operator  $H_\omega=-\Delta+v_n\omega_n\delta_{n,n'}$, where $v_n=\kappa  |n|^{-\alpha}$ and $\omega=\{\omega_n\}_{n\in\Z^d}\in \{\pm 1\}^{\Z^d}$ is  a sequence of  i.i.d.  Bernoulli  random  variables.  Under the assumptions of $d\geq 5$, $\alpha>\frac14$ and $0<\kappa\ll1$, we construct   the   extended states for a deterministic renormalization of $H_\omega$ for most $\omega$.  This extends the   work   of Bourgain   [{\it Geometric Aspects of Functional Analysis}, LNM 1807:  70--98,  2003], where the case $\alpha>\frac13$ was  handled. Our proof is based on Green's function estimates via a $6$th-order renormalization scheme.  Among the main new ingredients  are  the proof of a  generalized   Khintchine inequality  via   Bonami's  lemma,  and the application of  the fractional  Gagliardo-Nirenberg inequality to control  a  new type of non-random  operators  arising from  the $6$th-order renormalization.  
\end{abstract}

\maketitle

\tableofcontents

\section{Introduction}

The Schr\"odinger operator on $\Z^d$ with homogeneous  i.i.d. random potentials, known as the Anderson model, was first introduced by Anderson \cite{And58} to describe the motion of a single quantum particle in disordered media. The study of the Anderson model primarily focuses on its spectral and dynamical properties. Of particular importance is the celebrated Anderson localization (i.e., pure point spectrum with exponentially decaying eigenfunctions) and delocalization (e.g., the existence of absolutely continuous spectrum) phase transition phenomenon. This phase transition relies heavily on the dimension $d$, the strength of the disorder, and the energy.  Indeed,  it is a general consensus that  Anderson localization (for the Anderson model) should occur  for all energies and all non-zero disorder if $d= 1,2$, while for the case of $d\geq 3$ and small disorder, there should exist an absolutely continuous spectrum in some energy interval. Mathematically,  localization  has been proven for three regimes: (i) for all energies and arbitrary disorder in $d = 1$, (ii) in any dimension and for all energies at large disorder, and (iii) near the edges of the spectrum in any dimension and for arbitrary disorder, cf. e.g., \cite{GMP77,KS80,FS83,FMSS85,DLS85,SW86,AM93,BK05,DS20}. However, the problem of proving the existence of the absolutely continuous spectrum for the Anderson model remains largely open(cf. e.g., Problem 1 in \cite{Sim00}).  In fact, even proving the existence of extended states (e.g., wave functions belonging to $\ell^\infty(\Z^d)\setminus\ell^2(\Z^d)$) for the Anderson model with non-zero disorder is far from reach. Delocalization  has  only been established for two special classes of random Schr\"odinger operators:  operators on Bethe lattices (cf. e.g.,  \cite{Kle98, ASW06}) and  operators  on $\Z^d$ with decaying random potentials (cf. e.g., \cite{Kri90,KKO00,JL00,Bou02,Bou03}).  

Now, consider the Anderson model $H_\omega=\Delta+\kappa V_\omega(n)\delta_{n,n'}$,  wheer $\kappa\in \R$ denotes the coupling,  $\Delta$  the discrete Laplacian, and $\{V_\omega(n)\}_{n\in\Z^d}$ is the potential given by  a  sequence of i.i.d. random variables. It is known that, for a broad class of random potentials(including the completely singular Bernoulli ones),  if $d=1$,  $H_\omega$  has Anderson localization almost surely  for all  $\kappa\neq 0$. However,  a new type of phase transition occurs  if $V_\omega(n)$ is replaced by  some decaying potential $V_\omega'(n)=|n|^{-\alpha}V_\omega(n)$ for some $\alpha>0.$  
More precisely, for  $H'_\omega=\Delta+\kappa V_\omega'(n)\delta_{n,n'}$ with $d=1$ and $\kappa\neq 0$, it has been proven in \cite{Sim82, DSS85,KLS98} that  the spectrum   is  almost surely dense pure point in   $(-2, 2)$  if  $0\leq \alpha<\frac12 $, and  is almost surely purely absolutely continuous in  $(-2, 2)$  if $\alpha>\frac12$.  In this important work  \cite{DSS85}, they  also proved purely singular spectrum in some energy region if $\alpha=\frac12.$ The proof of \cite{KLS98} relies crucially on  one-dimensional methods, such as the transfer matrix formalism, which may  not  be available in higher dimensions. Thus, it is  natural to ask  if the above phase transition diagram has a higher-dimensional analogy. In the remarkable work \cite{Bou03}, Bourgain provided a negative answer to this question and discovered  new  higher-dimensional phenomena. Specifically, he proved the existence of proper extended states for random   Schr\"odinger operators on $\Z^d$  with  decaying random potentials  $\kappa |n|^{-\alpha} {\omega}_n+\mathcal O(\kappa^2|n|^{-2\alpha})$  for  $d\geq 5$ and $\frac13<\alpha\leq \frac12,$  where $\{\omega_n\}_{n\in\Z^d}\in \{\pm 1\}^{\Z^d}$ is a sequence of   i.i.d.  Bernoulli  random variables.  The case of $\alpha=0$ corresponds to the standard Anderson-Bernoulli model, so improving the bound $\alpha>\frac 13$  to a smaller one is of significant  importance. In \cite{Bou03}, Bourgain remarked that ``{\it It is
likely that the method may be made to work for all $\alpha>0$...It is reasonable to expect this type of argument to succeed for any fixed $\alpha>0$ (with a number of resolvent iterations dependent on $\alpha$). To achieve this requires further renormalizations and taking care of certain additional difficulties due to the presence of a potential.}''  
Later in \cite{Bou08}, Bourgain revisited this problem and outlined a proof of the absence of dynamical localization for all  $\alpha>0.$  To the best of our knowledge, the existence of proper extended states for Schr\"odinger operators on $\Z^d$ ($d>1$) with decaying random  potential $\kappa |n|^{-\alpha} {\omega}_n+\mathcal O(\kappa^2|n|^{-2\alpha})$ satisfying $0<\alpha\leq \frac 13$ has remained completely open until the present paper.

In this paper, we aim to generalize the work of Bourgain \cite{Bou03} to the case of $\frac14<\alpha\leq \frac13$ via a further 6th-order renormalization scheme.  In this procedure, the presence of a new type of non-random operator of the 6th order  poses key challenge:  This operator cannot be written as  a symmetrical  combination of some  diagonal  and convolution  operators, so  the essential perturbation lemma (cf. Lemma 1.2)  in \cite{Bou03} does not apply.   To overcome this difficulty, we perform two arrangements on the resolvent expansion and move this operator to the 8th-order remaining terms. This  leads to the restriction of  $6\alpha+1>2$ (or  $\alpha>\frac 16$)  in dealing with the  8th-order remaining terms,  making it difficult to improve the  bound $\alpha>\frac14$ to $0<\alpha\leq \frac16$  via the present approach.   
Even in the estimates of the symmetrical type non-random operators of 6th order,  we introduce the fractional Gagliardo-Nirenberg inequality to perform the convolution  regularization. 
To establish moment estimates on both Green's function and extended states, we also prove  a generalized Khintchine inequality  via   Bonami's  lemma, which may be  of independent interest. Finally, we want to mention that  the present work is also motivated by  another famous open problem of Simon  (cf. Problem 8, \cite{Sim00}): The presence of absolutely continuous spectrum of Schr\"odinger operators  $-\Delta+V(x)$ on $\R^d$ provided $d\geq 2$ and 
\[\int_{\R^d}\frac{V^2(x)}{(|x|+1)^{d-1}}{\rm d}x<\infty.\]
In the present  context of $x\in\Z^d$  and   $V_\omega(x)\sim |x|^{-\alpha} $, the above condition is just $\alpha>\frac 12$.   For more results on the study of Schr\"odinger operators with decaying potentials, we refer to the excellent review \cite{DK07}.

\subsection{Main results}
We first introduce the notation. 
\begin{itemize}
\item For $x,y\in\R$, let 
\[x\wedge y=\min\{x,y\},\ x\vee y=\max\{x,y\}. \]
\item Throughout this paper, we denote 
\[|n|=(\max_{1\leq i\leq d}|n_i|)\vee 1\ {\rm for}\ n=(n_1,\cdots, n_d)\in \Z^d,\]
so $|0|=1.$  We also define for $\xi\in\R^d,$ 
\[|\xi|_1=\sum_{i=1}^d|\xi_i|,\ \|\xi\|=\sqrt{\sum_{i=1}^d|\xi_i|^2}.\]
\item For two nonnegative quantities $f$ and $g$, we write $f\lesssim g$, if there is an absolute constant $D>0$ 
such that $f\leq  Dg$. If we want to emphasize that $D$ depends on some  parameters $x,y,\cdots$ independent of $f,g$, then we write $f\lesssim_{x,y,\cdots}g$. 
\end{itemize}

Our main  model takes 
\begin{align}\label{model}
H_\omega=-\Delta+V_\omega^{(6)}(n)\delta_{n,n'}, \ n\in\Z^d,
\end{align}
where  the discrete Laplacian is defined by 
 \begin{align}\label{modified Laplacian}
\Delta(n,n')=\delta_{|n-n'|_1,1}-2d. 
\end{align}
   For the potential, we have 
 \begin{align}
V_\omega^{(6)}(n)=v_n\omega_n+V'(n),\ v=\{v_n\}_{n\in\Z^d}\ {\rm with}\ v_n=\kappa|n|^{-\alpha}, 
 \end{align}
where $\kappa\geq 0$ and $\omega=\{\omega_n\}_{n\in\Z^d}\in\{\pm 1\}^{\Z^d}$ is a sequence of i.i.d. Bernoulli random variables.  The deterministic potential $V'=\mathcal O(v^2)$ arising from the 6th-order renormalization scheme is defined explicitly by  \eqref{6-order renormalization of potential} (it depends only on $v$ and $G_0=(-\Delta)^{-1}$).  

Our first main result concerns the estimates of the Green's function.

\begin{thm}\label{greenthm}
Let $H_\omega$ be defined by \eqref{model} with fixed   $d\geq 5$ and $\frac14<\alpha\leq \frac13.$  Let $0<\varepsilon<\frac{4\alpha-1}{50}$. Then for any $p>\frac{2d+2}{\varepsilon}$, there is some $\kappa_0=\kappa_0(d,\alpha, p)$ so that the following holds true: If $0<\kappa\leq \kappa_0,$ then  there exists some $\Omega\subset \{\pm 1\}^{\Z^d}$ with $\mathbb P( \{\pm 1\}^{\Z^d}\setminus \Omega)\lesssim_{d,\alpha} \kappa^p$ so that for $\omega\in\Omega,$ we have (denote 
$G=G_\omega=H_\omega^{-1}$)  
\begin{align}
|G(n,n')|\lesssim  \frac{1}{|n-n'|^{d-2-\varepsilon}} \ {\rm for}\ \forall n,n'\in\Z^d. 
\end{align}
\end{thm}

\begin{rmk}
\begin{itemize}

\item The bound $d\geq 5$ primarily stems from the restriction $\sum\limits_{n\in\Z^d}|G_0(n,n')|^2<\infty$ (where $G_0=(-\Delta)^{-1}$) when applying  the Khintchine inequality.    It is noteworthy that such a bound is sufficient for the 6th-order renormalization scheme.

\item  The case of  $\alpha>\frac13$ has been addressed by Bourgain \cite{Bou03}, and it was conjectured there that the result  should hold for all $\alpha>0$.  It is possible  that the present method could be extended for $\frac16<\alpha\leq \frac14$ through further 10th-order renormalization. However, as we will see later (cf. e.g., Remark \ref{finrmk}), due to the presence of the aforementioned new type of non-random operator in the 6th-order renormalization, it appears challenging to extend the current result to the case of $0<\alpha\leq \frac16$.

\item The probability bound of $\kappa^p$ can be improved to $e^{-\frac{1}{\sqrt\kappa}}$ by employing  the Chernoff bound in the probability tail estimate. This is  because Bonami's lemma (cf. Lemma \ref{Bonami})  allows for an effective  estimate  of 
   \[\mathbb{E}_p(|f|)\leq (p-1)^{\frac{s}{2}}\mathbb{E}_2(|f|),\]
where $f$ is  the  Boolean polynomial  of degree $s$.   
\end{itemize}
\end{rmk}

Based on the above result, we also have

\begin{thm}\label{extthm}
Under the assumptions of Theorem \ref{greenthm}, there exists a set $\Omega'\subset\Omega$ with 
$\mathbb P(\Omega')=1-\mathcal O(\kappa^{\frac p2})$ so that,  for each $\omega\in\Omega'$, there is some $\zeta=\zeta_\omega=\{\zeta_\omega(n)\}_{n\in\Z^d}$ satisfying 
\begin{align}
H_\omega\zeta=0,\ \zeta=\hat\delta_0+\mathcal O(\sqrt{\kappa})\ {\rm in}\ \ell^\infty(\Z^d),
\end{align}
where $\hat \delta_0=\{\delta_0(n)\equiv 1\}_{n\in\Z^d}$. 
\end{thm}

\begin{rmk}
The Green's function estimates in Theorem \ref{greenthm} are  insufficient for the construction of extended states. Indeed,  it requires  the  second rearrangement of the resolvent expansion,  and additional random variables need to be removed  to prove the existence of proper extended states.  In this step, the generalized Khintchine inequality is again heavily employed to derive the probabilistic estimates. 
\end{rmk}

\subsection{Ideas of the proof}
Definitely, the main scheme of our proof is adapted from \cite{Bou03}. As mentioned by Bourgain \cite{Bou03}, his approach is also motivated by the one initiated by  Spencer \cite{Spe93} (cf. \cite{Elg09} for a related result), but is technically different: it replaces  the Feynman diagram machinery with the random decoupling estimate.   

For simplicity, we only outline the proof of Theorem \ref{greenthm}, and the proof of Theorem \ref{extthm} remains similar. Note first that (cf. e.g., \cite{MS22}) for $G_0=(-\Delta)^{-1},$ we have
\begin{align}\label{g0-1}
|G_0(n,n')|\lesssim_d\frac{1}{|n-n'|^{d-2}}\ {\rm for}\ \forall n,n'\in\Z^d. 
\end{align}
Denote $\widetilde V=V_\omega^{(6)}$ with $V^{(6)}$ given by \eqref{model}. From the resolvent identity, one can write down a 8th-order  (in $v$) Born series expansion  for $G=H^{-1}$
\begin{align*}
 G=&G_0 -G_0\widetilde{V}G_0+G_0\widetilde{V}G_0\widetilde{V}G_0-G_0\widetilde{V}G_0\widetilde{V}G_0\widetilde{V}G_0+G_0\widetilde{V}G_0\widetilde{V}G_0\widetilde{V}G_0\widetilde{V}G_0\\
   \notag &-G_0\widetilde{V}G_0\widetilde{V}G_0\widetilde{V}G_0\widetilde{V}G_0\widetilde{V}G_0+G_0\widetilde{V}G_0\widetilde{V}G_0\widetilde{V}G_0\widetilde{V}G_0\widetilde{V}G_0\widetilde{V}G_0\\
   \notag &-G_0\widetilde{V}G_0\widetilde{V}G_0\widetilde{V}G_0\widetilde{V}G_0\widetilde{V}G_0\widetilde{V}G_0\widetilde{V}G_0+G\widetilde{V}G_0\widetilde{V}G_0\widetilde{V}G_0\widetilde{V}G_0\widetilde{V}G_0\widetilde{V}G_0\widetilde{V}G_0\widetilde{V}G_0,\\
   :=&A+GB,
\end{align*}
where $A$ contains the  $i$ th-order remaining terms for $i\leq 7$, and $B$ contains  the $8$th-order ones.  From \eqref{g0-1}, it follows that $G_0$ may be unbounded on $\ell^2(\Z^d)$. Thus, it is more appropriate to control $G(n,n')$ for every $n,n'$.  This then leads to the study  of multiple infinite summations, such as, 
\[G_0VG_0VG_0(n,n')=\sum_{n_1,n_2\in\Z^d}G_0(n, n_1)\omega_{n_1}v_{n_1}G_0(n_1,n_2)\omega_{n_2}v_{n_2}G_0(n_2,n').\]
Again by \eqref{g0-1}, we observe that it is challenging to obtain a good estimate on the summation satisfying  $n_1=n_2 \in\Z^d$, since in this case the randomness cancels (i.e., $\omega_{n_1}\omega_{n_2}\equiv 1$). However, for the summation with  $n_1\neq n_2$, we can use  the generalized Khintchine inequality (cf. Lemma \ref{decoupling}) to  get for any $p\geq 2,$
\begin{align*}
(\E|G_0VG_0VG_0(n,n')|^p)^{\frac 2p}&\leq \sum_{n_1\neq n_2\in\Z^d}|G_0(n, n_1)|^2v_{n_1}^2|G_0(n_1,n_2)|^2v_{n_2}^2|G_0(n_2,n')|^2\\
&\lesssim  \sum_{n_1\neq n_2\in\Z^d}\frac{1}{|n-n_1|^{2(d-2)}|n_1|^{2\alpha}|n_1-n_2|^{2(d-2)} |n_2|^{2\alpha}|n-n_1|^{2(d-2)}}.
\end{align*}
From the assumption of  $d\geq 5$, we know  $2(d-2)>d$ and the above summation can be well controlled. We expect that  the estimates on other terms in $A$ and $B$ should be  similar,   but require  much  more efforts.  

Indeed, motivated by  the above argument, we can distinguish  terms in $A$ (and $B$)  into two classes: random terms and non-random ones.  As we will see later,  the generalized  Khintchine inequality only works for {\it admissible}  random summations (cf. Lemma \ref{decoupling}), but not all random ones. This   would   prevent us from controlling those {\it non-admissible random terms}.  However, it is remarkable that in the renormalization scheme (at least for the $6$th-order one),  non-admissible random terms {\it automatically offset}  with each other and do not appear at all.

For the non-random terms, one can renormalize the potential to eliminate terms of the form $G_0v^2G_0, G_0v^4G_0, G_0v^6G_0$, since those terms cannot be  well controlled.   In fact, there are also non-random terms that cannot be eliminated but with a symmetrical form, such as $G_0WG_0$, where  $W=v^2Mv^2$ and $M$ is a convolution operator mainly coming from $\hat G_0*\hat G_0*\hat G_0$. Such  term is of 4th order. While $G_0WG_0$ cannot be well  controlled directly, one can use the symmetrical difference trick and convolution  regularization argument  to decompose $G_0WG_0$ into several operators, each of which has the desired estimates as in \cite{Bou03}.  

We  want to emphasize  that, however,  in the present 6th-order renormalization scheme, a new type of non-random operator $G_0CG_0$ (cf. \eqref{4.9}) appears.  By developing a  more complicated symmetrical difference trick (cf. the proof of Theorem \ref{6,7 resolvent estimate}),  we can obtain  $C=(C-P_6'')+P_6''$ with $G_0(C-P_6'')G_0$ having a good control. The singular operator $P_6''$, given by 
\[P_6^{''}(n_1,n_3)= \wtg(n_1,n_3) \sum_{n_2\in \Z^d}(v_{n_2}^6-v_{n_1}^6) \wtg(n_1,n_2)^2 \wtg(n_2,n_3)^2,\]
cannot be well handled: it only has the estimate 
\[|P_6^{''}(n_1,n_3)| \lesssim_{d,\alpha} \kappa^6\frac{1}{|n_1-n_3|^{3(d-2)-1}(|n_3|\wedge|n_1|)^{6\alpha+1}},\]
  rather than a $6\alpha+2$   decay rate  as required by the renormalization scheme. 
 Clearly, the operator $P_6''$ cannot be written as a symmetrical combination of diagonal and convolutional operators as that of   $G_0WG_0$ or $G_0(C-P_6'')G_0$.  So we have  to do the arrangement on the Born series and move $P_6''G_0$ to $B'$ so that $$G=A'+GB',$$ and  thus 
  \begin{align*}
|G_0P_6''(n,n')|&\lesssim_{d,\alpha} \kappa^6 \frac{1}{|n-n'|^{d-2}(|n|\wedge|n'|)^{6\alpha+1}}\\
&\lesssim \kappa^6 \frac{1}{|n-n'|^{d-2}(|n|\wedge|n'|)^{8\alpha}}\ ({\rm since}\ \alpha\leq \frac13),
\end{align*} 
which  suffices for the moment estimates. 

Finally, we arrive at 
\begin{align*}
 \E_p|(A'-G_0)(n,n')|   &\lesssim_{d,p,\alpha}\kappa \frac{1}{|n-n'|^{d-2}(|n|\wedge|n'|)^{\alpha}},\\
 \E_p|B'(n,n')| &\lesssim_{d,p,\alpha} \kappa^6 \frac{1}{|n-n'|^{d-2}(|n|\wedge|n'|)^{8\alpha}}. 
 \end{align*}
And we can use the Chebyshev's inequality to get good estimates on $A'(n,n'), B(n,n')$ with high probability. 
To get desired estimates on $G$,     it  requires  the existence of $(I-B')^{-1}$,  which leads to the condition $8\alpha>2$, namely, $\alpha>\frac14$. \smallskip

Thus, the main novelties of our proof are as follows: 

\begin{itemize}

\item We introduce graph representations to compute the remaining terms in the 6th-order renormalization scheme. We also identify some iteration relations between remaining terms of different orders. These arguments allow us to easily detect the remarkable offsets between non-admissible random terms and perform more flexible rearrangements of the remaining terms in the Born series expansion. For details, refer to Appendices \ref{App6th} and \ref{App7th}. 

\item  As mentioned above, a new type of non-random operator emerges in the 6th-order renormalization scheme, posing a key challenge. 
While we rearrange the terms so that the singular operator $P_6''$ can be moved to $B'$, controlling the non-singular operator  $G_0(C-P_6'')G_0$ is also non-trivial. Indeed,  we propose a  new  symmetrical difference argument (cf. the proof of Theorem \ref{6,7 resolvent estimate}),   which turns out to be more complicated  than that in \cite{Bou03}. For the convolution regularization argument, we also need to handle some convolution  operator  (given by  $f^2$ with $f$ defined  by \eqref{fxi}),  which is not entirely a convolution product of $\hat G_0$. For this, we use  the fractional   Gagliardo-Nirenberg inequality (cf.  \cite{BM18})  to obtain  fine estimates on   higher-order derivatives of $f^2$. 
 
\item  We prove a generalized  Khintchine inequality  (cf. Lemma \ref{decoupling}) based on hypercontractivity estimate (e.g., Bonami's  lemma). Previously, Bourgain \cite{Bou03} 
employed  random decoupling to establish the $L^2\to L^2$ estimate. Our new contribution here  is a proof of the $L^p\to L^2$ estimate for any $p\geq1$, which may be of independent interest. 
Since we can get  directly high-order moment estimates on Green's functions,  the probability estimate in the proof of Theorem \ref{greenthm} becomes more straightforward, and the application of Chebyshev's inequality suffices for this purpose.
\end{itemize}

\subsection{Structure of the  paper}

The paper is organized as follows.  In \S \ref{presec}, we introduce some basic but useful estimates  on   products of $G_0$ and $G_0vG_0$; in \S \ref{khisec}, we employ Bonami's lemma to prove a generalized Khintchine inequality involving admissible tuples.  In \S \ref{6thsec}, we present  the 6th-order renormalization result. In \S \ref{thm1sec}, we prove our first main result on Green's function estimates  (cf. Theorem \ref{greenthm}). In \S \ref{thm2sec}, we construct the desired  extended states,  thereby completing the proof of Theorem \ref{extthm}. The computations of the 6th 
and 7th-order remaining terms are completed  in Appendix \ref{App6th} and  Appendix  \ref{App7th}, respectively. The proofs of several key technical lemmas can be found in Appendix \ref{tecApp}. In Appendix \ref{GNineq},  the fractional Gagliardo-Nirenberg inequality is employed  to prove  Lemma \ref{FGNlem}. 

\section{Preliminaries: some useful  estimates}\label{presec}
In this section, we will introduce  some useful estimates concerning  products of  both $G_0=(-\Delta)^{-1}$ and  $G_0vG_0$.  

Recall the discrete Laplacian  
\begin{equation*}
\Delta(n,n')=\delta_{|n-n'|_1,1}-2d
\end{equation*}
and  its Fourier transform 
\begin{equation}\label{Laplacian frequecy}
  -\hat{\Delta}(\xi)=2d-2\sum_{j=1}^{d}\cos 2\pi\xi_j=c\|\xi\|^2+\mcO(\|\xi\|^4),
\end{equation}
where $c>0$ is some absolute constant. Denote by 
\[G_0(n,n')=(-\Delta)^{-1}(n,n')=\int_{\T^d}\frac{e^{-2\pi i(n-n')\cdot \xi}}{-\hat{\Delta}(\xi)}\mathrm{d}\xi,\ \T^d=\R^d/\Z^d\]
 the resolvent (or the Green's function) of  $-\Delta$.   
A  standard  estimate  on  $G_0$  (cf. e.g., \cite{MS22})  is  
\begin{equation}\label{critical decay of G_0}
  |G_0(n,n')|\lesssim_d \frac{1}{|n-n'|^{d-2}}.
\end{equation}

In this paper, we have to   control operators involving products of $G_0$. Therefore, it is useful to introduce some estimates on summations of power-law decay sequences. Recall that $G_0$ is an {\it unbounded}  operator on $\ell^2(\Z^d).$ 

The first important lemma reads as 
\begin{lem}\label{Lemma 2.1}
For any $a,b>0$ satisfying $a+b>d$ and  $\max\{a,b\}\neq d$, we have 
\[\sum_{n_1\in \Z^d}\frac{1}{|n_1|^a |m-n_1|^b}\lesssim_{a,b,d}\frac{1}{|m|^{\min\{a,b,a+b-d\}}}.\]
\end{lem}

\begin{rmk}\label{rmk 2.1}
  \begin{itemize}
    \item[(1)] As we will see from the proof of Lemma \ref{Lemma 2.1},  if $a\leq d=b$, then  the  estimate  becomes 
              \[\sum_{n_1\in \Z^d}\frac{1}{|n_1|^a |m-n_1|^d}\lesssim_{a,d}\frac{\log |m|}{|m|^a}\lesssim \frac{1}{|m|^{a-}}.\]
    \item[(2)]  As an application of Lemma \ref{Lemma 2.1}, we can recover an upper bound  on products  of $G_0 $.  More precisely,  consider 
    \begin{equation}\label{q-fold Green's fucntion}
      G_0^q(n,n')=\int_{\T^d}\frac{e^{-2\pi i (n-n')\cdot \xi}}{(-\hat{\Delta}(\xi))^q}\mathrm{d}\xi, \ q\in \N. 
    \end{equation} 
   Repeatedly applying  Lemma \ref{Lemma 2.1} yields  for $2\leq q<\frac d2,$
    \begin{align*}
     |G_0^q(n,n')|  
                &\leq \sum_{n_1,n_2,\cdots,n_{q-1}\in \Z^d} |G_0(n,n_1)|\cdot|G_0(n_1,n_2)|\cdots|G_0(n_{q-1},n')|\\
                &\lesssim_d \sum_{n_1,n_2,\cdots,n_{q-1}\in \Z^d} \frac{1}{|n-n_1|^{d-2}}\cdot \frac{1}{|n_1-n_2|^{d-2}}\cdots \frac{1}{|n_{q-1}-n'|^{d-2}}\\
                &\lesssim_d \sum_{n_2,\cdots,n_{q-1}\in \Z^d} \frac{1}{|n-n_2|^{d-4}}\cdot \frac{1}{|n_2-n_3|^{d-2}}\cdots \frac{1}{|n_{q-1}-n'|^{d-2}}\\
                &\cdots  \\
                &\lesssim_{d,q}\frac{1}{|n-n'|^{d-2q}}.
    \end{align*}
    Note that we have the  $q$-loss in the above  estimate  on   $G_0^q.$
   \end{itemize}
\end{rmk}

\begin{proof}
  We refer to Appendix \ref{tecApp} for a detailed proof. 
\end{proof}

The next  lemma aims to control  summations involving products  of $G_0vG_0$.
\begin{lem}\label{Lemma 2.2}
  For any $0<\varepsilon<d$, $0<a\leq b$ satisfying $b+\varepsilon >d$ and $b\neq d$, we have 
  \[\sum_{n_1\in \Z^d}\frac{1}{|n-n_1|^a |n_1|^{\varepsilon} |n_1-n'|^b}\lesssim_{a,b,\varepsilon,d}\frac{1}{|n-n'|^a(|n|\wedge |n'|)^{\min\{\varepsilon,a, \varepsilon+b-d\}}}.\]
  \end{lem}

\begin{rmk}
  \begin{itemize}
    \item[(1)] Especially, if  $0<\varepsilon<a=b<d,a+\varepsilon>d$,  then 
  \begin{equation}\label{2.19}
    \sum_{n_1\neq 0}\frac{1}{|n-n_1|^a|n_1|^{\varepsilon}|n_1-n'|^a}\lesssim_{a,\varepsilon,d}\frac{1}{|n-n'|^a (|n|\wedge|n'|)^{a+\varepsilon-d}}. 
  \end{equation}
 
   \item[(2)] If  $b=d$,  similar to  Remark  \ref{rmk 2.1}  (1),  we have 
  \[\sum_{n_1\in \Z^d}\frac{1}{|n-n_1|^a |n_1|^{\varepsilon} |n_1-n'|^d}\lesssim_{a,b,\varepsilon,d}\frac{1}{|n-n'|^a(|n|\wedge |n'|)^{\min\{a,\varepsilon\}-}}.\]
\end{itemize}
   
\end{rmk}

\begin{proof}
   We refer to  Appendix \ref{tecApp} for a detailed proof. 
 \end{proof}

\section{A generalized Khintchine inequality}\label{khisec}
This section is devoted to proving a generalized Khintchine inequality via  the {\it hypercontractivity} estimate (cf. e.g., \cite{Jan97,SS12,ODO14}), which plays an essential role in our estimates on the  Green's function.  In contrast, in \cite{Bou03}, Bourgain proposed  an  analogous inequality based on the  standard $L^2$-random {\it decoupling}. 
Our proof builds on Bonami's  lemma, which primarily focuses on Boolean functions estimates.

We first introduce Bonami's   lemma \cite{Bon70}. 

\begin{lem}[cf. \cite{Bon70, ODO14}]\label{Bonami}
  Let $m, s \in \N$  and let  $f(Y_1,Y_2,\cdots,Y_m)$  be  a real-valued  polynomial  in  i.i.d. Bernoulli random variables $Y_1,\cdots,Y_m\in\{\pm 1\}$ with  the degree $\deg (f)=s$. Then 
  \[\mathbb{E}(f^4)\leq 9^s (\mathbb{E}f^2)^2.\]
\end{lem}
\begin{proof}
 For completeness, we give a proof  here.   Define 
  \[(D_m f)(Y_1,Y_2,\cdots,Y_{m-1})=\frac{1}{2}(f(Y_1,\cdots,Y_{m-1},1)-f(Y_1,\cdots,Y_{m-1},-1)),\]
  \[(E_m f)(Y_1,Y_2,\cdots,Y_{m-1})=\frac{1}{2}(f(Y_1,\cdots,Y_{m-1},1)+f(Y_1,\cdots,Y_{m-1},-1)).\]
  Since $Y_1,\cdots, Y_m$ are i.i.d. Bernoulli random variables, we  know  that $Y_m$ is independent  of  $D_m f,E_m f$,  and  
  \begin{equation}\label{decomposition of Boolean function}
    f(Y_1,Y_2,\cdots,Y_m)=Y_m\cdot (D_m f) +(E_m f).
  \end{equation}
  The proof is based on an  induction on $m$. Indeed, when $m=0$, the polynomial $f$ is a constant and  Lemma \ref{Bonami} holds trivially. Now,  assume that Lemma \ref{Bonami}   holds for polynomials with  $m-1$ variables.  By using the decomposition \eqref{decomposition of Boolean function} and the independence property, we obtain 
  \begin{align*}
    \E(f^4) &= \E(Y_m\cdot D_m f +E_m f)^4 \\
       &=\E(Y_m^4)\E(D_m f)^4+ 4\E(Y_m^3)\E((D_m f)^3 \cdot (E_m f)) + 6\E(Y_m^2)\E((D_m f)^2 \cdot (E_m f)^2)\\
         &\ \ \  +4 \E(Y_m)\E((D_m f) \cdot (E_m f)^3) +\E(E_m f)^4\\
        &=\E(D_m f)^4+ 6\E((D_m f)^2 \cdot (E_m f)^2)+ \E(E_m f)^4.
  \end{align*}
  Similarly, 
  \begin{align*}
    \E(f^2)=\E(D_m f)^2+\E(E_m f)^2.
  \end{align*}
  Since   $f$ is a polynomial of degree $s$, $D_m f$ is a polynomial of degree $s-1$ and $E_m f$ is a polynomial of degree $s$. By  the induction assumption, we get  
  \begin{equation}\label{3.2}
    \E(D_m f)^4\leq 9^{s-1} (\E(D_m f)^2)^2,
  \end{equation}
  \begin{equation}\label{3.3}
    \E(E_m f)^4\leq 9^{s} (\E(E_m f)^2)^2.
  \end{equation}
Using the  Cauchy-Schwarz inequality implies 
  \begin{equation}\label{3.4}
    \E((D_m f)^2\cdot(E_m f)^2)\leq (\E(D_m f)^4) ^{\frac{1}{2}}\cdot (\E(E_m f)^4)^{\frac12}\leq \frac13 \cdot 9^s \E(D_m f)^2\cdot  \E(E_m f)^2.
  \end{equation}
 Combining  \eqref{3.2}, \eqref{3.3} and \eqref{3.4}  shows 
  \begin{align*}
    \E(f^4) &\leq 9^{s-1} (\E(D_m f)^2)^2 +2\cdot 9^s \E(D_m f)^2\cdot  \E(E_m f)^2 +9^{s} (\E(E_m f)^2)^2\\
     &\leq 9^s (\E f^2)^2. 
  \end{align*}
  This finishes the induction step (i.e., $m-1\to m$), and hence the proof. 
\end{proof}
As a  corollary of Lemma \ref{Bonami}, we have 
\begin{cor}\label{moment equivalence}
Under the assumptions of  Lemma \ref{Bonami}, we have for all $p\geq 1,$ 
\begin{align}\label{peqv}
\E_p|f| := (\E|f|^p)^{\frac{1}{p}}\lesssim_{p,s} \E_2|f|.
\end{align}
More generally, if  $\{Y_{n}\}_{n\in\Z^d}$ is  a sequences of  i.i.d. random Bernoulli variables and 
$$f=\sum_{n_1,\cdots, n_s\in\Z^d}a_{n_1,\cdots, n_s}Y_{n_1}\cdots Y_{n_s}\  {\rm with}\  a_{n_1,\cdots, n_s}\geq 0,$$ then the estimate   \eqref{peqv} remains true for this $f$. 

\end{cor}
\begin{proof}
  If $1\leq p \leq  2$,  then we get by   H\"older inequality that 
  \[\E(|f|^p)\leq (\E|f|^2)^{\frac{p}{2}} \cdot (\E1)^{1-\frac{p}{2}} =(\E |f|^2)^{\frac{p}{2}}, \] 
  which implies \eqref{peqv} in this case. 
  
  If $p>2$, we first  consider the cases of $p=2^k,k=2,3,\cdots$. Note that $f^2$ is a polynomial of degree at most $2s$. By Lemma \ref{Bonami}, we have 
  \begin{align*}
    \E(f^8) = \E((f^2)^4) &\leq 9^{2s} (\E f^4)^2\\
       &\leq 9^{2s} (9^s (\E f^2)^2)^2\\
       &\lesssim_{s} (\E f^2)^{4}.
  \end{align*} 
   Repeatedly applying  Lemma \ref{Bonami}  yields  
  \[(\E f^{2^k})^{\frac{1}{2^k}}\lesssim_{k,s} (\E f^2)^{\frac{1}{2}}.\]
  Next,  using the standard interpolation inequality  gives (by $p\in [2^k, 2^{k+1}]$ if $p>2$)
  \[(\E |f|^{p})^{\frac{1}{p}}\lesssim_{p,s} (\E f^2)^{\frac{1}{2}}, \ p> 2. \]
  This proves \eqref{peqv} if $p>2$.

  Now, we  consider the   $\{Y_n\}_{n\in\Z^d}$ case. Denote for $N\geq1,$
  $$f_N=\sum_{|n_1|\leq N,\cdots, |n_s|\leq N}a_{n_1,\cdots, n_s}Y_{n_1}\cdots Y_{n_s}.$$
Then applying \eqref{peqv} to $f_N$ gives 
  \[\E_p|f_N| \lesssim_{p,s} \E_2|f_N| \leq \E_2|f|,\]
  where for  the last inequality, we used the  fact that  the  coefficient $a_{n_1,\cdots, n_s}\geq 0$  and 
  \[\E(Y_{n_1}^{d_1} Y_{n_2}^{d_2}\cdots Y_{n_k}^{d_k})=0 \ {\rm or} \ 1\ {\rm for}\ d_1,\cdots,d_k\in \N.\]
  So from Fatou's lemma, it   follows that $\E_p|f|\lesssim_{p,s} \E_2|f|, p\geq 1$.
\end{proof}
Next,  recall that $\{\omega_n\}_{n\in\Z^d}\in\{\pm 1\}^{\Z^d}$  is the i.i.d. random Bernoulli  variables. For a $s$-tuple $(n_1,n_2,\cdots,$ $n_s)$, we say that its randomness ``\textbf{cancels}'' if 
\[\mathbb{P}\big(\prod_{i=1}^{s}\omega_{n_i}=1\big)=1.\]
It's easy to see that the randomness of $(n_1,n_2,\cdots,n_s)$   cancels if and only if each $n_i$ ($1\leq i\leq s$)  is repeated  an  {\it even  number}  of times  in the $s$-tuple.  We say that $(n_1,n_2,\cdots,n_s)$ is ``\textbf{admissible}''  if for any  $1\leq s_1<s_2\leq s$, the randomness of sub-tuple $(n_{s_1},n_{s_1+1},\cdots,n_{s_2})$ does not cancel. We use the notation $\sum\limits^{(*)}\limits_{n_1,\cdots,n_s}$ to indicate a summation restricted to admissible $s$-tuples. 
We  then  introduce  the generalized Khintchine  inequality, which is a refined  version  of Lemma 2.2 of Bourgain  \cite{Bou03}.
\begin{lem}\label{decoupling}
Let $\{\omega_n\}_{n\in\Z^d}$ be a sequence of i.i.d. random Bernoulli variables. For $s\geq 1$ and $p\geq 2$, we have 
\begin{align}\label{Khintchine}
  \E_p\left|\sum_{n_1,\cdots,n_s}^{(*)} \omega_{n_1}\cdots\omega_{n_s} a^{(0)}_{n,n_1}a^{(1)}_{n_1,n_2}\cdots a^{(s)}_{n_s,n'}\right|\lesssim_{p,s} \left[\sum_{n_1,\cdots,n_s}|a^{(0)}_{n,n_1}a^{(1)}_{n_1,n_2}\cdots a^{(s)}_{n_s,n'}|^2   \right]^{\frac{1}{2}},
\end{align}  
where all $a_{m,n}^{(j)}\in\R.$
\end{lem}
\begin{rmk}
 If  $s=1$, Lemma \ref{decoupling} is just the classical Khintchine  inequality.
\end{rmk}
\begin{proof}[Proof of Lemma \ref{decoupling}.]
  Without loss of generality, we can assume $a_{m,n}^{(j)}\geq 0$. Then by  Corollary \ref{moment equivalence}, it  suffices to prove \eqref{Khintchine}  for $p=2$, which will be completed by  induction on $s$ below.

  If $s=1$, then  by the orthogonality of $\{\omega_n\}_{n\in \Z^d}$ in $L^2$, we have
  \[  \E_2\left| \sum_{n_1\in \Z^d} \omega_{n_1} a^{(0)}_{n,n_1}a^{(1)}_{n_1,n'}\right|= \left[\sum_{n_1\in\Z^d}|a^{(0)}_{n,n_1}a^{(1)}_{n_1,n'}|^2   \right]^{\frac{1}{2}}.\]
  Now,  assume \eqref{Khintchine} holds  with $s$ replaced by $s'\leq s-1$ and $p=2$. Since in the summation ``$\sum^{(*)}$" no tuple $(n_1,\cdots,n_s)$'s randomness cancels, there are some  distinct   $m_1,m_2,\cdots,m_k$ (as a sub-tuple of $(n_1,\cdots,n_s)$), each of  which  is repeated an odd number of  times. Specify  all possible positions of those sits  as disjoint $I_1,I_2,\cdots,I_k$. That is to say, for $i=1,\cdots,k, $ 
  \[I_i=\{k:\ n_k=m_i\}\subset \{1,2,\cdots,s\}.\]
  After this specifying  of $I_1,\cdots,I_k$, the $s$-tuple $(n_1,\cdots, n_s)$ has the form of  
  \[(\nu^{(1)},m_{i_1},\cdots, \nu^{(2)},m_{i_2},\cdots),\]
  where $\nu^{(1)},\nu^{(2)},\cdots,\nu^{(l)}$ are admissible sub-tuples with indexes   determined by  $i_j\in I_j$ ($1\leq j\leq k$). By the Minkowski inequality, we obtain  
  \begin{align}\label{3.6}
     &\E_2\left| \sum_{n_1,\cdots,n_s}^{(*)} \omega_{n_1}\cdots\omega_{n_s}  a^{(0)}_{n,n_1}a^{(1)}_{n_1,n_2}\cdots a^{(s)}_{n_s,n'}\right|  \\
    \notag &\leq \sum_{I_1,\cdots,I_k}\E_2\left |\sum_{m_1,\cdots,m_k}\omega_{m_1}\cdots\omega_{m_k}   \underbrace{\left( \sum_{\nu^{(1)}=(n_{s_0},\cdots,n_{s_1})}^{(*)}\omega_{n_{s_0}}\cdots a_{m_{i_1},n_{s_0}}^{(s_0-1)}\cdots\right)\left(\sum_{\nu^{(2)}}^{(*)}\cdots\right)\cdots}_{A_{m_1,\cdots,m_k}} \right| \\
    \notag & = \sum_{I_1,\cdots,I_k} \E_2\left | \sum_{\{m_1,\cdots,m_k\} } \omega_{m_1}\cdots \omega_{m_k} (\sum_{\sigma\in S_k}A_{m_{\sigma(1)},\cdots,m_{\sigma(k)}}) \right|,
  \end{align}
  where $S_k, k\leq s$ denotes  the $k$-order permutation group, and $A_{m_1,\cdots,m_k}$ has  indeed   no randomness while each $\nu^{(i)}$ is admissible. Note that we have the orthogonality relation 
  \[\{m_1,\cdots,m_k\}\neq \{m_1',\cdots,m_k'\}\Rightarrow \E[(\omega_{m_1}\cdots\omega_{m_k})\cdot (\omega_{m_1'}\cdots\omega_{m_k'})]=0,\]
  \[\{m_1,\cdots,m_k\}= \{m_1',\cdots,m_k'\}\Rightarrow \E[(\omega_{m_1}\cdots\omega_{m_k})\cdot (\omega_{m_1'}\cdots\omega_{m_k'})]=1.\]
  Hence,
  \begin{align}
    \notag \eqref{3.6} &\leq \sum_{I_1,\cdots,I_k} \left[ \sum_{\{m_1,\cdots,m_k\}} (\sum_{\sigma\in S_k}A_{m_{\sigma(1)},\cdots,m_{\sigma(k)}})^2   \right]^{\frac12} \\
     \notag &\leq \sum_{I_1,\cdots,I_k} \left[ (\# S_k)\cdot \sum_{\{m_1,\cdots,m_k\}} \sum_{\sigma\in S_k}(A_{m_{\sigma(1)},\cdots,m_{\sigma(k)}})^2   \right]^{\frac12} \\
      \notag &\leq (\# S_s)^{\frac12} \cdot \sum_{I_1,\cdots,I_k} \left[  \sum_{m_1,\cdots,m_k} (A_{m_{1},\cdots,m_{k}})^2   \right]^{\frac12} \\
       \label{3.7}  &\leq (\# S_s)^{\frac12} \cdot \sum_{I_1,\cdots,I_k}  \left[  \sum_{m_1,\cdots,m_k}   \E\bigg | \big[\sum_{\nu^{(1)}}^{(*)}\cdots\big]\cdots \big[\sum_{\nu^{(l)}}^{(*)}\cdots\big]\bigg |^2     \right]^{\frac12},  
  \end{align}
 where  for the second inequality, we  apply the  Cauchy-Schwarz inequality, and for the third inequality,  we use  $k\leq s$. We continue to control  \eqref{3.7} by using H\"older's inequality and  Corollary \ref{moment equivalence}, and  get (since $l\leq s$)
  \begin{align}
   \notag \eqref{3.7} &\lesssim_s \sum_{I_1,\cdots,I_k}  \left[  \sum_{m_1,\cdots,m_k}   \bigg (\E_{2l}|\sum_{\nu^{(1)}}^{(*)}\cdots|\bigg )^2\cdots\bigg (\E_{2l}|\sum_{\nu^{(l)}}^{(*)}\cdots|\bigg )^2     \right]^{\frac12}\\
    \label{3.8}   &\lesssim_s\sum_{I_1,I_2,\cdots,I_k}  \left[  \sum_{m_1,\cdots,m_k}   \bigg (\E_2|\sum_{\nu^{(1)}}^{(*)} \cdots|\bigg )^2 \cdots \bigg (\E_2|\sum_{\nu^{(l)}}^{(*)}\cdots|\bigg )^2     \right]^{\frac12}. 
  \end{align}
  Finally, by the induction assumptions, we have 
  \[\left (\E_2 |\sum_{\nu^{(i)}}^{(*)}\cdots|  \right)^2\leq \sum_{\nu^{(i)}}|a^{(\cdots)}_{m_{\cdot},n_{\cdot}}\cdots|^2\]
  and thus,
  \begin{align*}
    \eqref{3.8} &\lesssim_s \sum_{I_1,\cdots,I_k} \left[  \sum_{m_1,\cdots,m_k}   \sum_{\nu^{(1)}}\cdots\sum_{\nu^{(l)}} |a^{(0)}_{n,n_1}a^{(1)}_{n_1,n_2}\cdots a^{(s)}_{n_s,n'}|^2 \right]^{\frac{1}{2}} \\
    &\lesssim_s  \left[\sum_{n_1,\cdots,n_s}|a^{(0)}_{n,n_1}a^{(1)}_{n_1,n_2}\cdots a^{(s)}_{n_s,n'}|^2 \right]^{\frac12},
  \end{align*}
where   for the third  inequality,   we enlarge the summation by recalling  $a^{(j)}_{m,n}\geq 0$. 
\end{proof}

\section{The 6th-order  renormalization}\label{6thsec}

In this section, we will introduce  a 6th-order renormalization result via iterating the resolvent identity. Previously, Bourgain \cite{Bou03} performed a 4th-order renormalization, which allowed him to construct extended states provided $\alpha>\frac13$.   

\subsection{The $4$th-order renormalization of Bourgain}\label{Bousub}

For convenience, we use the  notation  from \cite{Bou03}.  We first  recall  the 4th-order  renormalization result of Bourgain \cite{Bou03}. We have  
\[V(n)=V_{\omega}(n)=v_n\omega_n,v_n=\kappa |n|^{-\alpha},\]
\[\sigma= G_0(0,0), \rho=2\sigma^3-\hat{K}(0),\hat{K}(\xi)=\hat{G_0}* \hat{G_0}* \hat{G_0}(\xi),\]
\[\widetilde{G_0}(n,n')=G_0(n,n')-\sigma \delta_{n,n'},\]
where $G_0$ is the Green's function of $-\Delta.$
Define further 
\[M_4(n_1,n_2) = \widetilde{G_0}(n_1,n_2)^3, W_4= v^2 M_4 v^2,\]
\[M= M_4 -(\sigma^3-\rho) , W=v^2 M v^2=W_4-(\sigma^3-\rho) v^4, \]
where $W$ arises  from the 4-tuples $(n_1,n_2,n_1,n_2),n_1\neq n_2\in\Z^d$.
\begin{figure}[htp]
  \centering

  \tikzset{every picture/.style={line width=0.75pt}} 

\begin{tikzpicture}[x=0.75pt,y=0.75pt,yscale=-0.7,xscale=0.7]

\draw  [fill={rgb, 255:black, 208; green, 2; blue, 27 }  ,fill opacity=1 ] (171,167) .. controls (171,164.32) and (173.17,162.14) .. (175.86,162.14) .. controls (178.54,162.14) and (180.71,164.32) .. (180.71,167) .. controls (180.71,169.68) and (178.54,171.86) .. (175.86,171.86) .. controls (173.17,171.86) and (171,169.68) .. (171,167) -- cycle ;
\draw  [fill={rgb, 255:black, 208; green, 2; blue, 27 }  ,fill opacity=1 ] (289,165.86) .. controls (289,163.17) and (291.17,161) .. (293.86,161) .. controls (296.54,161) and (298.71,163.17) .. (298.71,165.86) .. controls (298.71,168.54) and (296.54,170.71) .. (293.86,170.71) .. controls (291.17,170.71) and (289,168.54) .. (289,165.86) -- cycle ;
\draw  [fill={rgb, 255:black, 208; green, 2; blue, 27 }  ,fill opacity=1 ] (407,165.86) .. controls (407,163.17) and (409.17,161) .. (411.86,161) .. controls (414.54,161) and (416.71,163.17) .. (416.71,165.86) .. controls (416.71,168.54) and (414.54,170.71) .. (411.86,170.71) .. controls (409.17,170.71) and (407,168.54) .. (407,165.86) -- cycle ;
\draw  [fill={rgb, 255:black, 208; green, 2; blue, 27 }  ,fill opacity=1 ] (525,164.71) .. controls (525,162.03) and (527.17,159.86) .. (529.86,159.86) .. controls (532.54,159.86) and (534.71,162.03) .. (534.71,164.71) .. controls (534.71,167.4) and (532.54,169.57) .. (529.86,169.57) .. controls (527.17,169.57) and (525,167.4) .. (525,164.71) -- cycle ;
\draw    (175.86,167) .. controls (232,110.36) and (344.86,95.36) .. (411.86,165.86) ;
\draw    (293.86,165.86) .. controls (350,109.21) and (462.86,94.21) .. (529.86,164.71) ;

\draw (171,189) node [anchor=north west][inner sep=0.75pt]   [align=left] {$n_1$};
\draw (287,188) node [anchor=north west][inner sep=0.75pt]   [align=left] {$n_2$};
\draw (407,189) node [anchor=north west][inner sep=0.75pt]   [align=left] {$n_1$};
\draw (524,188) node [anchor=north west][inner sep=0.75pt]   [align=left] {$n_2$};
\end{tikzpicture}
\end{figure}\\ 
 The symbol of $M_4$ is   \[\hat{M_4}(\xi)=(\hat{G_0}-\sigma)* (\hat{G_0}-\sigma)*(\hat{G_0}-\sigma)(\xi)=\hat{K}(\xi)-\sigma^3.\]
Also, we have the  diagonal operator
\[D_4(n_1)=v^2_{n_1}\bigg[ \sum_{n_2\in\Z^d}v^2_{n_2} \wtg(n_1,n_2)^4\bigg],\]
where    $VD_4 (n_1)$ arises from the 5-tuples $(n_1,n_2,n_1,n_2,n_1),n_1\neq n_2\in\Z^d.$ 
\begin{figure}[htp]
  \centering

\tikzset{every picture/.style={line width=0.75pt}} 

\begin{tikzpicture}[x=0.75pt,y=0.75pt,yscale=-0.7,xscale=0.7]

\draw  [fill={rgb, 255:black, 208; green, 2; blue, 27 }  ,fill opacity=1 ] (112,157) .. controls (112,154.32) and (114.17,152.14) .. (116.86,152.14) .. controls (119.54,152.14) and (121.71,154.32) .. (121.71,157) .. controls (121.71,159.68) and (119.54,161.86) .. (116.86,161.86) .. controls (114.17,161.86) and (112,159.68) .. (112,157) -- cycle ;
\draw  [fill={rgb, 255:black, 208; green, 2; blue, 27 }  ,fill opacity=1 ] (230,155.86) .. controls (230,153.17) and (232.17,151) .. (234.86,151) .. controls (237.54,151) and (239.71,153.17) .. (239.71,155.86) .. controls (239.71,158.54) and (237.54,160.71) .. (234.86,160.71) .. controls (232.17,160.71) and (230,158.54) .. (230,155.86) -- cycle ;
\draw  [fill={rgb, 255:black, 208; green, 2; blue, 27 }  ,fill opacity=1 ] (348,155.86) .. controls (348,153.17) and (350.17,151) .. (352.86,151) .. controls (355.54,151) and (357.71,153.17) .. (357.71,155.86) .. controls (357.71,158.54) and (355.54,160.71) .. (352.86,160.71) .. controls (350.17,160.71) and (348,158.54) .. (348,155.86) -- cycle ;
\draw  [fill={rgb, 255:black, 208; green, 2; blue, 27 }  ,fill opacity=1 ] (466,154.71) .. controls (466,152.03) and (468.17,149.86) .. (470.86,149.86) .. controls (473.54,149.86) and (475.71,152.03) .. (475.71,154.71) .. controls (475.71,157.4) and (473.54,159.57) .. (470.86,159.57) .. controls (468.17,159.57) and (466,157.4) .. (466,154.71) -- cycle ;
\draw    (116.86,157) .. controls (173,100.36) and (285.86,85.36) .. (352.86,155.86) ;
\draw    (234.86,155.86) .. controls (267.67,122.75) and (319.87,103.87) .. (370.95,107.97) .. controls (407.26,110.89) and (443.02,125.42) .. (470.86,154.71) ;
\draw    (352.86,155.86) .. controls (409,99.21) and (521.86,84.21) .. (588.86,154.71) ;
\draw  [fill={rgb, 255:black, 208; green, 2; blue, 27 }  ,fill opacity=1 ] (584,154.71) .. controls (584,152.03) and (586.17,149.86) .. (588.86,149.86) .. controls (591.54,149.86) and (593.71,152.03) .. (593.71,154.71) .. controls (593.71,157.4) and (591.54,159.57) .. (588.86,159.57) .. controls (586.17,159.57) and (584,157.4) .. (584,154.71) -- cycle ;

\draw (112,179) node [anchor=north west][inner sep=0.75pt]   [align=left] {$n_1$};
\draw (228,178) node [anchor=north west][inner sep=0.75pt]   [align=left] {$n_2$};
\draw (348,179) node [anchor=north west][inner sep=0.75pt]   [align=left] {$n_1$};
\draw (465,178) node [anchor=north west][inner sep=0.75pt]   [align=left] {$n_2$};
\draw (584,178) node [anchor=north west][inner sep=0.75pt]   [align=left] {$n_1$};
\end{tikzpicture}
\end{figure}\\
 Recalling Lemma \ref{Khintchine}, we use the notation
$$(A_0V_\omega A_1V_\omega\cdots A_s)^{(*)}$$
to indicate that, when writing out the matrix product as a sum over multi-indices, we do restrict the sum to the {\it admissible}  multi-indices  generated by    $\omega=\{\omega_n\}_{n\in\Z^d}$.  Define  the renormalized potentials
\begin{align}\label{V4}
V^{(0)}_{\omega}=V_{\omega},\ V^{(2)}_{\omega}=V_{\omega}+\sigma v^2,\ V^{(4)}_{\omega}=V_{\omega}+\sigma v^2-\rho v^4,
\end{align}
and the corresponding renormalized random Schr\"odinger operator 
$$H^{(4)}=-\Delta+V_\omega^{(4)}\delta_{n,n'}.$$
Denote by $G$ the Green's function of  $H^{(4)}$, namely,  $G=(H^{(4)})^{-1}.$
Below,  we hide  the dependence of  potentials on $\omega$  for simplicity.  Moreover, we label the terms which have no randomness with  a box, i.e.,  $\boxed{{\rm TERM}}$.  By ``order'' of remaining terms we mean that  in $v$. Then 
iterating the resolvent identity 
$$G=G_0-GV^{(4)}G_0$$
and taking account of  cancellations in the expansion lead to 
\begin{align*}
G&=\mathcal{R}_5+G X_6,
\end{align*}
where $X_6$ denotes the $6$th-order remaining terms (are all ``admissible'') and 
\begin{align}\label{R5}
\mathcal{R}_5=&G_0-G_0VG_0+(G_0VG_0VG_0)^{(*)} \\
\nonumber&+\sigma^2 G_0v^2 V G_0-(G_0VG_0VG_0VG_0)^{(*)}\\
\nonumber&  -\sigma^2(G_0v^2VG_0 VG_0)^{(*)}-\sigma^2 (G_0 V G_0 v^2V G_0)^{\admissible} 
 \nonumber+(G_0 VG_0 VG_0 VG_0 VG_0)^{\admissible} \\
\nonumber& + \boxed{G_0 W G_0}\\
 \nonumber& +2\sigma \rho G_0 v^4VG_0 +G_0 V D_4 G_0  -G_0 V \widetilde{G_0} W G_0-G_0 W\widetilde{G_0} V G_0  +\sigma^2 (G_0 v^2V G_0 V G_0 VG_0)^{\admissible}\\
\nonumber &  +\sigma^2 (G_0 V G_0 v^2 V G_0 V G_0)^{\admissible}  +\sigma^2 (G_0 V G_0 V G_0 v^2V G_0)^{(*)} - (G_0 V G_0 V G_0 V G_0 V G_0 V G_0)^{\admissible}. 
\end{align}
The above  5th-order remaining terms are mainly  obtained from \cite{Bou03} (cf. (5.5)--(5.6)). Here we only  make  additional simplifications  for those terms.

\subsection{The 6th-order renormalization}\label{8thsub}

In the following, we aim to perform further $6$th-order renormalization through the renormalized potential given by 
\begin{align}\label{6-order renormalization of potential}
 \nonumber V^{(6)}_{\omega}&= V_{\omega} +\sigma v^2-\rho v^4 +(4\eta-3\sigma^5+5\sigma^2 \rho)v^6+ R_6,\\
  &=V_\omega^{(4)}+(4\eta-3\sigma^5+5\sigma^2 \rho)v^6+ R_6,
\end{align}
where $V_\omega^{(4)}$ is given by \eqref{V4} and 
\begin{align*}
\eta&=(\hat{G_0}-\sigma)*\big((\hat{G_0}-\sigma)*(\hat{G_0}-\sigma)\big)^2(0),\\
R_6 (n_1)&= v^2_{n_1}\cdot \big(\widetilde{G_0} W \widetilde{G_0}  \big)(n_1,n_1),
\end{align*}
arising  from the 6-tuples $(n_1,n_2,n_3,n_2,n_3,n_1)$ with $n_1\neq n_2\neq n_3\in\Z^d.$ 
\begin{figure}[htbp]
  \centering

\tikzset{every picture/.style={line width=0.75pt}} 

\begin{tikzpicture}[x=0.75pt,y=0.75pt,yscale=-0.5,xscale=0.5]

\draw  [fill={rgb, 255:black, 208; green, 2; blue, 27 }  ,fill opacity=1 ] (40,156) .. controls (40,153.32) and (42.17,151.14) .. (44.86,151.14) .. controls (47.54,151.14) and (49.71,153.32) .. (49.71,156) .. controls (49.71,158.68) and (47.54,160.86) .. (44.86,160.86) .. controls (42.17,160.86) and (40,158.68) .. (40,156) -- cycle ;
\draw  [fill={rgb, 255:black, 208; green, 2; blue, 27 }  ,fill opacity=1 ] (158,154.86) .. controls (158,152.17) and (160.17,150) .. (162.86,150) .. controls (165.54,150) and (167.71,152.17) .. (167.71,154.86) .. controls (167.71,157.54) and (165.54,159.71) .. (162.86,159.71) .. controls (160.17,159.71) and (158,157.54) .. (158,154.86) -- cycle ;
\draw  [fill={rgb, 255:black, 208; green, 2; blue, 27 }  ,fill opacity=1 ] (276,154.86) .. controls (276,152.17) and (278.17,150) .. (280.86,150) .. controls (283.54,150) and (285.71,152.17) .. (285.71,154.86) .. controls (285.71,157.54) and (283.54,159.71) .. (280.86,159.71) .. controls (278.17,159.71) and (276,157.54) .. (276,154.86) -- cycle ;
\draw  [fill={rgb, 255:black, 208; green, 2; blue, 27 }  ,fill opacity=1 ] (394,153.71) .. controls (394,151.03) and (396.17,148.86) .. (398.86,148.86) .. controls (401.54,148.86) and (403.71,151.03) .. (403.71,153.71) .. controls (403.71,156.4) and (401.54,158.57) .. (398.86,158.57) .. controls (396.17,158.57) and (394,156.4) .. (394,153.71) -- cycle ;
\draw    (162.86,154.86) .. controls (195.67,121.75) and (247.87,102.87) .. (298.95,106.97) .. controls (335.26,109.89) and (371.02,124.42) .. (398.86,153.71) ;
\draw    (280.86,154.86) .. controls (337,98.21) and (449.86,83.21) .. (516.86,153.71) ;
\draw  [fill={rgb, 255:black, 208; green, 2; blue, 27 }  ,fill opacity=1 ] (512,153.71) .. controls (512,151.03) and (514.17,148.86) .. (516.86,148.86) .. controls (519.54,148.86) and (521.71,151.03) .. (521.71,153.71) .. controls (521.71,156.4) and (519.54,158.57) .. (516.86,158.57) .. controls (514.17,158.57) and (512,156.4) .. (512,153.71) -- cycle ;
\draw  [fill={rgb, 255:black, 208; green, 2; blue, 27 }  ,fill opacity=1 ] (625.14,152.57) .. controls (625.14,149.89) and (627.32,147.71) .. (630,147.71) .. controls (632.68,147.71) and (634.86,149.89) .. (634.86,152.57) .. controls (634.86,155.25) and (632.68,157.43) .. (630,157.43) .. controls (627.32,157.43) and (625.14,155.25) .. (625.14,152.57) -- cycle ;
\draw    (44.86,156) .. controls (185,23.36) and (510,20.36) .. (630,152.57) ;

\draw (40,178) node [anchor=north west][inner sep=0.75pt]   [align=left] {$n_1$};
\draw (156,177) node [anchor=north west][inner sep=0.75pt]   [align=left] {$n_2$};
\draw (276,178) node [anchor=north west][inner sep=0.75pt]   [align=left] {$n_3$};
\draw (393,177) node [anchor=north west][inner sep=0.75pt]   [align=left] {$n_2$};
\draw (512,177) node [anchor=north west][inner sep=0.75pt]   [align=left] {$n_3$};
\draw (625,175) node [anchor=north west][inner sep=0.75pt]   [align=left] {$n_1$};

\end{tikzpicture}

\end{figure}\\ 
Indeed, if  
\[H=-\Delta+\widetilde{V},\ G=H^{-1},\]
then we obtain by  iterating the resolvent identity
\begin{equation}\label{resolvent identity}
  G=G_0-G\widetilde{V}G_0
\end{equation}
that 
\begin{align}\label{8 order expansion}
   G=&G_0 -G_0\widetilde{V}G_0+G_0\widetilde{V}G_0\widetilde{V}G_0-G_0\widetilde{V}G_0\widetilde{V}G_0\widetilde{V}G_0+G_0\widetilde{V}G_0\widetilde{V}G_0\widetilde{V}G_0\widetilde{V}G_0\\
   \notag &-G_0\widetilde{V}G_0\widetilde{V}G_0\widetilde{V}G_0\widetilde{V}G_0\widetilde{V}G_0+G_0\widetilde{V}G_0\widetilde{V}G_0\widetilde{V}G_0\widetilde{V}G_0\widetilde{V}G_0\widetilde{V}G_0\\
   \notag &-G_0\widetilde{V}G_0\widetilde{V}G_0\widetilde{V}G_0\widetilde{V}G_0\widetilde{V}G_0\widetilde{V}G_0\widetilde{V}G_0+G\widetilde{V}G_0\widetilde{V}G_0\widetilde{V}G_0\widetilde{V}G_0\widetilde{V}G_0\widetilde{V}G_0\widetilde{V}G_0\widetilde{V}G_0. 
\end{align}

Before presenting our main theorem in this section, we first introduce some notation and computations on $r$th-order ($r\geq 8$) remaining terms in the 6th-order renormalization scheme. We will repeatedly use this argument to do some rearrangements, which will play an important role in  both  Green's function estimates and the construction of extended states in the rest of the paper.

Let $H=-\Delta+V_\omega^{(6)}$ with $V_\omega^{(6)}$ given by \eqref{6-order renormalization of potential}, and let $G=H^{-1}$.  Denote 
$$\Delta_{2k} V= V^{(2k)}_{\omega}-V^{(2k-2)}_{\omega},$$
which  is exactly the $2k$th  order renormalized  potential, where $V_\omega^{(2k)}$ ($0\leq k\leq 3$) are defined  by \eqref{V4} and \eqref{6-order renormalization of potential}. 
So we get 
\begin{equation*}
\widetilde{V}=V^{(6)}_{\omega}=V+\Delta_2 V +\Delta_4 V +\Delta_6 V. 
\end{equation*}
From now on, we use  the following notation:  denote by  $\boxed{G_0,i},0\leq i \leq 7$  the exactly  $i$th  order remaining terms , and by  $\boxed{G,i},0\leq i \leq 7$  the  terms  with the first $G_0$ in $\boxed{G_0,i}$  replaced by $G$.  For example, we have 
\begin{align*}
&\boxed{G_0,2}=(G_0V G_0 VG_0)^{\admissible},\ \boxed{G,2}=(G VG_0 VG_0)^{\admissible},\\
&\boxed{G_0,3}=\sigma^2 G_0v^2 V G_0-(G_0VG_0VG_0VG_0)^{(*)},\\
&\boxed{G,3}=\sigma^2 Gv^2 V G_0-(GVG_0VG_0VG_0)^{(*)}.
\end{align*}
 We can write  
  \begin{align}\label{D.6}
    G-G_0=\sum_{i=1}^{7}\boxed{G_0,i}+(r{\rm th \ order \ remaining \ terms})\ (r\geq 8).
  \end{align}
From now on, we label the  $r$th-order terms with  $r\geq 8$ by a \uwave{$\qquad$}.  We begin with an important lemma.  
\begin{lem}\label{8thlem}
  For $2\leq i\leq 7$, we have 
  \begin{align}\label{D.2}
    \boxed{G_0,i} =-G_0 V \boxed{G_0,i-1}-G_0 \Delta_2 V \boxed{G_0,i-2}-\cdots -G_0 \Delta_{\lfloor i\rfloor_{e}} V \boxed{G_0,i-\lfloor i\rfloor_{e}},
  \end{align}
  where $\lfloor i\rfloor_e$ denotes  the biggest even number less than $i$. Similarly, 
  \begin{align}\label{D.3}
    \boxed{G,i} =-G V \boxed{G_0,i-1}-G \Delta_2 V \boxed{G_0,i-2}-\cdots -G \Delta_{\lfloor i\rfloor_{e}} V \boxed{G_0,i-\lfloor i\rfloor_{e}}. 
  \end{align}
Moreover, we have 
\begin{align}\label{8thterm}
  G
      =&\sum_{i=0}^{7}\boxed{G_0,i}\\
     \nonumber &-\uwave{\text{$G V \boxed{G_0,7}$}} -\uwave{\text{$G\Delta_2 V\cdot\sum_{i=6}^{7}\boxed{G_0,i}$}}-\uwave{\text{$G\Delta_4 V\cdot\sum_{i=4}^{7}\boxed{G_0,i}$}}-\uwave{\text{$G\Delta_6 V\cdot\sum_{i=2}^{7}\boxed{G_0,i}$}}. 
\end{align}

\end{lem}

\begin{proof}
  When $i=2$, \eqref{D.3} can be verified directly.
 By the resolvent identity, for $G=(-\Delta+\widetilde{V})^{-1}$ ($\widetilde{V}=V^{(6)}$), we have 
  \begin{align*}
    G=G_0-G_0\widetilde{V}G=G_0-G_0\widetilde{V}G_0+G_0\widetilde{V}(G_0-G).
  \end{align*}
If  $i\geq 3$, the  $i$th-order terms  can only be  generated by $G_0\widetilde{V}(G_0-G)$.  From  
  \begin{align}\label{D.5}
    G_0\widetilde{V}(G_0-G) =&-G_0 V (G-G_0)-G_0 \Delta_2 V(G-G_0) \\
        \notag         & -G_0\Delta_4 V (G-G_0)-G_0 \Delta_6 V (G-G_0)
  \end{align}
   and by  substituting  \eqref{D.6} into \eqref{D.5} to  extracting  the $i$th order terms, it follows that \eqref{D.2} holds true. Then, replacing  the first $G_0$ in all terms in \eqref{D.2} implies  \eqref{D.3}.

Next, the resolvent identity also has the form of  $G=G_0-G\widetilde{V}G_0$. This implies  
\begin{align}\label{D.7}
  \boxed{G,i}&=\boxed{G_0,i}-G\widetilde{V}\boxed{G_0,i}\\
  \notag &=\boxed{G_0,i}-G V \boxed{G_0,i}-G\Delta_2 V \boxed{G_0,i}-G\Delta_4 V \boxed{G_0,i}-G\Delta_6 V \boxed{G_0,i}. 
\end{align}
Note that $\boxed{G_0,0}=G_0,\boxed{G,1}=-GVG_0$.  Using the resolvent identity yields 
\begin{align*}
  G&=G_0-G\widetilde{V}G_0\\
   &=G_0-GVG_0-G\Delta_2 V G_0-G\Delta_4 V G_0-G\Delta_6 VG_0\\
   &=\boxed{G_0,0}+\boxed{G,1}-G\Delta_2 V\boxed{G_0,0}-G\Delta_4 V\boxed{G_0,0}-G\Delta_6 V\boxed{G_0,0}\\
   &\overset{{\rm by}\ \eqref{D.7}}{=}
   \boxed{G_0,0}+\boxed{G_0,1}+(-G V\boxed{G_0,1}-G\Delta_2 V\boxed{G_0,0})\\
   &\qquad \qquad  -G\Delta_2 V\boxed{G_0,1}-G\Delta_4 V\boxed{G_0,1}-G\Delta_6 V\boxed{G_0,1}\\
   &\qquad \qquad  -G\Delta_4 V\boxed{G_0,0}-G\Delta_6 V\boxed{G_0,0}\\
   &\overset{{\rm by}\ \eqref{D.3}}{=}\boxed{G_0,0}+\boxed{G_0,1}+\boxed{G,2}\\
   &\qquad \qquad  -G\Delta_2 V\boxed{G_0,1}-G\Delta_4 V\boxed{G_0,1}-G\Delta_6 V\boxed{G_0,1}\\
   &\qquad \qquad  -G\Delta_4 V\boxed{G_0,0}-G\Delta_6 V\boxed{G_0,0}\\
   &\overset{{\rm by}\ \eqref{D.7}}{=}
   \boxed{G_0,0}+\boxed{G_0,1}+\boxed{G_0,2}\\
   &\qquad \qquad +(-G V\boxed{G_0,2}-G\Delta_2 V\boxed{G_0,1}) \\
   &\qquad \qquad  -G\Delta_2 V\boxed{G_0,2}-G\Delta_4 V\boxed{G_0,2}-\uwave{\text{$G\Delta_6 V\boxed{G_0,2}$}}\\
   &\qquad \qquad  -G\Delta_4 V\boxed{G_0,1}-G\Delta_6 V\boxed{G_0,1}\\
   &\qquad \qquad  -G\Delta_4 V\boxed{G_0,0}-G\Delta_6 V\boxed{G_0,0}\\
   &\overset{{\rm by}\ \eqref{D.3}}{=}\boxed{G_0,0}+\boxed{G_0,1}+\boxed{G_0,2}+\boxed{G,3}\\
   &\qquad \qquad  -G\Delta_2 V\boxed{G_0,2}-G\Delta_4 V\boxed{G_0,2}-\uwave{\text{$G\Delta_6 V\boxed{G_0,2}$}}\\
   &\qquad \qquad  -G\Delta_4 V\boxed{G_0,1}-G\Delta_6 V\boxed{G_0,1}\\
   &\qquad \qquad  -G\Delta_4 V\boxed{G_0,0}-G\Delta_6 V\boxed{G_0,0}\\
   &\overset{{\rm by}\ \eqref{D.7}}{=}\boxed{G_0,0}+\boxed{G_0,1}+\boxed{G_0,2}+\boxed{G_0,3}\\
   &\qquad -G V\boxed{G_0,3} -G\Delta_2 V\boxed{G_0,3}-\uwave{\text{$G\Delta_4 V\boxed{G_0,3}$}}-\uwave{\text{$G\Delta_6 V\boxed{G_0,3}$}}\\
   &\qquad \qquad  -G\Delta_2 V\boxed{G_0,2}-G\Delta_4 V\boxed{G_0,2}-\uwave{\text{$G\Delta_6 V\boxed{G_0,2}$}}\\
   &\qquad \qquad  -G\Delta_4 V\boxed{G_0,1}-G\Delta_6 V\boxed{G_0,1}\\
   &\qquad \qquad  -G\Delta_4 V\boxed{G_0,0}-G\Delta_6 V\boxed{G_0,0}\\
   &\overset{{\rm via\ reusing \ \eqref{D.3}\ and \ \eqref{D.7}}}{\cdots}\\
   &=\sum_{i=0}^{7}\boxed{G_0,i}-\uwave{\text{$G V \boxed{G_0,7}$}} -\uwave{\text{$G\Delta_2 V\cdot\sum_{i=6}^{7}\boxed{G_0,i}$}}-\uwave{\text{$G\Delta_4 V\cdot\sum_{i=4}^{7}\boxed{G_0,i}$}}-\uwave{\text{$G\Delta_6 V\cdot\sum_{i=2}^{7}\boxed{G_0,i}$}}.
\end{align*}
This proves  \eqref{8thterm}.
\end{proof}
 
Our main theorem in this section is 
\begin{thm}\label{6ththm}
Let $H=-\Delta+V_\omega^{(6)}$ with $V_\omega^{(6)}$ given by \eqref{6-order renormalization of potential}, and let $G=H^{-1}$. Then 
\begin{align}
G&=\mathcal{R}_5+\boxed{G_0,6}+\boxed{G_0,7}+GB\\
\label{resolvent}&:= A+GB,
\end{align}
where $\mathcal R_5$ is given by \eqref{R5}  and
\begin{itemize}
  \item[(i)]  {\bf (The 6th-order remaining terms)}
      \begin{align}
        \notag \boxed{G_0,6}=& \\
       \label{4.4}    &2\sigma\rho \big((G_0 v^4 VG_0 V G_0)^{\admissible} + (G_0  VG_0  v^4 V G_0)^{\admissible}  \big)      +\sigma^4 (G_0 v^2 VG_0 v^2 VG_0)^{\admissible} \\
        \label{4.5}   &- \sigma^2\big(     (G_0 v^2 V G_0 V G_0 V G_0 VG_0)^{\admissible} + (G_0 V G_0 v^2 V G_0 V G_0 VG_0)^{\admissible} \\
       \notag    & + (G_0 V G_0 V G_0 v^2 V G_0 VG_0)^{\admissible}        +  (G_0 V G_0 V G_0 V G_0 v^2 VG_0)^{\admissible}        \big)\\
         \label{4.6}  &+ (G_0 V G_0 V G_0 VG_0 VG_0 V G_0 VG_0)^{\admissible} \\
          \label{4.7} &+\big(   (G_0 W \widetilde{G_0} V G_0 V G_0)^{\admissible}  +(G_0 V \wtg W \wtg V G_0)^{\admissible} +(G_0 V G_0 V \wtg W G_0)^{\admissible}      \big)\\
          \label{4.8} &-\big( (G_0 VD_4 G_0 V G_0)^{\admissible}+(G_0 V G_0 VD_4 G_0)^{\admissible}   \big)\\
        \label{4.9}   &+ 4\boxed{ G_0 C G_0}-2\sigma^2(\boxed{G_0 v^2 W G_0}+\boxed{G_0 W v^2 G_0}),
      \end{align}
      where the new type of non-random operator in the above representation is 
      \begin{align}  \label{nonconvolution term}
        C(n_1,n_3)=v^2_{n_1}v^2_{n_3}\wtg(n_1,n_3) \sum_{n_2}\wtg(n_1,n_2)^2 v^2_{n_2} \wtg(n_2,n_3)^2-\eta\delta_{n_1,n_3} v^6_{n_1}
      \end{align}  
      arising from the 6-tuples,  such as $(n_1,n_2,n_3,n_1,n_2,n_3),n_1\neq n_2\neq n_3\in\Z^d$ (there are also   tuples of other forms producing \eqref{nonconvolution term}, of which the details  can be found in the Appendix \ref{App6th}).\\
      \begin{figure}[htbp]
        \centering

\tikzset{every picture/.style={line width=0.75pt}} 

\begin{tikzpicture}[x=0.75pt,y=0.75pt,yscale=-0.7,xscale=0.7]

\draw  [fill={rgb, 255:black, 208; green, 2; blue, 27 }  ,fill opacity=1 ] (57.89,151.22) .. controls (57.89,148.41) and (59.82,146.14) .. (62.21,146.14) .. controls (64.59,146.14) and (66.52,148.41) .. (66.52,151.22) .. controls (66.52,154.02) and (64.59,156.29) .. (62.21,156.29) .. controls (59.82,156.29) and (57.89,154.02) .. (57.89,151.22) -- cycle ;
\draw  [fill={rgb, 255:black, 208; green, 2; blue, 27 }  ,fill opacity=1 ] (162.79,150.02) .. controls (162.79,147.22) and (164.72,144.95) .. (167.11,144.95) .. controls (169.49,144.95) and (171.42,147.22) .. (171.42,150.02) .. controls (171.42,152.82) and (169.49,155.09) .. (167.11,155.09) .. controls (164.72,155.09) and (162.79,152.82) .. (162.79,150.02) -- cycle ;
\draw  [fill={rgb, 255:black, 208; green, 2; blue, 27 }  ,fill opacity=1 ] (267.69,150.17) .. controls (267.69,147.37) and (269.62,145.1) .. (272,145.1) .. controls (274.39,145.1) and (276.32,147.37) .. (276.32,150.17) .. controls (276.32,152.97) and (274.39,155.24) .. (272,155.24) .. controls (269.62,155.24) and (267.69,152.97) .. (267.69,150.17) -- cycle ;
\draw  [fill={rgb, 255:black, 208; green, 2; blue, 27 }  ,fill opacity=1 ] (372.58,148.83) .. controls (372.58,146.03) and (374.52,143.76) .. (376.9,143.76) .. controls (379.29,143.76) and (381.22,146.03) .. (381.22,148.83) .. controls (381.22,151.63) and (379.29,153.9) .. (376.9,153.9) .. controls (374.52,153.9) and (372.58,151.63) .. (372.58,148.83) -- cycle ;
\draw  [fill={rgb, 255:black, 208; green, 2; blue, 27 }  ,fill opacity=1 ] (477.48,148.83) .. controls (477.48,146.03) and (479.42,143.76) .. (481.8,143.76) .. controls (484.18,143.76) and (486.12,146.03) .. (486.12,148.83) .. controls (486.12,151.63) and (484.18,153.9) .. (481.8,153.9) .. controls (479.42,153.9) and (477.48,151.63) .. (477.48,148.83) -- cycle ;
\draw  [fill={rgb, 255:black, 208; green, 2; blue, 27 }  ,fill opacity=1 ] (582.38,147.78) .. controls (582.38,144.98) and (584.31,142.71) .. (586.7,142.71) .. controls (589.08,142.71) and (591.02,144.98) .. (591.02,147.78) .. controls (591.02,150.59) and (589.08,152.86) .. (586.7,152.86) .. controls (584.31,152.86) and (582.38,150.59) .. (582.38,147.78) -- cycle ;
\draw    (62.21,151.22) .. controls (122.71,87.29) and (313.71,83.29) .. (376.9,148.83) ;
\draw    (167.11,151.22) .. controls (227.61,87.29) and (418.61,83.29) .. (481.8,148.83) ;
\draw    (272,150.17) .. controls (332.51,86.24) and (523.51,82.24) .. (586.7,147.78) ;

\draw (57,175) node [anchor=north west][inner sep=0.75pt]   [align=left] {$n_1$};
\draw (163,173) node [anchor=north west][inner sep=0.75pt]   [align=left] {$n_2$};
\draw (268,174) node [anchor=north west][inner sep=0.75pt]   [align=left] {$n_3$};
\draw (373,172) node [anchor=north west][inner sep=0.75pt]   [align=left] {$n_1$};
\draw (476,171) node [anchor=north west][inner sep=0.75pt]   [align=left] {$n_2$};
\draw (582,171) node [anchor=north west][inner sep=0.75pt]   [align=left] {$n_3$};

\end{tikzpicture}

      \end{figure}

  \item[(ii)] {\bf (The 7th-order remaining terms)}  
    \begin{align} 
      \notag \boxed{G_0,7}=& \\
     \label{4.11} &2\sigma G_0 V R_6 G_0 +(8\eta\sigma -7\sigma^6+12\sigma^3\rho)G_0 v^6 V G_0\\
      \label{4.12} &-2\sigma\rho \big( (G_0 v^4 VG_0 VG_0 VG_0)^{\admissible}+ (G_0 VG_0 v^4 VG_0 VG_0)^{\admissible} \\
                \notag   &+(G_0  VG_0 VG_0 v^4 VG_0)^{\admissible}       \big)\\
      \label{4.13} &-\sigma^4\big((G_0 v^2 VG_0 v^2 VG_0 VG_0)^{\admissible}+(G_0 v^2 VG_0 VG_0 v^2 VG_0)^{\admissible} \\
            \notag    &+(G_0  VG_0 v^2 VG_0 v^2 VG_0)^{\admissible}\big)\\
       \label{4.14}  &+\sigma^2 \big((G_0 v^2 VG_0 VG_0 VG_0 VG_0 VG_0)^{\admissible}+(G_0  VG_0 v^2 VG_0 VG_0 VG_0 VG_0)^{\admissible} \\
         \notag & +(G_0  VG_0 VG_0 v^2 VG_0 VG_0 VG_0)^{\admissible}+(G_0 VG_0 VG_0 VG_0 v^2 VG_0 VG_0)^{\admissible} \\
         \notag &  +(G_0  VG_0 VG_0 VG_0 VG_0 v^2 VG_0)^{\admissible}\big)  \\  
 \label{4.15} &+\sigma^2(G_0 W \wtg v^2V G_0 +G_0 v^2 V \wtg W G_0) \\
     \label{4.16} &+2\sigma^2 (G_0 W v^2 \wtg VG_0+ G_0 v^2 W \wtg V G_0 +G_0 V \wtg v^2 W G_0+ G_0 V \wtg W v^2 G_0) \\
     \label{4.17} &-3\sigma ^2 G_0 v^2 V D_4 G_0 -2\sigma^2 G_0 V D_6^{(1)} G_0 \\
     \label{4.18} &-(G_0 VG_0 VG_0 VG_0 VG_0 VG_0 VG_0 VG_0)^{\admissible}\\
     \label{4.19} &-\big((G_0 W \wtg V G_0 V G_0 VG_0)^{\admissible}+ (G_0 V \wtg W \wtg V G_0 VG_0)^{\admissible} \\
      \notag &\qquad + (G_0 V G_0 V \wtg W \wtg VG_0)^{\admissible}+(G_0 V G_0 V G_0 V \wtg WG_0)^{\admissible} \big)\\
     \label{4.20} &+((G_0 V D_4 G_0 V G_0V G_0)^{\admissible} +(G_0 V G_0 VD_4 G_0V G_0)^{\admissible}+(G_0 V G_0 V G_0 VD_4 G_0)^{\admissible})\\
     \label{4.21} &+(G_0 v^2 M_4 v^2 V M_4 v^2G_0)-G_0 D_7 G_0 \\
     \label{4.22} &-4(G_0C\wtg V G_0+G_0 V\wtg C G_0) \\
     \label{4.23} &+4 G_0 V D^{(2)}_6 G_0 \\
     \label{4.24} &+(G_0 V S G_0 +G_0 S^{\top} V G_0), 
    \end{align}
   where the new type of  operators (as compared with the $4$th-order renormalization of \cite{Bou03}) in the  above representation  are 
    \begin{equation*}
      D_6^{(1)}(n_1)=v^2_{n_1}\bigg[ \sum_{n_2\in\Z^d}v^4_{n_2} \wtg(n_1,n_2)^4\bigg]
    \end{equation*}
    arising from the 7-tuples $(n_1,n_2,n_2,n_2,n_1,n_2,n_1),(n_1,n_2,n_1,n_2,n_2,n_2,n_1),n_1\neq n_2\in \Z^d$,  and {(we emphasize that $D_7$ is a random diagonal operator)}
     \begin{equation*}
      D_7 = v^4_{n_1} \bigg[\sum_{n_2\in \Z^d} v^3_{n_2}\omega_{n_2} \wtg(n_1,n_2)^6  \bigg]
     \end{equation*}
    arising from the 7-tuples $(n_1,n_2,n_2,n_2,n_1,n_2,n_1),(n_1,n_2,n_1,n_2,n_1,n_2,n_1),n_1\neq n_2\in \Z^d$,  and 
    \begin{equation*}
      D_6^{(2)}(n_1)= v^2_{n_1}\sum_{n_2,n_3\in\Z^d}\bigg[  v^2_{n_2 } v^2_{n_3} \wtg(n_1,n_2)^2 \wtg(n_1,n_3)^2\wtg(n_2,n_3)^2  \bigg]
    \end{equation*}
    arising from the 7-tuples such as $(n_1,n_2,n_3,n_1,n_2,n_3,n_1),n_1\neq n_2\neq n_3\in \Z^d$ (there are also other tuples producing this term), and 
    \begin{equation*}
      S(n_1,n_3)=v^2_{n_1}v^2_{n_3}\wtg (n_1,n_3)^2 \sum_{n_2} \wtg(n_1,n_2)^3 v^2_{n_2} \wtg(n_2,n_3)
    \end{equation*}
    arising from the 7-tuple $(n_1,n_2,n_1,n_2,n_3,n_1,n_3),n_1\neq n_2\neq n_3\in \Z^d$, with $S^{\top}$ denoting  the transposed operator of $S$.
    
   \item[(iii)]{\bf (The 8th-order remaining terms)} 
    \begin{align}\label{B}
  &B= -\uwave{\text{$V \boxed{G_0,7}$}} -\uwave{\text{$\Delta_2 V\cdot\sum_{i=6}^{7}\boxed{G_0,i}$}}-\uwave{\text{$\Delta_4 V\cdot\sum_{i=4}^{7}\boxed{G_0,i}$}}-\uwave{\text{$\Delta_6 V\cdot\sum_{i=2}^{7}\boxed{G_0,i}$}}.
\end{align}

   \end{itemize}
\end{thm}

\begin{rmk}
 In \eqref{resolvent}, we  can  also rewrite 
   \begin{equation}\label{A}
    A=\sum_{i=0}^{7}\boxed{G_0,i},
\end{equation}
 which   is independent of $G.$
\end{rmk}

\begin{proof}[Proof of Theorem \ref{6ththm}]
The computations of $\boxed{G_0,6},\ \boxed{G_0,7}$ are based on certain graph representations, of which the details  can be found in Appendixes \ref{App6th} and \ref{App7th}. Once those computations were finished, the derivation  of $B$ just  follows directly from Lemma \ref{8thlem} (cf. \eqref{8thterm}). 
\end{proof}

\section{Green's function estimates: Proof of Theorem \ref{greenthm}}\label{thm1sec}
In this section, we aim to establish the Green's function estimates and complete the proof of Theorem \ref{greenthm}.  The proof relies on the scheme introduced by Bourgain \cite{Bou03} for dealing with the 4th-order renormalization: For the admissible terms in the expansion (cf. Theorem \ref{6ththm}), one can use  the hypercontractivity estimates (cf. Lemma \ref{decoupling}) to deduce the probabilistic bounds. For the non-random terms,  we can take advantage of both  the convolution regularization  argument and the symmetrical difference  trick to obtain the desired estimates. However, in the present 6th-order renormalization context, a new type of non-random operator emerges, which is not a symmetrical combination of diagonal operators and convolutional operators. To address this difficulty, we introduce an approach mainly based on a partial symmetrical difference reduction, rearrangement, and the application of the fractional Gagliardo-Nirenberg inequality \cite{BM18}. 

 From now on,  we assume $d\geq 5, \frac{1}{4}<\alpha\leq \frac{1}{3}$, since the case of $\alpha>\frac{1}{3}$ has been handled by Bourgain  \cite{Bou03}.  Recalling \eqref{resolvent}, it is necessary to control remaining terms of orders up to $8$.

\subsection{Estimates on $i$th-order  remaining terms for $i\leq 5$}
Now we begin with controlling lower orders  (i.e., less than $5$) remaining terms. We have 
\begin{thm}\label{0-5 remaining term expactation}
For $0\leq i\leq 5$ and $p\geq 1$, we have 
\begin{equation}\label{i order remaining expactation}
  \mathbb{E}_p\left |\boxed{G_0,i}(n,n')\right |\lesssim_{d,p,\alpha}\kappa^i \frac{1}{(|n|\wedge|n'|)^{i\alpha}|n-n'|^{d-2}}. 
\end{equation}
\end{thm}

\begin{proof}
It  suffices to  control each term in $\boxed{G_0,i}(n,n')$. 

  When $i=0$,  the  remaining term is just $G_0$ which  is deterministic. In this case, we have 
  \[\E_p |G_0(n,n')|=|G_0(n,n')|\lesssim \frac{1}{|n-n'|^{d-2}}.\]
 \\ 

If  $i=1,2,3$, due to $-G_0VG_0=-(G_0 VG_0)^{\admissible}$, we can apply the decoupling Lemma \ref{decoupling} to get the desired estimates.  For example,  for the 3th-order remaining term
$$(G_0VG_0 VG_0 VG_0)^{\admissible},$$
 we get by Lemma \ref{decoupling} that
  \begin{align*}
   &\ \ \  \E_p\left |  (G_0VG_0 VG_0 VG_0)^{\admissible}(n,n')  \right | \\
    &=\kappa^3 \E_p\left |  \sum_{(m_1,m_2,m_3)}^{\admissible}\omega_{m_1}\omega_{m_2}\omega_{m_3} G_0(n,m_1)v_{m_1}G_0(m_1,m_2)v_{m_2}G_0(m_2,m_3)v_{m_3}G_0(m_3,n')    \right |\\
    &\lesssim_{p} \kappa^3 \left( \sum_{m_1,m_2,m_3\in\Z^d}\frac{1}{|n-m_1|^{2(d-2)}|m_1|^{2\alpha}|m_1-m_2|^{2(d-2)}|m_2|^{2\alpha}\cdots|m_3-n'|^{2(d-2)}} \right)^{\frac{1}{2}}.
  \end{align*}
 Hence,  by $d\geq 5\Rightarrow 2(d-2)> d$, we can apply Lemma \ref{Lemma 2.2} to get 
  \begin{align*}
    &\ \ \ \E_p\left |  (G_0VG_0 VG_0 VG_0)^{\admissible}(n,n')  \right | \\
    &\lesssim_{d,p,\alpha} \kappa^3 \left( \sum_{m_2,m_3\in\Z^d}\frac{1}{|n-m_2|^{2(d-2)}(|m_2|\wedge|n|)^{2\alpha}|m_2|^{2\alpha}\cdots|m_3-n'|^{2(d-2)}} \right)^{\frac{1}{2}}\\
    &\lesssim_{d,p,\alpha} \kappa^3 \left( \sum_{m_2,m_3\in\Z^d}\frac{1}{|n-m_2|^{2(d-2)}} (\frac{1}{|m_2|^{2\alpha}}+\frac{1}{|n|^{2\alpha}}) \frac{1}{|m_2|^{2\alpha}\cdots|m_3-n'|^{2(d-2)}} \right)^{\frac{1}{2}}\\
    &\overset{{\rm \ applying \ repeatedly\  \ Lemma\ \ref{Lemma 2.2}}}{\cdots}\\
    &\lesssim_{d,p,\alpha} \kappa^3 \frac{1}{|n-n'|^{d-2}(|n|\wedge|n'|)^{3\alpha}}. 
  \end{align*}
  The other remaining terms of orders at most $3$ can be controlled similarly  (and can be easier to handle). 

 If  $i=4$, the admissible remaining terms (i.e., random terms) of exactly  4th order  can be controlled similarly  to those of orders less than $3$. However,  the deterministic term $G_0 W G_0$ needs to be controlled very carefully, since  we cannot use the decoupling lemma to gain the regularization, namely, the estimate from $|n-n'|^{-(d-2)}\to |n-n'|^{-2(d-2)}$. Instead of applying Lemma 1.2  in  \cite{Bou03}, we directly  estimate  this term via the  symmetrical difference  regularization and  convolution regularization (in the Fourier space)  arguments  originating from \cite{Bou03}.  We need the following lemma:
 
 \begin{lem}[Difference regularization]\label{difference property}
    For $\alpha>0$, we have 
    \[||n_1|^{-\alpha}-|n_2|^{-\alpha}|\lesssim_{\alpha} \frac{|n_1-n_2|}{(|n_1|+|n_2|)\cdot (|n_1|\wedge|n_2|)^{\alpha}}.\]
  \end{lem} 
  \begin{proof}
     We refer to  Appendix \ref{tecApp} for a detailed proof. 
     \end{proof}
     
Recall that $M$ is a convolution operator. Since 
  \[v^2_{n_1}v^2_{n_2}=\frac12(v^4_{n_1}+v^4_{n_2})-\frac{1}{2}(v^2_{n_1}-v^2_{n_2})^2,\]
  we have the decomposition
  \begin{equation}\label{Decompose G0WG0}
    G_0 W G_0=\frac{1}{2}(G_0v^4 MG_0+G_0 M v^4 G_0)-\frac{1}{2}G_0 P_4 G_0. 
  \end{equation}
Then by Lemma \ref{difference property},
  \begin{align}
    | P_4(n_1,n_2) | & =|(v^2_{n_1}-v^2_{n_2})^2 M(n_1,n_2)|\\
    \notag &\lesssim_{\alpha} \kappa^4 \frac{|n_1-n_2|^2}{(|n_1|+|n_2|)^2(|n_1|\wedge|n_2|)^{4\alpha}|n_1-n_2|^{3(d-2)}}. 
  \end{align}
  Moreover, by the convolution regularization  argument in  \cite[(1.7)--(1.9)]{Bou03} and \cite[(3.9)--(3.10)]{Bou03}, we have 
  \[|G_0 M(n_1,n_2)|\lesssim \frac{1}{|n_1-n_2|^{d+2-}},\ |M G_0 (n_1,n_2)|\lesssim \frac{1}{|n_1-n_2|^{d+2-}}.\]
  We remark that this regularization estimate is performed in the Fourier space via  controlling (derivatives of)   $\hat M*\hat G_0.$
  So, we get by applying Lemma \ref{Lemma 2.2} that
  \begin{align}\label{estimate on convolution part}
    |G_0 v^4 MG_0(n,n')|& \lesssim \kappa^4 \sum_{n_1\in Z^d}\frac{1}{|n-n_1|^{d-2}|n_1|^{4\alpha}|n_1-n'|^{d+2-}}\\
     \notag      &\lesssim_{d,\alpha} \kappa^4 \frac{1}{|n-n'|^{d-2}(|n|\wedge|n'|)^{4\alpha}}.
  \end{align}
The term $|G_0 M v^4 G_0(n,n')|$ has the same estimate. Moreover, we have 
  \begin{align}\label{estimate on difference part}
 |G_0 P_4 G_0(n,n')|&\lesssim_{d,\alpha}\kappa^4 \sum_{n_1,n_2\in \Z^d} \frac{1}{|n-n_1|^{d-2}(|n_1|+|n_2|)^2(|n_1|\wedge|n_2|)^{4\alpha}|n_1-n_2|^{3d-8}|n_2-n'|^{d-2}} \\
    \notag & \lesssim_{d,\alpha}\kappa^4 \sum_{n_1,n_2\in \Z^d} \frac{1}{|n-n_1|^{d-2}(|n_1|+|n_2|)^2}(\frac{1}{|n_1|^{4\alpha}}+\frac{1}{|n_2|^{4\alpha}})\frac{1}{|n_1-n_2|^{3d-8}|n_2-n'|^{d-2}} \\
    \notag &\lesssim_{d,\alpha}\kappa^4 \sum_{n_1,n_2\in \Z^d} \frac{1}{|n-n_1|^{d-2}|n_1|^{2+4\alpha}|n_1-n_2|^{3d-8}|n_2-n'|^{d-2}} \\
    \notag &\qquad +\kappa^4 \sum_{n_1,n_2\in \Z^d} \frac{1}{|n-n_1|^{d-2}|n_1-n_2|^{3d-8}|n_2|^{2+4\alpha}|n_2-n'|^{d-2}} \\
    \notag &\lesssim_{d,\alpha}\kappa^4 \left( \sum_{n_1\in \Z^d} \frac{1}{|n-n_1|^{d-2}|n_1|^{2+4\alpha}|n_1-n'|^{d-2}}+ \sum_{n_2\in \Z^d} \frac{1}{|n-n_2|^{d-2}|n_2|^{2+4\alpha}|n_2-n'|^{d-2}}\right)\\
    \notag &\lesssim_{d,\alpha} \kappa^4 \frac{1}{|n-n'|^{d-2}(|n|\wedge|n'|)^{4\alpha}},
  \end{align}
  where for the fourth inequality, we apply Lemma \ref{Lemma 2.1} and $d\geq 5 \Rightarrow 3d-8> d$, and for the fifth inequality,  we use Lemma \ref{Lemma 2.2} and $d\geq 5, \alpha\leq \frac13$ (this implies $d-2>4\alpha$). 
Taking account of  all above estimates   and  the decomposition \eqref{Decompose G0WG0} yields 
  \begin{equation}\label{decay of G0WG0}
      |G_0 WG_0(n,n')|\lesssim_{d,\alpha} \kappa^4 \frac{1}{|n-n'|^{d-2}(|n|\wedge|n'|)^{4\alpha}}. 
  \end{equation}
  Finally, combining with the moment estimates  on other random remaining terms  of  4th-order   proves  \eqref{i order remaining expactation} for $i=4$.

  When $i=5$, there are only random remaining  terms.  They can be estimated  directly by applying  Lemma \ref{decoupling}  similar to those of orders  $i=1,2,3$, except for terms like $G_0 V D_4 G_0,G_0V \wtg W G_0$ and $G_0W \wtg V G_0$. 
  On one hand, recall that
  \[D_4(n_1)=v^2_{n_1}\bigg[ \sum_{n_2\in\Z^d}v^2_{n_2} \wtg(n_1,n_2)^4\bigg].\]
  By Lemma \ref{Lemma 2.1},  we have since  $4(d-2)>d,$ 
  \[|D_4(n_1)|\lesssim_{d,\alpha}\kappa^4\frac{1}{|n_1|^{4\alpha}}.\]
  Hence by  applying the (decoupling) Lemma \ref{decoupling}, we get 
  \begin{align*}
    \E_p\left | G_0 VD_4 G_0 (n,n')\right | &\lesssim_{d,\alpha,p} \left(  \sum_{n_1\in \Z^d} |G_0(n,n_1)|^2|v_{n_1}D_4(n_1)|^2 |G_0(n,n')|^2 \right)^{\frac{1}{2}} \\
    \notag & \lesssim_{d,\alpha,p} \kappa^5 \left(  \sum_{n_1\in \Z^d} \frac{1}{|n-n_1|^{2(d-2)} |n_1|^{10\alpha} |n_1-n'|^{2(d-2)}}  \right) ^{\frac{1}{2}}\\
    \notag & \lesssim_{d,\alpha,p} \kappa^5 \frac{1}{|n-n'|^{d-2}(|n|\wedge|n'|)^{5\alpha}},
  \end{align*}
  where for the last inequality,  we use  Lemma \ref{Lemma 2.2} and $d\geq 5 \Rightarrow 2(d-2)> d$. 
  On the other hand,  by applying again  Lemma \ref{decoupling}, we have 
  \begin{align} \label{4.35}
    &\ \ \ \E_p\left |  G_0 W \wtg V G_0(n,n')  \right |\lesssim_{d,\alpha,p} \left(  \sum_{n_1\in \Z^d} (G_0 W \wtg (n,n_1))^2 v^2_{n_1} G_0(n_1,n')^2 \right)^{\frac{1}{2}}.
  \end{align}
  Now recall  that 
  \[G_0 W \wtg=G_0 WG_0 -\sigma G_0 W.\]
By  \eqref{decay of G0WG0} and 
  \begin{align*}
    |G_0 W(n,n')|&\lesssim \kappa^4 \sum_{n_1\in \Z^d} \frac{1}{|n-n_1|^{d-2}|n_1|^{2\alpha}|n_1-n'|^{3(d-2)}|n'|^{2\alpha}}\\
    \notag & \lesssim_{d,\alpha}\kappa^4 \frac{1}{|n-n'|^{d-2}(|n|\wedge|n'|)^{4\alpha}} \ ({\rm again \ by\  Lemma\  \ref{Lemma 2.2}}),  
  \end{align*}
  we have 
  \[|G_0 W \wtg (n,n')|\lesssim_{d,\alpha} \kappa^4 \frac{1}{|n-n'|^{d-2}(|n|\wedge|n'|)^{4\alpha}}.\]
  Hence, we continue to estimate \eqref{4.35} and obtain 
  \begin{align*} 
    &\ \ \ \E_p\left |  G_0 W \wtg V G_0(n,n')  \right |   \lesssim_{d,\alpha,p} \kappa^5 \left(  \sum_{n_1\in \Z^d} \frac{1}{|n-n_1|^{2(d-2)}(|n|\wedge|n_1|)^{8\alpha}|n_1|^{2\alpha}|n_1-n'|^{2(d-2)}} \right)^{\frac{1}{2}} \\
  \notag &\lesssim_{d,\alpha,p} \kappa^5 \frac{1}{|n-n'|^{d-2}(|n|\wedge|n'|)^{5\alpha}}. 
  \end{align*}
 The estimate on $G_0V \wtg W G_0$ remains the same.  Thus, we have proven  \eqref{i order remaining expactation} for $i=5$. 
\end{proof}

\subsection{Estimates on $i$th-order  remaining terms for $i=6,7$}

In this subsection, we aim to control remaining terms of orders $6$ and $7$. 

As we will see below,  it is  the  non-random and non-convolutional operator  $G_0CG_0$ (cf.  \eqref{4.9}) that plays a central role in the estimates. Recall that 
  \begin{align*}
  C&=C_6-\eta v^6,\ C_6(n_1,n_3)=v^2_{n_1}v^2_{n_3}\wtg(n_1,n_3) \sum_{n_2}\wtg(n_1,n_2)^2 v^2_{n_2} \wtg(n_2,n_3)^2.
  \end{align*}
  Visually, the operator $C$ involves the following summation graph.  Here in the  graph, the blue number on the edge represents the order of matrix $\wtg$, and the $v^2$ in each vertex means the corresponding $v^2$ occurring in the summation.  \\ 
  \begin{figure}[htbp]
    \centering

\tikzset{every picture/.style={line width=0.75pt}} 

\begin{tikzpicture}[x=0.75pt,y=0.75pt,yscale=-0.5,xscale=0.5]

\draw  [fill={rgb, 255:black, 208; green, 2; blue, 27 }  ,fill opacity=1 ] (159.4,70.29) .. controls (159.4,67.49) and (161.33,65.22) .. (163.71,65.22) .. controls (166.1,65.22) and (168.03,67.49) .. (168.03,70.29) .. controls (168.03,73.09) and (166.1,75.36) .. (163.71,75.36) .. controls (161.33,75.36) and (159.4,73.09) .. (159.4,70.29) -- cycle ;
\draw  [fill={rgb, 255:black, 208; green, 2; blue, 27 }  ,fill opacity=1 ] (259.79,115.02) .. controls (259.79,112.22) and (261.72,109.95) .. (264.11,109.95) .. controls (266.49,109.95) and (268.42,112.22) .. (268.42,115.02) .. controls (268.42,117.82) and (266.49,120.09) .. (264.11,120.09) .. controls (261.72,120.09) and (259.79,117.82) .. (259.79,115.02) -- cycle ;
\draw  [fill={rgb, 255:black, 208; green, 2; blue, 27 }  ,fill opacity=1 ] (381.4,114.29) .. controls (381.4,111.49) and (383.33,109.22) .. (385.71,109.22) .. controls (388.1,109.22) and (390.03,111.49) .. (390.03,114.29) .. controls (390.03,117.09) and (388.1,119.36) .. (385.71,119.36) .. controls (383.33,119.36) and (381.4,117.09) .. (381.4,114.29) -- cycle ;
\draw  [fill={rgb, 255:black, 208; green, 2; blue, 27 }  ,fill opacity=1 ] (323.08,221.29) .. controls (323.08,218.49) and (325.01,216.22) .. (327.4,216.22) .. controls (329.78,216.22) and (331.71,218.49) .. (331.71,221.29) .. controls (331.71,224.09) and (329.78,226.36) .. (327.4,226.36) .. controls (325.01,226.36) and (323.08,224.09) .. (323.08,221.29) -- cycle ;
\draw  [fill={rgb, 255:black, 208; green, 2; blue, 27 }  ,fill opacity=1 ] (481.4,67.29) .. controls (481.4,64.49) and (483.33,62.22) .. (485.71,62.22) .. controls (488.1,62.22) and (490.03,64.49) .. (490.03,67.29) .. controls (490.03,70.09) and (488.1,72.36) .. (485.71,72.36) .. controls (483.33,72.36) and (481.4,70.09) .. (481.4,67.29) -- cycle ;
\draw    (163.71,70.29) -- (264.11,115.02) ;
\draw    (264.11,115.02) -- (327.4,221.29) ;
\draw    (327.4,221.29) -- (385.71,114.29) ;
\draw    (264.11,115.02) -- (385.71,114.29) ;
\draw    (385.71,114.29) -- (485.71,67.29) ;
\draw (128,60) node [anchor=north west][inner sep=0.75pt]   [align=left] {$n$};
\draw (503,60) node [anchor=north west][inner sep=0.75pt]   [align=left] {$n'$};
\draw (254,133) node [anchor=north west][inner sep=0.75pt]   [align=left] {$n_1,v^2$};
\draw (318,232) node [anchor=north west][inner sep=0.75pt]   [align=left] {$n_2,v^2$};
\draw (386,129) node [anchor=north west][inner sep=0.75pt]   [align=left] {$n_3,v^2$};
\draw (217,70) node [anchor=north west][inner sep=0.75pt]   [align=left] {\textcolor[rgb]{0.29,0.56,0.89}{1}};
\draw (323,95) node [anchor=north west][inner sep=0.75pt]   [align=left] {\textcolor[rgb]{0.29,0.56,0.89}{1}};
\draw (429,74) node [anchor=north west][inner sep=0.75pt]   [align=left] {\textcolor[rgb]{0.29,0.56,0.89}{1}};
\draw (284,173) node [anchor=north west][inner sep=0.75pt]   [align=left] {\textcolor[rgb]{0.29,0.56,0.89}{2}};
\draw (358.56,170.79) node [anchor=north west][inner sep=0.75pt]   [align=left] {\textcolor[rgb]{0.29,0.56,0.89}{2}};
\end{tikzpicture}
  \end{figure}
  \ \\ Our aim is to use the symmetric difference trick, \cite[(1.10)]{Bou03}, to decompose this  diagram into \\
    \begin{figure}[htbp]
    \centering

    \tikzset{every picture/.style={line width=0.75pt}} 

    \begin{tikzpicture}[x=0.75pt,y=0.75pt,yscale=-0.7,xscale=0.7]
    
    \draw  [fill={rgb, 255:black, 208; green, 2; blue, 27 }  ,fill opacity=1 ] (26.71,71.59) .. controls (26.71,69.63) and (27.89,68.05) .. (29.34,68.05) .. controls (30.79,68.05) and (31.97,69.63) .. (31.97,71.59) .. controls (31.97,73.55) and (30.79,75.14) .. (29.34,75.14) .. controls (27.89,75.14) and (26.71,73.55) .. (26.71,71.59) -- cycle ;
    \draw  [fill={rgb, 255:black, 208; green, 2; blue, 27 }  ,fill opacity=1 ] (87.8,102.88) .. controls (87.8,100.92) and (88.98,99.34) .. (90.43,99.34) .. controls (91.88,99.34) and (93.06,100.92) .. (93.06,102.88) .. controls (93.06,104.84) and (91.88,106.43) .. (90.43,106.43) .. controls (88.98,106.43) and (87.8,104.84) .. (87.8,102.88) -- cycle ;
    \draw  [fill={rgb, 255:black, 208; green, 2; blue, 27 }  ,fill opacity=1 ] (161.81,102.37) .. controls (161.81,100.41) and (162.98,98.82) .. (164.43,98.82) .. controls (165.88,98.82) and (167.06,100.41) .. (167.06,102.37) .. controls (167.06,104.33) and (165.88,105.91) .. (164.43,105.91) .. controls (162.98,105.91) and (161.81,104.33) .. (161.81,102.37) -- cycle ;
    \draw  [fill={rgb, 255:black, 208; green, 2; blue, 27 }  ,fill opacity=1 ] (126.32,177.2) .. controls (126.32,175.24) and (127.49,173.66) .. (128.95,173.66) .. controls (130.4,173.66) and (131.57,175.24) .. (131.57,177.2) .. controls (131.57,179.16) and (130.4,180.75) .. (128.95,180.75) .. controls (127.49,180.75) and (126.32,179.16) .. (126.32,177.2) -- cycle ;
    \draw  [fill={rgb, 255:black, 208; green, 2; blue, 27 }  ,fill opacity=1 ] (222.66,69.49) .. controls (222.66,67.54) and (223.83,65.95) .. (225.28,65.95) .. controls (226.74,65.95) and (227.91,67.54) .. (227.91,69.49) .. controls (227.91,71.45) and (226.74,73.04) .. (225.28,73.04) .. controls (223.83,73.04) and (222.66,71.45) .. (222.66,69.49) -- cycle ;
    \draw    (29.34,71.59) -- (90.43,102.88) ;
    \draw    (90.43,102.88) -- (128.95,177.2) ;
    \draw    (128.95,177.2) -- (164.43,102.37) ;
    \draw    (90.43,102.88) -- (164.43,102.37) ;
    \draw    (164.43,102.37) -- (225.28,69.49) ;
    \draw  [fill={rgb, 255:black, 208; green, 2; blue, 27 }  ,fill opacity=1 ] (284.06,72.59) .. controls (284.06,70.63) and (285.24,69.05) .. (286.69,69.05) .. controls (288.14,69.05) and (289.32,70.63) .. (289.32,72.59) .. controls (289.32,74.55) and (288.14,76.14) .. (286.69,76.14) .. controls (285.24,76.14) and (284.06,74.55) .. (284.06,72.59) -- cycle ;
    \draw  [fill={rgb, 255:black, 208; green, 2; blue, 27 }  ,fill opacity=1 ] (345.15,103.88) .. controls (345.15,101.92) and (346.33,100.34) .. (347.78,100.34) .. controls (349.23,100.34) and (350.41,101.92) .. (350.41,103.88) .. controls (350.41,105.84) and (349.23,107.43) .. (347.78,107.43) .. controls (346.33,107.43) and (345.15,105.84) .. (345.15,103.88) -- cycle ;
    \draw  [fill={rgb, 255:black, 208; green, 2; blue, 27 }  ,fill opacity=1 ] (419.15,103.37) .. controls (419.15,101.41) and (420.33,99.82) .. (421.78,99.82) .. controls (423.23,99.82) and (424.41,101.41) .. (424.41,103.37) .. controls (424.41,105.33) and (423.23,106.91) .. (421.78,106.91) .. controls (420.33,106.91) and (419.15,105.33) .. (419.15,103.37) -- cycle ;
    \draw  [fill={rgb, 255:black, 208; green, 2; blue, 27 }  ,fill opacity=1 ] (383.67,178.2) .. controls (383.67,176.24) and (384.84,174.66) .. (386.29,174.66) .. controls (387.75,174.66) and (388.92,176.24) .. (388.92,178.2) .. controls (388.92,180.16) and (387.75,181.75) .. (386.29,181.75) .. controls (384.84,181.75) and (383.67,180.16) .. (383.67,178.2) -- cycle ;
    \draw  [fill={rgb, 255:black, 208; green, 2; blue, 27 }  ,fill opacity=1 ] (480.01,70.49) .. controls (480.01,68.54) and (481.18,66.95) .. (482.63,66.95) .. controls (484.08,66.95) and (485.26,68.54) .. (485.26,70.49) .. controls (485.26,72.45) and (484.08,74.04) .. (482.63,74.04) .. controls (481.18,74.04) and (480.01,72.45) .. (480.01,70.49) -- cycle ;
    \draw    (286.69,72.59) -- (347.78,103.88) ;
    \draw    (347.78,103.88) -- (386.29,178.2) ;
    \draw    (386.29,178.2) -- (421.78,103.37) ;
    \draw    (347.78,103.88) -- (421.78,103.37) ;
    \draw    (421.78,103.37) -- (482.63,70.49) ;
    \draw  [fill={rgb, 255:black, 208; green, 2; blue, 27 }  ,fill opacity=1 ] (537.06,72.59) .. controls (537.06,70.63) and (538.24,69.05) .. (539.69,69.05) .. controls (541.14,69.05) and (542.32,70.63) .. (542.32,72.59) .. controls (542.32,74.55) and (541.14,76.14) .. (539.69,76.14) .. controls (538.24,76.14) and (537.06,74.55) .. (537.06,72.59) -- cycle ;
    \draw  [fill={rgb, 255:black, 208; green, 2; blue, 27 }  ,fill opacity=1 ] (598.15,103.88) .. controls (598.15,101.92) and (599.33,100.34) .. (600.78,100.34) .. controls (602.23,100.34) and (603.41,101.92) .. (603.41,103.88) .. controls (603.41,105.84) and (602.23,107.43) .. (600.78,107.43) .. controls (599.33,107.43) and (598.15,105.84) .. (598.15,103.88) -- cycle ;
    \draw  [fill={rgb, 255:black, 208; green, 2; blue, 27 }  ,fill opacity=1 ] (672.15,103.37) .. controls (672.15,101.41) and (673.33,99.82) .. (674.78,99.82) .. controls (676.23,99.82) and (677.41,101.41) .. (677.41,103.37) .. controls (677.41,105.33) and (676.23,106.91) .. (674.78,106.91) .. controls (673.33,106.91) and (672.15,105.33) .. (672.15,103.37) -- cycle ;
    \draw  [fill={rgb, 255:black, 208; green, 2; blue, 27 }  ,fill opacity=1 ] (636.67,178.2) .. controls (636.67,176.24) and (637.84,174.66) .. (639.29,174.66) .. controls (640.75,174.66) and (641.92,176.24) .. (641.92,178.2) .. controls (641.92,180.16) and (640.75,181.75) .. (639.29,181.75) .. controls (637.84,181.75) and (636.67,180.16) .. (636.67,178.2) -- cycle ;
    \draw  [fill={rgb, 255:black, 208; green, 2; blue, 27 }  ,fill opacity=1 ] (733.01,70.49) .. controls (733.01,68.54) and (734.18,66.95) .. (735.63,66.95) .. controls (737.08,66.95) and (738.26,68.54) .. (738.26,70.49) .. controls (738.26,72.45) and (737.08,74.04) .. (735.63,74.04) .. controls (734.18,74.04) and (733.01,72.45) .. (733.01,70.49) -- cycle ;
    \draw    (539.69,72.59) -- (600.78,103.88) ;
    \draw    (600.78,103.88) -- (639.29,178.2) ;
    \draw    (639.29,178.2) -- (674.78,103.37) ;
    \draw    (600.78,103.88) -- (674.78,103.37) ;
    \draw    (674.78,103.37) -- (735.63,70.49) ;
    
    \draw (5.65,61.99) node [anchor=north west][inner sep=0.75pt]   [align=left] {$n$};
    \draw (233.45,61.99) node [anchor=north west][inner sep=0.75pt]   [align=left] {$n'$};
    \draw (80.76,113.05) node [anchor=north west][inner sep=0.75pt]   [align=left] {$n_1,v^6$};
    \draw (119.7,182.29) node [anchor=north west][inner sep=0.75pt]   [align=left] {$n_2$};
    \draw (161.08,110.25) node [anchor=north west][inner sep=0.75pt]   [align=left] {$n_3$};
    \draw (59.81,68.99) node [anchor=north west][inner sep=0.75pt]   [align=left] {\textcolor[rgb]{0.29,0.56,0.89}{1}};
    \draw (124.31,86.47) node [anchor=north west][inner sep=0.75pt]   [align=left] {\textcolor[rgb]{0.29,0.56,0.89}{1}};
    \draw (183.95,69.69) node [anchor=north west][inner sep=0.75pt]   [align=left] {\textcolor[rgb]{0.29,0.56,0.89}{1}};
    \draw (100.58,141.03) node [anchor=north west][inner sep=0.75pt]   [align=left] {\textcolor[rgb]{0.29,0.56,0.89}{2}};
    \draw (145.95,139.48) node [anchor=north west][inner sep=0.75pt]   [align=left] {\textcolor[rgb]{0.29,0.56,0.89}{2}};
    \draw (263,62.99) node [anchor=north west][inner sep=0.75pt]   [align=left] {$n$};
    \draw (490.8,62.99) node [anchor=north west][inner sep=0.75pt]   [align=left] {$n'$};
    \draw (338.11,114.05) node [anchor=north west][inner sep=0.75pt]   [align=left] {$n_1$};
    \draw (377.05,183.29) node [anchor=north west][inner sep=0.75pt]   [align=left] {$n_2,v^6$};
    \draw (418.43,111.25) node [anchor=north west][inner sep=0.75pt]   [align=left] {$n_3$};
    \draw (317.16,69.99) node [anchor=north west][inner sep=0.75pt]   [align=left] {\textcolor[rgb]{0.29,0.56,0.89}{1}};
    \draw (381.66,87.47) node [anchor=north west][inner sep=0.75pt]   [align=left] {\textcolor[rgb]{0.29,0.56,0.89}{1}};
    \draw (441.3,70.69) node [anchor=north west][inner sep=0.75pt]   [align=left] {\textcolor[rgb]{0.29,0.56,0.89}{1}};
    \draw (357.93,142.03) node [anchor=north west][inner sep=0.75pt]   [align=left] {\textcolor[rgb]{0.29,0.56,0.89}{2}};
    \draw (403.3,140.48) node [anchor=north west][inner sep=0.75pt]   [align=left] {\textcolor[rgb]{0.29,0.56,0.89}{2}};
    \draw (516,62.99) node [anchor=north west][inner sep=0.75pt]   [align=left] {$n$};
    \draw (743.8,62.99) node [anchor=north west][inner sep=0.75pt]   [align=left] {$n'$};
    \draw (591.11,114.05) node [anchor=north west][inner sep=0.75pt]   [align=left] {$n_1$};
    \draw (630.05,183.29) node [anchor=north west][inner sep=0.75pt]   [align=left] {$n_2$};
    \draw (671.43,111.25) node [anchor=north west][inner sep=0.75pt]   [align=left] {$n_3,v^6$};
    \draw (570.16,69.99) node [anchor=north west][inner sep=0.75pt]   [align=left] {\textcolor[rgb]{0.29,0.56,0.89}{1}};
    \draw (634.66,87.47) node [anchor=north west][inner sep=0.75pt]   [align=left] {\textcolor[rgb]{0.29,0.56,0.89}{1}};
    \draw (694.3,70.69) node [anchor=north west][inner sep=0.75pt]   [align=left] {\textcolor[rgb]{0.29,0.56,0.89}{1}};
    \draw (610.93,142.03) node [anchor=north west][inner sep=0.75pt]   [align=left] {\textcolor[rgb]{0.29,0.56,0.89}{2}};
    \draw (656.3,140.48) node [anchor=north west][inner sep=0.75pt]   [align=left] {\textcolor[rgb]{0.29,0.56,0.89}{2}};

    \end{tikzpicture}

  \end{figure}
\ \\ 
  
 Denote  first 
\[\boxed{G_0,6}_r=\boxed{G_0,6}-4G_0P_6^{''}G_0,\] 
\[\boxed{G_0,7}_r=\boxed{G_0,6}+4G_0 VG_0P_6^{''}G_0+4G_0P_6^{''}G_0VG_0,\]
where $P_6^{''}$ is a {\it singular part}  extracted from operator $C$, with 
\[P_6^{''}(n_1,n_3)= \wtg(n_1,n_3) \sum_{n_2\in \Z^d}(v_{n_2}^6-v_{n_1}^6) \wtg(n_1,n_2)^2 \wtg(n_2,n_3)^2.\]
The main theorem in this subsection is 
\begin{thm}\label{6,7 resolvent estimate}
  We have the following estimates: 
 \begin{itemize}
 \item[(1)]For  $C=(C-P^{''}_6)+P_6^{''}$,  
  \begin{align*}
    |P_6^{''}(n_1,n_3)| &\lesssim_{d,\alpha} \kappa^6\frac{1}{|n_1-n_3|^{3(d-2)-1}(|n_3|\wedge|n_1|)^{1+6\alpha}},\\
  |G_0 (C-P_6^{''}) G_0(n,n')| &\lesssim_{d,\alpha}\kappa^6 \frac{1}{|n-n'|^{d-2}(|n|\wedge|n'|)^{6\alpha}}. 
  \end{align*}
  \item[(2)] For $i=6,7$ and $p\geq 1$,
  \begin{equation}\label{moment estimate on 6,7 resolvent}
    \mathbb{E}_p \big| \ \boxed{G_0,i}_r (n,n')\big| \lesssim_{d,p,\alpha}\kappa^i \frac{1}{(|n|\wedge|n'|)^{i\alpha}|n-n'|^{d-2}}.  
  \end{equation}
\end{itemize}
\end{thm}

\begin{proof}[Proof of Theorem \ref{6,7 resolvent estimate}]
  (1) The proof is   based on combining the   symmetrical  difference regularization and convolution regularization arguments. 
   Indeed, by \\
  \begin{align*}
    v^2_{n_1}v_{n_2}^2v_{n_3}^2 
    =&\frac{1}{6}\bigg(2(v_{n_1}^6+v_{n_2}^6+v_{n_3}^6)\\
    & -(v_{n_1}^4-v_{n_3}^4)(v_{n_1}^2-v_{n_3}^2)-(v_{n_2}^4-v_{n_3}^4)(v_{n_2}^2-v_{n_3}^2)-(v_{n_1}^4-v_{n_2}^4)(v_{n_1}^2-v_{n_2}^2) \\
    & -v_{n_3}^2(v_{n_1}^2-v_{n_2}^2)^2-v_{n_2}^2(v_{n_1}^2-v_{n_3}^2)^2-v_{n_1}^2(v_{n_2}^2-v_{n_3}^2)^2\bigg), 
  \end{align*}
  we obtain the decomposition
  \begin{equation}\label{decomposition of G_0CG_0, initial}
    G_0 C G_0=\frac{1}{3}G_0(C^{(1)}+C^{(2)}+C^{(3)})G_0-\frac{1}{6}G_0(P_{1,2}+P_{1,3}+P_{2,3}+\widetilde{P}_{1,2}+\widetilde{P}_{2,3}+\widetilde{P}_{1,3})G_0
  \end{equation}
  with 
  \[C^{(i)}(n_1,n_3)=\wtg(n_1,n_3) \sum_{n_2} v_{n_i}^6\wtg(n_1,n_2)^2  \wtg(n_2,n_3)^2-\eta v^6_{n_1}\delta_{n_1,n_3},\ i=1,2,3,\]
  \[P_{1,2}=\wtg(n_1,n_3) \sum_{n_2} v_{n_3}^2(v_{n_1}^2-v_{n_2}^2)^2\wtg(n_1,n_2)^2 \wtg(n_2,n_3)^2,\]
  \[P_{1,3}=\wtg(n_1,n_3) \sum_{n_2} v_{n_2}^2(v_{n_1}^2-v_{n_3}^2)^2\wtg(n_1,n_2)^2 \wtg(n_2,n_3)^2,\]
  \[P_{2,3}=\wtg(n_1,n_3) \sum_{n_2} v_{n_1}^2(v_{n_2}^2-v_{n_3}^2)^2\wtg(n_1,n_2)^2 \wtg(n_2,n_3)^2,\]
  \[\widetilde{P}_{{i,j}}=\wtg(n_1,n_3) \sum_{n_2} (v_{n_i}^4-v_{n_j}^4)(v_{n_i}^2-v_{n_j}^2)\wtg(n_1,n_2)^2 \wtg(n_2,n_3)^2,\ 1\leq i\neq j \leq 3.\]
  We first control $P_{i,j},\ \widetilde{P}_{i,j}$ which depends on difference regularization lemma (cf. Lemma \ref{difference property}). Applying  Lemma \ref{difference property} implies 
  \begin{align}\label{4.43}
    |P_{1,2}(n_1,n_3)|& \lesssim \frac{\kappa^6}{|n_1-n_3|^{d-2}|n_3|^{2\alpha}}\sum_{n_2\in \Z^d }\frac{|n_1-n_2|^2}{|n_1-n_2|^{2(d-2)}(|n_1|+|n_2|)^2(|n_2|\wedge|n_2|)^{4\alpha}|n_2-n_3|^{2(d-2)}}\\
    \notag   &\lesssim \frac{\kappa^6}{|n_1-n_3|^{d-2}|n_3|^{2\alpha}}\sum_{n_2\in \Z^d } \frac{1}{|n_1-n_2|^{2(d-2)-2}}(\frac{1}{|n_1|^{2+4\alpha}}+\frac{1}{|n_2|^{2+4\alpha}})\frac{1}{|n_2-n_3|^{2(d-2)}}\\
    \notag &\lesssim_{d,\alpha}\kappa^6\frac{1}{|n_1-n_3|^{3(d-2)-2}(|n_3|\wedge|n_1|)^{2+6\alpha}},
  \end{align}
  where for the third inequality, we apply Lemma \ref{Lemma 2.1}, Lemma \ref{Lemma 2.2} together with $d\geq 5,  \alpha \leq \frac13$ (this implies $ 2(d-2)>d, 2(d-2)-2 \geq 2 + 4\alpha$). Similarly, we have 
  \begin{align}
    |P_{2,3}(n_1,n_3)| & \lesssim_{d,\alpha}\kappa^6\frac{1}{|n_1-n_3|^{3(d-2)-2}(|n_3|\wedge|n_1|)^{2+6\alpha}}
  \end{align}
  and 
  \begin{align}
    |P_{1,3}(n_1,n_3)| & \lesssim \frac{\kappa^6|n_1-n_3|^2}{|n_1-n_3|^{d-2}(|n_1|+|n_3|)^2(|n_1|\wedge|n_3|)^{4\alpha}}\sum_{n_2\in \Z^d }\frac{1}{|n_1-n_2|^{2(d-2)}|n_2|^{2\alpha}|n_2-n_3|^{2(d-2)}}\\
    \notag &\lesssim_{d,\alpha}\kappa^6\frac{1}{|n_1-n_3|^{3(d-2)-2}(|n_3|\wedge|n_1|)^{2+6\alpha}}
  \end{align}
  Moreover,  similar to the proof of  \eqref{4.43}, we obtain 
  \begin{align}
    |\widetilde{P}_{1,2}(n_1,n_3)| & \lesssim \frac{\kappa^6}{|n_1-n_3|^{d-2}}\sum_{n_2\in \Z^d }\frac{|n_1-n_2|^2}{|n_1-n_2|^{2(d-2)}(|n_1|+|n_2|)^2(|n_2|\wedge|n_2|)^{6\alpha}|n_2-n_3|^{2(d-2)}}\\
    \notag &\lesssim_{d,\alpha}\kappa^6\frac{1}{|n_1-n_3|^{3(d-2)-2}(|n_3|\wedge|n_1|)^{2+6\alpha}}
  \end{align}
  and 
  \begin{align}
    |\widetilde{P}_{2,3}(n_1,n_3)| & \lesssim_{d,\alpha}\kappa^6\frac{1}{|n_1-n_3|^{3(d-2)-2}(|n_3|\wedge|n_1|)^{2+6\alpha}},
  \end{align}
  \begin{align}\label{4.48}
    |\widetilde{P}_{1,3}(n_1,n_3)| & \lesssim_{d,\alpha}\kappa^6\frac{1}{|n_1-n_3|^{3(d-2)-2}(|n_3|\wedge|n_1|)^{2+6\alpha}}. 
  \end{align}
  Hence, by $\eqref{4.43}\sim \eqref{4.48}$, if we denote 
  \[P'_6=P_{1,2}+P_{1,3}+P_{2,3}+\widetilde{P}_{1,2}+\widetilde{P}_{1,3}+\widetilde{P}_{2,3},\]
  we have 
  \begin{equation}\label{P6'}
    |P'_6(n_1,n_3)|  \lesssim_{d,\alpha}\kappa^6\frac{1}{|n_1-n_3|^{3(d-2)-2}(|n_3|\wedge|n_1|)^{2+6\alpha}}
  \end{equation}  
  and 
  \begin{equation}\label{decomposition of G_0CG_0}
    G_0 C G_0=\frac{1}{3}G_0(C^{(1)}+C^{(2)}+C^{(3)})G_0-\frac{1}{6}G_0P'_6G_0.
  \end{equation} 
  At this stage,  similar to the proof of  \eqref{estimate on difference part}, we  get (since $d\geq 5, \alpha\leq \frac13$ implies $d>6\alpha+2$)
  \begin{align}\label{estimate on difference part, P'6}
    |G_0 P'_6 G_0(n,n')| & \lesssim_{d,\alpha}\kappa^6 \sum_{n_1,n_3\in\Z^d}\frac{1}{|n-n_1|^{d-2}|n_1-n_3|^{3(d-2)-2}(|n_3|\wedge|n_1|)^{2+6\alpha}|n_3-n'|^{d-2}} \\
    \notag &\lesssim_{d,\alpha}\kappa^6 \frac{1}{|n-n'|^{d-2}(|n|\wedge|n'|)^{6\alpha}}. 
  \end{align}\\
  
  Next, it remains to control $G_0 C^{(1)}G_0,G_0 C^{(2)}G_0,G_0 C^{(3)}G_0$,  which is mainly based on the convolution regularization argument.  More precisely, we have
  \[G_0 C^{(1)}G_0=G_0 v^6 \widetilde{N}G_0,G_0 C^{(3)}G_0=G_0\widetilde{N}v^6G_0,\]
  where 
  $$N(n_1,n_3)=\wtg(n_1,n_3) \sum_{n_2}\wtg(n_1,n_2)^2  \wtg(n_2,n_3)^2,\ \widetilde{N}=N-\eta. $$
  We first estimate $G_0\widetilde{N}$ and $\widetilde{N}G_0$, which relies on  the following lemma.  
 \begin{lem}\label{FGNlem}
   Let $d\geq 5$. Then 
   \begin{itemize}
  \item[(1)] For $|\beta|_1<d-2$ and $p\geq 1$ satisfying that $p(2+|\beta|_1)<d$, we have $\partial^{\beta}\hat{G_0}\in L^p$. 
  \item[(2)]  For $|\beta|_1\leq 2(d-2)-1$, we have that $\partial^{\beta}f^2\in L^1$, where 
\begin{align}\label{fxi}
f(\xi)=[(\hat{G_0}-\sigma)*(\hat{G_0}-\sigma)](\xi),\sigma=G_0(0,0)=\int_{\T^d}\hat{G_0}(\xi){\rm d}\xi.
\end{align}
\end{itemize}
\end{lem}
\begin{rmk}
  \begin{itemize}
    \item[(1)] Especially, by taking  $p=1$, we have 
    \[\partial^{\beta}\hat{G_0}\in L^1,\  |\beta|_1<d-2.\]
    \item[(2)] A direct corollary of this lemma is  
           \[\partial^{\beta} \hat{K}=\partial^{\frac{\beta}{3}} \hat{G_0}*\partial^{\frac{\beta}{3}} \hat{G_0}*\partial^{\frac{\beta}{3}} \hat{G_0}\in L^1,\ |\beta|_1<3(d-2)\]
           and 
           \[\partial^{\beta} \hat{N}\in L^1,\ |\beta|_1<3(d-2)-1.\]
          So consider 
      \[\widehat{N}(\xi)=(\hat{G_0}-\sigma)*\big[(\hat{G_0}-\sigma)*(\hat{G_0}-\sigma)\big]^2(\xi)\]
     and  take $\eta=\widehat{N}(0)$.  Since  $\widehat{N}(\xi)$ is a an even symmetric function in $\xi_1,\cdots, \xi_d$,  by Lemma \ref{Lemma 2.1}, we have 
      \begin{align*}
        |N(n_1,n_3)|& \lesssim \frac{1}{|n_1-n_3|^{d-2}}\sum_{n_2\in \Z^d}\frac{1}{|n_1-n_2|^{2(d-2)}|n_2-n_3|^{2(d-2)}}\\
             &\lesssim_{d} \frac{1}{|n_1-n_3|^{3(d-2)}}.
      \end{align*}
      Moreover, this lemma  implies 
      \[\partial^{\beta}\widehat{N}\in L^1 \quad {\rm for} \quad |\beta|_1<3(d-2)-1.\]
      From the above analysis,  we have 
      \[\widehat{N}(\xi)-\eta=c\|\xi\|^2 +\mcO(\|\xi\|^4) \]
      and 
     \[ {{\widehat{G_0}(\xi)\widehat{N}(\xi)-\eta}}=c +\frac{\mcO(\|\xi\|^4)}{\|\xi\|^2+\mcO(\|\xi\|^4)},\]
     of which the  $((3(d-2)-1)-)${\rm th}  order weak derivatives  belong  to  $L^1$. 
      As a result, if we require that $3(d-2)-1\geq d+2$ (this needs $d\geq 5$), the standard Fourier analysis argument as in  \cite[(1.8)--(1.9)]{Bou03} will ensure that
      \begin{equation}\label{good estimate of G_0N}
      |G_0(N-\eta)(n,n')|\leq |n-n'|^{-(d+2-)}. 
      \end{equation}
Thus, we have obtained for $d\geq 5,$
 \begin{align}
     \label{G_0 tN} &|G_0 \widetilde{N}(n,n')|\lesssim \frac{1}{|n-n'|^{d+2-}}, \\
     \label{tNG_0} &| \widetilde{N}G_0(n,n')|\lesssim \frac{1}{|n-n'|^{d+2-}}.
  \end{align}

  \end{itemize}
\end{rmk}
\begin{proof}
 The proof is based on the fractional Gagliardo-Nirenberg inequality \cite{BM18}, and we refer to Appendix \ref{GNineq} for details. 
\end{proof}
We continue to the estimates. By \eqref{G_0 tN} and \eqref{tNG_0}, we have
  \begin{align}\label{4.54}
    |G_0 C^{(1)}G_0(n,n')| &\lesssim \kappa^6\sum_{n_1\in \Z^d}\frac{1}{|n-n_1|^{d-2}|n_1|^{6\alpha}|n_1-n'|^{d+2-}}\\
     \notag &\lesssim_{d,\alpha}\kappa^6 \frac{1}{|n-n'|^{d-2}(|n|\wedge|n'|)^{6\alpha}}
  \end{align}
  and 
  \begin{align}\label{4.55}
    |G_0 C^{(3)}G_0(n,n')| &\lesssim_{d,\alpha} \kappa^6 \frac{1}{|n-n'|^{d-2}(|n|\wedge|n'|)^{6\alpha}}.
  \end{align}
The main obstacle is the term $G_0 C^{(2)}G_0$, which  comes from the following summation graph:  \\
  \begin{figure}[htbp]
    \centering
    \tikzset{every picture/.style={line width=0.75pt}} 

    \begin{tikzpicture}[x=0.75pt,y=0.75pt,yscale=-0.7,xscale=0.7]
    
    \draw  [fill={rgb, 255:black, 208; green, 2; blue, 27 }  ,fill opacity=1 ] (229.06,71.59) .. controls (229.06,69.63) and (230.24,68.05) .. (231.69,68.05) .. controls (233.14,68.05) and (234.32,69.63) .. (234.32,71.59) .. controls (234.32,73.55) and (233.14,75.14) .. (231.69,75.14) .. controls (230.24,75.14) and (229.06,73.55) .. (229.06,71.59) -- cycle ;
    \draw  [fill={rgb, 255:black, 208; green, 2; blue, 27 }  ,fill opacity=1 ] (290.15,102.88) .. controls (290.15,100.92) and (291.33,99.34) .. (292.78,99.34) .. controls (294.23,99.34) and (295.41,100.92) .. (295.41,102.88) .. controls (295.41,104.84) and (294.23,106.43) .. (292.78,106.43) .. controls (291.33,106.43) and (290.15,104.84) .. (290.15,102.88) -- cycle ;
    \draw  [fill={rgb, 255:black, 208; green, 2; blue, 27 }  ,fill opacity=1 ] (364.15,102.37) .. controls (364.15,100.41) and (365.33,98.82) .. (366.78,98.82) .. controls (368.23,98.82) and (369.41,100.41) .. (369.41,102.37) .. controls (369.41,104.33) and (368.23,105.91) .. (366.78,105.91) .. controls (365.33,105.91) and (364.15,104.33) .. (364.15,102.37) -- cycle ;
    \draw  [fill={rgb, 255:black, 208; green, 2; blue, 27 }  ,fill opacity=1 ] (328.67,177.2) .. controls (328.67,175.24) and (329.84,173.66) .. (331.29,173.66) .. controls (332.75,173.66) and (333.92,175.24) .. (333.92,177.2) .. controls (333.92,179.16) and (332.75,180.75) .. (331.29,180.75) .. controls (329.84,180.75) and (328.67,179.16) .. (328.67,177.2) -- cycle ;
    \draw  [fill={rgb, 255:black, 208; green, 2; blue, 27 }  ,fill opacity=1 ] (425.01,69.49) .. controls (425.01,67.54) and (426.18,65.95) .. (427.63,65.95) .. controls (429.08,65.95) and (430.26,67.54) .. (430.26,69.49) .. controls (430.26,71.45) and (429.08,73.04) .. (427.63,73.04) .. controls (426.18,73.04) and (425.01,71.45) .. (425.01,69.49) -- cycle ;
    \draw    (231.69,71.59) -- (292.78,102.88) ;
    \draw    (292.78,102.88) -- (331.29,177.2) ;
    \draw    (331.29,177.2) -- (366.78,102.37) ;
    \draw    (292.78,102.88) -- (366.78,102.37) ;
    \draw    (366.78,102.37) -- (427.63,69.49) ;
    
    \draw (208,61.99) node [anchor=north west][inner sep=0.75pt]   [align=left] {$n$};
    \draw (435.8,61.99) node [anchor=north west][inner sep=0.75pt]   [align=left] {$n'$};
    \draw (283.11,113.05) node [anchor=north west][inner sep=0.75pt]   [align=left] {$n_1$};
    \draw (322.05,182.29) node [anchor=north west][inner sep=0.75pt]   [align=left] {$n_2,v^6$};
    \draw (363.43,110.25) node [anchor=north west][inner sep=0.75pt]   [align=left] {$n_3$};
    \draw (262.16,68.99) node [anchor=north west][inner sep=0.75pt]   [align=left] {\textcolor[rgb]{0.29,0.56,0.89}{1}};
    \draw (326.66,86.47) node [anchor=north west][inner sep=0.75pt]   [align=left] {\textcolor[rgb]{0.29,0.56,0.89}{1}};
    \draw (386.3,69.69) node [anchor=north west][inner sep=0.75pt]   [align=left] {\textcolor[rgb]{0.29,0.56,0.89}{1}};
    \draw (302.93,141.03) node [anchor=north west][inner sep=0.75pt]   [align=left] {\textcolor[rgb]{0.29,0.56,0.89}{2}};
    \draw (348.3,139.48) node [anchor=north west][inner sep=0.75pt]   [align=left] {\textcolor[rgb]{0.29,0.56,0.89}{2}};
  
    \end{tikzpicture}

  \end{figure}
  \ \\
This  triangle structure between vertices $n_1$ and $n_3$ is {\it non-convolutional} due to the presence of $v^6$ at the  vertex $n_2$. So,  we cannot construct a symmetric difference  for this graph, but    take  a direct difference. Specifically,  from  \[v^6_{n_2}=v^6_{n_1}+(v_{n_2}^6-v_{n_1}^6),\] 
  it follows that 
  \begin{equation}
    G_0 C^{(2)} G_0 =G_0 C^{(1)} G_0 +G_0 P_6^{''}G_0. 
  \end{equation}
  This is why we introduce the error term $P_6^{''}$ of such form.  Then  by Lemma \ref{difference property}, we get 
  \begin{align}\label{P6''}
    |P_6^{''}(n_1,n_3)| &= |\wtg(n_1,n_3) \sum_{n_1\in \Z^d}(v_{n_1}^6-v_{n_2}^6) \wtg(n_1,n_2)^2 \wtg(n_2,n_3)^2 |\\
    \notag & \lesssim_{d,\alpha} \frac{\kappa^6}{|n_1-n_3|^{d-2}}\sum_{n_2\in \Z^d }\frac{|n_1-n_2|}{|n_1-n_2|^{2(d-2)}(|n_1|+|n_2|)(|n_2|\wedge|n_2|)^{6\alpha}|n_2-n_3|^{2(d-2)}}\\
    \notag &\lesssim_{d,\alpha} \kappa^6\frac{1}{|n_1-n_3|^{3(d-2)-1}(|n_3|\wedge|n_1|)^{1+6\alpha}},
  \end{align}
  which  is the first conclusion  in  Theorem \ref{6,7 resolvent estimate} (1). 
  
  Next,  by summarizing all  the estimates  \eqref{estimate on difference part, P'6}, \eqref{4.54} and \eqref{4.55} and combining  the relation
  \[C-P_6^{''}=\frac{1}{3}(2C^{(1)}+C^{(3)})-\frac{1}{6}P_6',\]
  we obtain 
  \begin{equation}
    |G_0 (C-P_6^{''}) G_0(n,n')|\lesssim_{d,\alpha} \kappa^6\frac{1}{|n-n'|^{d-2}(|n|\wedge|n'|)^{6\alpha}}, 
  \end{equation}
 which  is the second conclusion in Theorem \ref{6,7 resolvent estimate} (1).
  \medskip

  (2) Now let's control  the refined remaining term $\boxed{G_0,i}_r$ for  $i=6,7$.\\

  When $i=6$, as  compared with the  initial $\boxed{G_0,6}$, 
  \begin{align}\label{refined 6 nonrandom term}
     \notag \boxed{G_0,6}_r= &\eqref{4.4}\sim \eqref{4.8}\\
          & + 4\boxed{ G_0 (C-P_6^{''}) G_0}-2\sigma^2(\boxed{G_0 v^2 W G_0}+\boxed{G_0 W v^2 G_0}). 
  \end{align}
  Using the relation
  \[\wtg W \wtg =G_0 W G_0 -\sigma G_0 W -\sigma WG_0 +\sigma^2 W\]
  shows 
  \begin{equation}\label{decay of wtg W wtg}
    |\wtg W \wtg(n,n')|\lesssim_{d,\alpha} \kappa^4 \frac{1}{|n-n'|^{d-2}(|n|\wedge|n'|)^{4\alpha}}. 
  \end{equation}
  Combining  \eqref{decay of wtg W wtg}  with the previous arguments  in  the estimates of $i$th order remaining terms for $i=1,2,3,4,5$ concludes the  $\E_p$   bound on  $\eqref{4.4}\sim\eqref{4.8}$ is just \eqref{moment estimate on 6,7 resolvent}. So it only  needs  to control the non-random term \eqref{refined 6 nonrandom term}. For this, by Theorem \ref{6,7 resolvent estimate} (1), we already have
  \begin{equation}
    |G_0 (C-P_6^{''}) G_0(n,n')|\lesssim_{d,\alpha} \kappa^6\frac{1}{|n-n'|^{d-2}(|n|\wedge|n'|)^{6\alpha}}. 
  \end{equation}
It  suffices  to estimate  $G_0(v^2 W +W v^2)G_0$ in \eqref{refined 6 nonrandom term}. Rewrite it as 
  \[ G_0 (v^2 W +W v^2)G_0= G_0(v^4 M v^2 +v^2 M v^4)G_0.\]
  By the symmetric difference 
  \[ v_{n_1}^4 v_{n_2}^2+v^2_{n_1} v^4_{n_2}=v^6_{n_1}+v^6_{n_2}-(v^4_{n_1}-v^4_{n_2})(v_{n_1}^2-v_{n_2}^2),\]
 we have 
  \begin{equation}\label{Decompose of G_0 (v^2 W +W v^2)G_0}
    G_0 (v^2 W +W v^2)G_0 =(G_0 v^6 M G_0 +G_0 M v^6 G_0)-G_0 P_6 G_0. 
  \end{equation}  
  So  by Lemma \ref{difference property},
  \begin{align}
    |P_6(n_1,n_2)|&=|(v^4_{n_1}-v^4_{n_2})(v_{n_1}^2-v_{n_2}^2) M(n_1,n_2)|\\
    \notag &\lesssim_{d,\alpha}\kappa^6 \frac{|n_1-n_2|^2}{(|n_1|+|n_2|)^2 (|n_1|\wedge|n_2|)^{6\alpha}|n_1-n_2|^{3(d-2)}}. 
  \end{align}
Similar to the proof of  \eqref{estimate on convolution part} and \eqref{estimate on difference part}, we obtain
  \[|G_0 v^6 M G_0(n,n')|\lesssim_{d,\alpha}\kappa^6 \frac{1}{|n-n'|^{d-2}(|n|\wedge|n'|)^{6\alpha}},\]
  \[|G_0  Mv^6 G_0(n,n')|\lesssim_{d,\alpha}\kappa^6 \frac{1}{|n-n'|^{d-2}(|n|\wedge|n'|)^{6\alpha}},\]
  \[|G_0 P_6 G_0(n,n')|\lesssim_{d,\alpha}\kappa^6 \frac{1}{|n-n'|^{d-2}(|n|\wedge|n'|)^{6\alpha}}.\]
  Thus by \eqref{Decompose of G_0 (v^2 W +W v^2)G_0}, we get  
  \begin{equation}
    | (G_0v^2 WG_0 +G_0W v^2G_0)(n,n')| \lesssim_{d,\alpha}\kappa^6 \frac{1}{|n-n'|^{d-2}(|n|\wedge|n'|)^{6\alpha}}. 
  \end{equation}
  Taking account of all above estimates concludes the  estimate \eqref{moment estimate on 6,7 resolvent} for $i=6$.\smallskip

  When $i=7$, 
 we have 
  \begin{align}
    \boxed{G_0,7}_r  =&[\eqref{4.11}\sim \eqref{4.21}] +[\eqref{4.23}\sim \eqref{4.24} ]\\
    \notag  &-4(G_0 C \wtg V G_0+G_0 V \wtg C G_0)+4G_0 VG_0P_6^{''}G_0+4G_0P_6^{''}G_0VG_0\\
    \notag =&[\eqref{4.11}\sim\eqref{4.21}] +[\eqref{4.23}\sim \eqref{4.24} ]\\
    \label{4.63} & +4\sigma(G_0 C V G_0 +G_0 V C G_0)-4[G_0 (C-P_6^{''})  G_0 V G_0+G_0 V G_0 (C-P_6^{''}) G_0]
  \end{align} 
  First, we explain why the $\E_p$ bounds  of terms $\eqref{4.11}\sim \eqref{4.20}$ and \eqref{4.23} can be controlled by   \eqref{moment estimate on 6,7 resolvent}. Indeed, we have 
  \begin{itemize}
    \item For the diagonal operators  $D_4,R_6,D_6^{(1)},D_6^{(2)}$, we obtain 
             \[|D_4(n_1)|\lesssim_{d,\alpha}\kappa^4 \frac{1}{|n_1|^{4\alpha}}\]
             and
             \begin{align*}
              |R_6(n_1)|&=|v^2_{n_1}\cdot (\wtg W \wtg)(n_1,n_1)|\\
                 &\lesssim \kappa^6 \frac{1}{|n_1|^{2\alpha}} \sum_{n_2,n_3\in \Z^d}\frac{1}{|n_1-n_2|^{d-2}|n_2|^{2\alpha}|n_2-n_3|^{3(d-2)}|n_3|^{2\alpha}|n_3-n_1|^{d-2}}\\
                 &\overset{{\rm by\  Lemma\  \ref{Lemma 2.2}}}{\lesssim_{d,\alpha }} \kappa^6 \frac{1}{|n_1|^{2\alpha}} \sum_{n_2\in \Z^d}\frac{1}{|n_1-n_2|^{2(d-2)}|n_2|^{2\alpha}(|n_1|\wedge|n_2|)^{2\alpha}} \\
                 &\lesssim_{d,\alpha}  \kappa^6 \frac{1}{|n_1|^{6\alpha}},
             \end{align*}
             and similarly,
             \[|D^{(1)}_6(n_1)|,|D^{(2)}_6(n_1)|\lesssim_{d,\alpha}\kappa^6 \frac{1}{|n_1|^{6\alpha}}.\]

    \item From  \eqref{decay of wtg W wtg}, we have 
            \begin{align*}
              &\ \ \ |G_0 W G_0(n,n')|,|G_0 W \wtg(n,n')|,|\wtg W G_0(n,n')|,|\wtg W \wtg(n,n')|\\
              &\lesssim_{d,\alpha} \kappa^4 \frac{1}{|n-n'|^{d-2}(|n|\wedge|n'|)^{4\alpha}}. 
        \end{align*}         
  \end{itemize}
  Again, using the  above facts  together with the previous decoupling arguments  in  the cases of orders $i=1,2,3,4,5$ shows 
  \[\E_p\left |   \bigg(\eqref{4.11}\sim \eqref{4.20} + \eqref{4.23} \bigg)(n,n')     \right | \lesssim_{d,p,\alpha}\kappa^7 \frac{1}{|n-n|^{d-2}(|n|\wedge|n'|)^{7\alpha}}.\]
  \\
Now,  let's handle the terms \eqref{4.21}, \eqref{4.63} and \eqref{4.24}. 
  For the modified terms \eqref{4.63}, on one hand,  we have  by Lemma \ref{Lemma 2.2},
  \begin{align*}
    |C(n_1,n_3)|&\lesssim \kappa^6 \frac{1}{|n_1|^{2\alpha}|n_1-n_3|^{d-2}|n_3|^{2\alpha}}\sum_{n_2\in \Z^d}\frac{1}{|n_1-n_2|^{2(d-2)}|n_2|^{2\alpha}|n_2-n_3|^{2(d-2)}}\\
              &\lesssim_{d,\alpha} \kappa^6 \frac{1}{|n_1-n_3|^{3(d-2)}|n_1|^{2\alpha}|n_3|^{2\alpha}(|n_1|\wedge|n_3|)^{2\alpha}}\\
              &\lesssim_{d,\alpha} \kappa^6 \frac{1}{|n_1-n_3|^{3(d-2)}(|n_1|\wedge|n_3|)^{6\alpha}}. 
  \end{align*}
  So the operator $G_0C$ (also $CG_0$) can be controlled again via Lemma \ref{Lemma 2.2}: 
  \begin{align*}
    |G_0C(n,n')|&\lesssim_{d,\alpha}\kappa^6 \sum_{n_1\in \Z^d}\frac{1}{|n-n_1|^{d-2}|n_1-n'|^{3(d-2)}(|n_1|\wedge|n'|)^{6\alpha}} \\
        &\lesssim_{d,\alpha}\kappa^6  \frac{1}{|n-n'|^{d-2}(|n|\wedge|n'|)^{6\alpha}}. 
  \end{align*}
  Hence, using   (decoupling) Lemma \ref{decoupling} yields  
  \begin{align*}
    \E_p\left | G_0CV G_0(n,n') \right | & \lesssim_{d,p,\alpha}\kappa^7 \bigg(    \sum_{n_1\in \Z^d} \frac{1}{|n-n_1|^{2(d-2)}(|n|\wedge|n_1|)^{12\alpha}|n_1|^{2\alpha}|n_1-n'|^{2(d-2)}} \bigg)\\
    &\lesssim_{d,\alpha}\kappa^7  \frac{1}{|n-n'|^{d-2}(|n|\wedge|n'|)^{7\alpha}}. 
  \end{align*}
  Similarly,  
  \begin{align*}
   \E_p \left | G_0V CG_0(n,n') \right |_p &\lesssim_{d,p,\alpha}\kappa^7  \frac{1}{|n-n'|^{d-2}(|n|\wedge|n'|)^{7\alpha}}. 
  \end{align*}
  On the other hand, by Theorem \ref{6,7 resolvent estimate} (1) and  Lemma \ref{decoupling}, we get 
  \[ \E_p\left | G_0(C-P_6^{''})G_0 V G_0(n,n') \right | \lesssim_{d,p,\alpha}\kappa^7  \frac{1}{|n-n'|^{d-2}(|n|\wedge|n'|)^{7\alpha}},\]
  \[ \E_p\left | G_0 V G_0(C-P_6^{''})G_0(n,n') \right |  \lesssim_{d,p, \alpha}\kappa^7  \frac{1}{|n-n'|^{d-2}(|n|\wedge|n'|)^{7\alpha}}.\]
  Thus we have established  the desired upper bound   for \eqref{4.63}.\\
 For term \eqref{4.21}, we first estimate  $G_0 v^2 M_4 v^2 V M_4 v^2 G_0$. We have by Lemma \ref{Lemma 2.2},
  \begin{align*}
    |G_0 v^2 M_4(n,n')|&\lesssim \kappa^2 \sum_{n_1\in \Z^d}\frac{1}{|n-n_1|^{d-2}|n_1|^{2\alpha}|n_1-n'|^{3(d-2)}}\\
        &\lesssim_{d,\alpha} \kappa^2 \frac{1}{|n-n'|^{d-2}(|n|\wedge|n'|)^{2\alpha}}, 
  \end{align*}
  and similarly, $|  M_4 v^2 G_0(n,n')|$ has the  the same estimate. Then by  Lemma \ref{decoupling}, we obtain 
  \[ \E_p\left | G_0 v^2 M_4 v^2 V M_4 v^2 G_0(n,n') \right | \lesssim_{d,p,\alpha}\kappa^7  \frac{1}{|n-n'|^{d-2}(|n|\wedge|n'|)^{7\alpha}}.\]
  Next, we estimate the second part $G_0 D_7 G_0$ of \eqref{4.21}.  {We remark  that $D_7$ is a random diagonal operator, so we cannot  renormalize it.} Recall that 
  \begin{equation*}
    D_7 = v^4_{n_1} \bigg[\sum_{n_2\in \Z^d} v^3_{n_2}\omega_{n_2} \wtg(n_1,n_2)^6 \bigg].
   \end{equation*}
 Using directly Lemma \ref{decoupling} shows 
   \begin{align*}
     \E_p\left | G_0 D_7 G_0(n,n')\right |
     = &\E_p\left | \sum_{n_1,n_2 \in\Z^d} \omega_{n_2} G_0(n,n_1) v_{n_1}^4 \wtg(n_1,n_2)^6 v^3_{n_2} G_0(n_1,n') \right |\\
          \lesssim_{p}& \left(   \sum_{n_2\in \Z^d} v_{n_2}^6 \big(   \sum_{n_1\in \Z^d}G_0(n,n_1) v_{n_1}^4 \wtg(n_1,n_2)^6 G_0(n_1,n')   \big)^2 \right)^{\frac{1}{2}}\\
          \lesssim_{p} &\kappa^7 \left(   \sum_{n_2\in \Z^d} \frac{1}{|n_2|^{6\alpha}} \big(   \sum_{n_1\in \Z^d}\frac{1}{|n-n_1|^{d-2}|n_1|^{4\alpha}|n_1-n_2|^{6(d-2)}|n_1-n'|^{d-2}}  \big)^2 \right)^{\frac{1}{2}}\\
          \lesssim_{p} &\kappa^7 \biggl(   \sum_{n_2\in \Z^d} \frac{1}{|n_2|^{6\alpha}} \big(   \sum_{n_1\in \Z^d}\frac{1}{|n-n_1|^{2(d-2)}|n_1|^{4\alpha}|n_1-n_2|^{6(d-2)}}  \big) \\
         &  \cdot \big(   \sum_{n_1\in \Z^d}\frac{1}{|n'-n_1|^{2(d-2)}|n_1|^{4\alpha}|n_1-n_2|^{6(d-2)}}  \big) \biggl)^{\frac12}\\
          \lesssim_{d,p,\alpha} &\kappa^7 \left( \sum_{n_2\in \Z^d} \frac{1}{|n_2|^{6\alpha}} \cdot \frac{1}{|n-n_2|^{2(d-2)}(|n|\wedge|n_2|)^{4\alpha}}\cdot \frac{1}{|n'-n_2|^{2(d-2)}(|n'|\wedge|n_2|)^{4\alpha}}  \right)^{\frac{1}{2}}\\
          \lesssim_{d,p,\alpha}& \kappa^7 \frac{1}{|n-n'|^{d-2}(|n|\wedge|n'|)^{7\alpha}},
   \end{align*}
   where for the third inequality above, we apply the Cauchy-Schwarz inequality,  and for the fourth and fifth inequalities  we use  Lemma \ref{Lemma 2.2}. Thus, putting  all above  estimates  together concludes the  desired bound on  \eqref{4.21}.\\
 For the each  term  in \eqref{4.24},  {it cannot be written as  a summation about admissible tuples.}  Similar to  the estimate on  $G_0 D_7 G_0$, we can directly use the  decoupling Lemma \ref{decoupling}. For example,  
   \begin{align*}
   &\ \ \  \E_p\left | G_0 V S G_0 (n,n') \right |\\
    =&\E_p \left |  \sum_{n_1,n_2,n_3\in \Z^d}  G_0(n,n_1) \omega_{n_1} v^3_{n_1} v^2_{n_2} v^2_{n_3}  \wtg(n_1,n_3)^2  \wtg(n_1,n_2)^3 \wtg(n_2,n_3) G_0 (n_3,n')         \right | \\
    \lesssim_p& \left(\sum_{n_1\in \Z^d} G_0 (n,n_1)^2 v^6_{n_1} \big( \sum_{n_2,n_3\in \Z^d} G_0(n_3,n')v^2_{n_2}v^2_{n_3}\wtg(n_2,n_3)\wtg(n_1,n_3)^2\wtg(n_1,n_2)^3 \big)^2 \right)^{\frac{1}{2}}\\
    \lesssim_p& \left(\sum_{n_1\in \Z^d} G_0 (n,n_1)^2 v^6_{n_1} \big[ \sum_{n_2\in \Z^d} \big(\sum_{n_3\in \Z^d}G_0(n_3,n')v^2_{n_3}\wtg(n_2,n_3)\wtg(n_1,n_3)^2 \big) v^2_{n_2}\wtg(n_1,n_2)^3 \big]^2 \right)^{\frac{1}{2}}\\
    \lesssim_p& \kappa^7 \biggl(\sum_{n_1\in \Z^d} \frac{1}{|n-n_1|^{2(d-2)}|n_1|^{6\alpha}}\\
    & \cdot \big[ \sum_{n_2\in \Z^d} \big(\sum_{n_3\in \Z^d}\frac{1}{|n_3-n'|^{d-2}|n_3|^{2\alpha}|n_2-n_3|^{d-2}|n_1-n_3|^{2(d-2)}} \big) \frac{1}{|n_2|^{2\alpha}|n_2-n_1|^{3(d-2)}} \big]^2 \biggl)^{\frac{1}{2}}.
  \end{align*}
   Applying  the  Cauchy-Schwarz inequality  for the summation about $n_3$ implies   (again by Lemma \ref{Lemma 2.2})
   \begin{align*}
    &\ \ \ \sum_{n_3\in \Z^d}\frac{1}{|n_3-n'|^{d-2}|n_3|^{2\alpha}|n_2-n_3|^{d-2}|n_1-n_3|^{2(d-2)}}\\
    &\leq (\sum_{n_3\in \Z^d}\frac{1}{|n_3-n'|^{2(d-2)}|n_3|^{2\alpha}|n_3-n_1|^{2(d-2)}})^{\frac{1}{2}} \cdot (\sum_{n_3\in \Z^d}\frac{1}{|n_3-n_2|^{2(d-2)}|n_3|^{2\alpha}|n_3-n_1|^{2(d-2)}})^{\frac{1}{2}}\\
    &\lesssim_{d,\alpha} \frac{1}{|n_1-n'|^{d-2}(|n_1|\wedge|n'|)^{\alpha}} \cdot  \frac{1}{|n_1-n_2|^{d-2}(|n_1|\wedge|n_2|)^{\alpha}}.
   \end{align*}
   This enables  us to continue the estimate: 
   \begin{align*}
    \ \ \ &\E_p\left | G_0 V S G_0 (n,n') \right | \\
 \lesssim_p&   \kappa^7 \biggl(\sum_{n_1\in Z^d} \frac{1}{|n-n_1|^{2(d-2)}|n_1|^{6\alpha}}\\
    & \cdot \big[ \sum_{n_2\in \Z^d} \frac{1}{|n_1-n'|^{d-2}(|n_1|\wedge|n'|)^{\alpha}} \cdot  \frac{1}{|n_1-n_2|^{d-2}(|n_1|\wedge|n_2|)^{\alpha}}\cdot \frac{1}{|n_2|^{2\alpha}|n_2-n_1|^{3(d-2)}} \big]^2 \biggl)^{\frac{1}{2}}\\
    =\kappa^7& \biggl( \sum_{n_1\in \Z^d} \frac{1}{|n-n_1|^{2(d-2)}|n_1-n'|^{2(d-2)}|n_1|^{6\alpha}(|n_1|\wedge|n'|)^{2\alpha}} \\
    & \cdot [\sum_{n_2\in \Z^d} \frac{1}{|n_2-n_1|^{4(d-2)}|n_2|^{2\alpha}(|n_2|\wedge|n_1|)^{\alpha}}]^2 \biggl)^{\frac{1}{2}} \\
    \overset{{\rm by\ Lemma\  \ref{Lemma 2.1}}}{\lesssim_{d,p,\alpha}}& \kappa^7 \biggl( \sum_{n_1\in \Z^d} \frac{1}{|n-n_1|^{2(d-2)}|n_1-n'|^{2(d-2)}|n_1|^{6\alpha}(|n_1|\wedge|n'|)^{2\alpha}} \cdot \frac{1}{|n_1|^{6\alpha}}\biggl)^{\frac{1}{2}}\\
    \overset{{\rm by\ Lemma\  \ref{Lemma 2.2}}}{\lesssim_{d,p,\alpha}}& \kappa^7 \frac{1}{|n-n'|^{d-2}(|n|\wedge|n'|)^{7\alpha}}. 
   \end{align*}
   Similarly, $G_0 S^{\top} V G_0$ has  the same estimate. Thus, we get the upper bound on  \eqref{4.24}.\\
   Combining all the above  estimates concludes the estimate \eqref{moment estimate on 6,7 resolvent} for $i=7$.
\end{proof}

\subsection{Rearrangement of  \eqref{resolvent}.}
Before we estimate the resolvent \eqref{resolvent}, we need to rearrange the decomposition. The main aim of the rearrangement is to remove the {\it singular part}  (i.e., $G_0 P_6''$) from $A$ to $B$. We will use the iteration technique in the proof of Lemma \ref{8thlem}. 

Recall that 
\begin{align*}
\boxed{G_0,6}_r&=\boxed{G_0,6}-4G_0P_6^{''}G_0,\\
\boxed{G_0,7}_r&=\boxed{G_0,6}+4G_0 VG_0P_6^{''}G_0+4G_0P_6^{''}G_0VG_0.
\end{align*}
Denote again by  $\boxed{G,6}_r$ (resp. $\boxed{G,7}_r$) the terms with the first $G_0$ replaced by $G$ in $\boxed{G_0,6}_r$ (resp. $\boxed{G_0,7}_r$). 
Repeatedly  using \eqref{D.7} and \eqref{D.3}   (start  with the $5$th order expansion) leads to  
\begin{align*}
  G = &\sum_{i=0}^{5} \boxed{G_0,i} + \boxed{G,6} \\
    &-G \Delta_2 V \boxed{G_0,5} \uwave{\text{$-G \Delta_4 V\boxed{G_0,5}--G \Delta_6 V\boxed{G_0,5}$}}\\
    & \uwave{\text{$-G \Delta_4 V\boxed{G_0,4}-G \Delta_6 V\boxed{G_0,4}$}}\\
    & -G \Delta_4 V \boxed{G_0,3} \uwave{\text{$-G \Delta_6 V\boxed{G_0,3}-G \Delta_6 V\boxed{G_0,2}$}}-G \Delta_6 V \boxed{G_0,1}\\
    = &\sum_{i=0}^{5} \boxed{G_0,i} + \boxed{G,6}_r + 4G P_6^{''}G_0 \\
    &    -G \Delta_2 V \boxed{G_0,5}-G \Delta_4 V \boxed{G_0,3}-G \Delta_6 V \boxed{G_0,1}\\
& \uwave{\text{ $ -G\Delta_6 V \sum_{i=2}^{5}\boxed{G_0,i}-G\Delta_4 V \sum_{i=4}^{5}\boxed{G_0,i}$}}\\
   =&\cdots\\
    =&\sum_{i=0}^{5} \boxed{G_0,i} + \boxed{G_0,6}_r +\boxed{G_0,7}_r + 4G P_6^{''}G_0 -4GP_6^{''}G_0VG_0 \\
    & \uwave{\text{ $-V\boxed{G_0,7}_r-G\Delta _2V (\boxed{G_0,6}_r+\boxed{G_0,7}_r)  -G\Delta_6 V (\sum_{i=2}^{5}\boxed{G_0,i}+\boxed{G_0,6}_r+\boxed{G_0,7}_r)$}}\\
    & \uwave{\text{$-G\Delta_4 V (\sum_{i=4}^{5}\boxed{G_0,i}+\boxed{G_0,6}_r +\boxed{G_0,7}_r)$}}.
\end{align*}
Thus, we rearrange the decomposition \eqref{resolvent} as 
\begin{equation}\label{rearrange 1}
  G= A'+GB'
\end{equation}
with 
\begin{align}
 \label{A'} A'=&\sum_{i=0}^{5} \boxed{G_0,i} + \boxed{G_0,6}_r +\boxed{G_0,7}_r,\\ 
\label{B'} B'= &4P_6^{''}G_0 -4P_6^{''}G_0VG_0 -V\boxed{G_0,7}_r-\Delta _2V (\boxed{G_0,6}_r+\boxed{G_0,7}_r) \\
  \notag &-\Delta_4 V (\sum_{i=4}^{5}\boxed{G_0,i}+\boxed{G_0,6}_r +\boxed{G_0,7}_r)\\
  \notag &-\Delta_6 V (\sum_{i=2}^{5}\boxed{G_0,i}+\boxed{G_0,6}_r+\boxed{G_0,7}_r).
\end{align}

\subsection{Green's function estimates}
In this subsection, we will finish the proof of Theorem \ref{greenthm}.

\begin{proof}[Proof of Theorem \ref{greenthm}]
We aim to control   \eqref{rearrange 1}.  From  $\eqref{A'}$, Theorem \ref{0-5 remaining term expactation} and Theorem \ref{6,7 resolvent estimate},  it follows that (the first term $G_0$ in $A'$ is deterministic) 
\begin{equation}\label{moment A'}
  \E_p|(A'-G_0)(n,n')|   \lesssim_{d,p,\alpha}\sum_{i=1}^{7} \kappa^i \frac{1}{|n-n'|^{d-2}(|n|\wedge|n'|)^{i\alpha}}\lesssim_{d,p,\alpha}\kappa \frac{1}{|n-n'|^{d-2}(|n|\wedge|n'|)^{\alpha}}. 
\end{equation}
Moreover, for the estimate of $B'$,  we first apply Theorem \ref{6,7 resolvent estimate} (1)  to get 
\begin{align*}
  |P_6^{''}G_0(n,n')| & \lesssim_{d,\alpha} \kappa^6 \sum_{n_1\in \Z^d}\frac{1}{|n-n_1|^{3d-7}(|n|\wedge|n_1|)^{1+6\alpha}|n_1-n'|^{d-2}}\\
         &\lesssim_{d,\alpha} \kappa^6 \frac{1}{|n-n'|^{d-2}(|n|\wedge|n'|)^{1+6\alpha}}\ ({\rm by\  Lemma\  \ref{Lemma 2.2}})\\
         &\lesssim_{d,\alpha} \kappa^6 \frac{1}{|n-n'|^{d-2}(|n|\wedge|n'|)^{8\alpha}},
\end{align*}
where for the last inequality,  we  use the fact   $d\geq5, \frac{1}{4}<\alpha\leq \frac{1}{3}\Rightarrow d-2\geq 1+6\alpha >8\alpha$. Then by  Lemma \ref{decoupling} and Lemma \ref{Lemma 2.2}, 
\begin{align*}
  \E_p\left |  P_6^{''}G_0 V G_0(n,n')  \right |&\lesssim_p  \left( \sum_{n_1\in \Z^d} (P^{''}_6G_0(n,n_1))^2 v^2_{n_1} G_0(n_1,n')^2  \right)^{\frac{1}{2}}\\
   &\lesssim_{d,p,\alpha} \kappa^7 \left(  \sum_{n_1\in \Z^d}  \frac{1}{|n-n_1|^{2(d-2)}(|n|\wedge|n_1|)^{16\alpha}|n_1|^{2\alpha}|n_1-n'|^{2(d-2)}}     \right)^{\frac{1}{2}}\\
   &\lesssim_{d,p,\alpha} \kappa^7 \frac{1}{|n-n'|^{d-2}(|n|\wedge|n'|)^{9\alpha}}.
\end{align*} 
This  together with Theorem \ref{0-5 remaining term expactation} and Theorem \ref{6,7 resolvent estimate} (2) implies, for example, 
\begin{align*}
  &\ \ \ \E_p\left | \Delta_2 V  (\boxed{G_0,6}_r+\boxed{G_0,7}_r)(n,n') \right | \\
  &\lesssim_{d,p,\alpha} \kappa^2\frac{1}{|n|^{2\alpha}}\left(\kappa^6\frac{1}{|n-n'|^{d-2}(|n|\wedge|n'|)^{6\alpha}}+\kappa^7\frac{1}{|n-n'|^{d-2}(|n|\wedge|n'|)^{7\alpha}}\right) \\
  &\lesssim_{d,p,\alpha} \kappa^8\frac{1}{|n-n'|^{d-2}(|n|\wedge|n'|)^{8\alpha}}.
\end{align*}
Combining all the above estimates yields 
\begin{equation}\label{moment B'}
  \E_p|B'(n,n')| \lesssim_{d,p,\alpha} \kappa^6 \frac{1}{|n-n'|^{d-2}(|n|\wedge|n'|)^{8\alpha}} \leq \kappa^2 \frac{1}{|n-n'|^{d-2}(|n|\wedge|n'|)^{8\alpha}},
\end{equation}
where for the second inequality, we require  that $\kappa$ is sufficiently small:  $0<\kappa\leq c(d,p,\alpha) \ll 1$. \smallskip

Next, we apply the  Chebyshev's inequality to control  $B'(n,n')$  provided  some $\omega$ were removed.  Since  $\frac{1}{4}<\alpha\leq \frac{1}{3}$, we  choose a small  $\varepsilon $ such that $$0<100\varepsilon<8\alpha-2<1.$$ Using  \eqref{moment B'} together with the Chebyshev's inequality concludes that for any fixed $n,n'\in \Z^d,$
\begin{align}\label{4.69}
  &\ \ \ \mathbb{P}\left(|B'(n,n')|>\kappa \frac{(|n|\vee |n'|)^{\varepsilon}}{|n-n'|^{d-2}(|n|\wedge|n'|)^{8\alpha}}\right) \\
 \notag &\leq \left(\kappa \frac{(|n|\vee |n'|)^{\varepsilon}}{|n-n'|^{d-2}(|n|\wedge|n'|)^{8\alpha}}\right)^{-p} \E|B(n,n')|^p\\
 \notag &\leq \left(\kappa \frac{1}{(|n|\vee|n'|)^{\varepsilon}}\right)^p.
\end{align}
Hence,
\begin{align}\label{Omega1}
\mathbb{P}(\Omega^{(1)})\geq 1-\kappa^p\sum_{n,n'\in \Z^d} \frac{1}{(|n|\vee|n'|)^{p\varepsilon}}, 
\end{align}
where 
\begin{align*}
\Omega^{(1)}&:=\bigcap_{n, n'\in\Z^d}\Omega^{(1)}_{n,n'},\\
 \Omega^{(1)}_{n,n'}&:= \left\{\omega\in\{\pm 1\}^{\Z^d}:\ |B'(n,n')|\leq \kappa \frac{(|n|\vee |n'|)^{\varepsilon}}{|n-n'|^{d-2}(|n|\wedge|n'|)^{8\alpha}}\right\}. 
\end{align*}
If  we choose $p$ sufficiently large such that $p\varepsilon>2d+2$, then 
\begin{align*}
  \sum_{n,n'\in \Z^d} \frac{1}{(|n|\vee|n'|)^{p\varepsilon}} &\leq \sum_{n,n'\in \Z^d} \frac{1}{(|n|\vee|n'|)^{2d+2}}\\
   &\leq  \sum_{n,n'\in \Z^d} \frac{1}{|n|^{d+1}|n'|^{d+1}}<\infty.
\end{align*}
Hence, with high probability (i.e., $\mathbb{P}(\{\pm 1\}^{\Z^d}\setminus\Omega^{(1)})\lesssim_d\kappa^p$), we have 
\begin{align*}
  |B'(n,n')| &\leq \kappa \frac{(|n|\vee |n'|)^{\varepsilon}}{|n-n'|^{d-2}(|n|\wedge|n'|)^{8\alpha}}\leq \kappa \frac{(|n|\wedge |n'|+|n-n'|)^{\varepsilon}}{|n-n'|^{d-2}(|n|\wedge|n'|)^{8\alpha}}.
\end{align*}
From  $2ab\geq  ab+1 \geq a+b,a,b\in \Z_+$,  it follows that 
\begin{equation} \label{Bou trick 1}
  2(|n_1|\wedge|n_2|)|n_1-n_2|\geq (|n_1|\wedge|n_2|) +|n_1-n_2|,
\end{equation}
which implies 
\begin{align}\label{B' high prob}
  |B'(n,n')| &\leq 2\kappa \frac{1}{|n-n'|^{d-2-\varepsilon}(|n|\wedge|n'|)^{8\alpha-\varepsilon} }\ {\rm for} \ \forall n,n'\in \Z^d.
\end{align} 
So, with high probability, \eqref{B' high prob} holds. Based on  this fact, we can  show that for $0<\kappa\ll1,$
\begin{equation}\label{Neumann}
  (I-B')^{-1} = I+\sum_{i=1}^{\infty} (B')^i = I+\widetilde{B'},
\end{equation} 
where $I$ denotes the identity operator. 
Indeed, using  \eqref{B' high prob}  yields for $\omega\notin\Omega^{(1)},$
\begin{align*}
  |(B')^2(n,n')|\leq& (2\kappa)^2 \sum_{n_1\in \Z^d}\frac{1}{|n-n_1|^{d-2-\varepsilon}(|n|\wedge|n_1|)^{8\alpha-\varepsilon}}\cdot \frac{1}{|n_1-n'|^{d-2-\varepsilon}(|n_1|\wedge|n'|)^{8\alpha-\varepsilon}}\\
   \leq& (2\kappa)^2 \biggl( \sum_{n_1\in \Z^d} \frac{1}{|n|^{8\alpha-\varepsilon}}  \frac{1}{|n'|^{8\alpha-\varepsilon}}  \frac{1}{|n-n_1|^{d-2-\varepsilon}|n_1-n'|^{d-2-\varepsilon}}\\
   &  + \sum_{n_1\in \Z^d} \frac{1}{|n|^{8\alpha-\varepsilon}}    \frac{1}{|n-n_1|^{d-2-\varepsilon}|n_1|^{8\alpha-\varepsilon}|n_1-n'|^{d-2-\varepsilon}}\\ 
   &   + \sum_{n_1\in \Z^d} \frac{1}{|n'|^{8\alpha-\varepsilon}}    \frac{1}{|n-n_1|^{d-2-\varepsilon}|n_1|^{8\alpha-\varepsilon}|n_1-n'|^{d-2-\varepsilon}}\\ 
   &  + \sum_{n_1\in \Z^d}    \frac{1}{|n-n_1|^{d-2-\varepsilon}|n_1|^{(8\alpha-\varepsilon)+(2+\varepsilon)}|n_1-n'|^{d-2-\varepsilon}}\biggr)\ ({\rm since }\ 8\alpha-2-2\varepsilon>0)\\
 \overset{{\rm by\ Lemmas\ \ref{Lemma 2.1}, \ref{Lemma 2.2}}}{\lesssim_{d,\alpha}}  & (2\kappa)^2 \biggl( \frac{1}{|n|^{8\alpha-\varepsilon}}  \frac{1}{|n'|^{8\alpha-\varepsilon}} \frac{1}{|n-n'|^{d-4-2\varepsilon}}\\
   & \qquad\qquad + \frac{1}{|n-n'|^{d-2-\varepsilon}(|n|\wedge|n'|)^{8\alpha-\varepsilon}}\biggr).
\end{align*}
Moreover, we have 
\begin{align*}
  \frac{1}{|n|^{8\alpha-\varepsilon}}  \frac{1}{|n'|^{8\alpha-\varepsilon}} \frac{1}{|n-n'|^{d-4-2\varepsilon}}&=\frac{|n-n'|^{2+\varepsilon}}{(|n|\vee|n'|)^{8\alpha-\varepsilon}}\frac{1}{|n-n'|^{d-2-\varepsilon}(|n|\wedge|n'|)^{8\alpha-\varepsilon}}\\
  &\leq \left(\frac{|n|+|n'|}{|n|\vee|n'|}\right)^{2+\varepsilon} \frac{1}{(|n|\vee|n'|)^{8\alpha-2-2\varepsilon}} \frac{1}{|n-n'|^{d-2-\varepsilon}(|n|\wedge|n'|)^{8\alpha-\varepsilon}}\\
  &\leq 2^{2+\varepsilon} \frac{1}{|n-n'|^{d-2-\varepsilon}(|n|\wedge|n'|)^{8\alpha-\varepsilon}}.
\end{align*}
Therefore, we obtain for some constant $f(d,\alpha)>0$ depending only on $d, \alpha$ that 
\begin{equation}\label{4.74}
  |(B')^2(n,n')| \leq (f(d,\alpha)\kappa)^2 \frac{1}{|n-n'|^{d-2-\varepsilon}(|n|\wedge|n'|)^{8\alpha-\varepsilon}}. 
\end{equation}
Iterating the estimate leading to  \eqref{4.74} shows 
\begin{equation}\label{4.75}
  |(B')^i(n,n')| \leq (f(d,\alpha) \kappa)^i \frac{1}{|n-n'|^{d-2-\varepsilon}(|n|\wedge|n'|)^{8\alpha-\varepsilon}}. 
\end{equation} 
This  implies  that if $0<\kappa<c(d,\alpha)\ll 1$, then we have  
\begin{align}\label{4.76}
  |\widetilde{B'}(n,n')|  &\leq \left(\sum_{i=1}^{\infty} (f(d,\alpha)\kappa)^i \right) \frac{1}{|n-n'|^{d-2-\varepsilon}(|n|\wedge|n'|)^{8\alpha-\varepsilon}}\\
  \notag & \lesssim_{d,\alpha} \kappa \frac{1}{|n-n'|^{d-2-\varepsilon}(|n|\wedge|n'|)^{8\alpha-\varepsilon}}.
\end{align}
As a result, $(I-B)^{-1}$ given by \eqref{Neumann} is well defined, while it is not a bounded linear operator on $\ell^2(\Z^d)$.

Now we deal with the operator $A'$.  Similarly, we define 
\begin{align*}
\Omega^{(2)}&:=\bigcap_{n, n'\in\Z^d}\Omega^{(2)}_{n,n'},\\
 \Omega^{(2)}_{n,n'}&:= \left\{\omega\in\{\pm 1\}^{\Z^d}:\ |A'(n,n')-G_0(n,n')|\leq   \frac{(|n|\vee |n'|)^{\varepsilon}}{|n-n'|^{d-2}(|n|\wedge|n'|)^{\alpha}}\right\}. 
\end{align*}
Again, by the Chebyshev's  inequality and \eqref{moment A'},  we obtain 
\begin{align*}
  \mathbb P(\Omega^{(2)})\geq 1-(f(d,\alpha)\kappa)^p \sum_{n,n'\in \Z^d} \frac{1}{(|n|\vee|n'|)^{p\varepsilon}}. 
\end{align*}
Now by $p\varepsilon>2d+2$, we get 
\begin{align*}
  \mathbb P(\{\pm 1\}^{\Z^d}\setminus\Omega^{(2)})\lesssim_{d,\alpha} \kappa^p.  
\end{align*}
and for $\omega\in\Omega^{(2)},$
\begin{align*}
  |A'(n,n')| &\leq \frac{1}{|n-n'|^{d-2}}+ \frac{(|n|\vee |n'|)^{\varepsilon}}{|n-n'|^{d-2}(|n|\wedge|n'|)^{\alpha}}\ {\rm for}\ \forall n,n'\in \Z^d.
\end{align*} 
Similar to the proof of  \eqref{Bou trick 1}, we have
\begin{align}\label{A' high prob}
  |A'(n,n')| &\lesssim \frac{1}{|n-n'|^{d-2-\varepsilon}}\ {\rm for}\  \forall n,n'\in \Z^d. 
\end{align} 
Hence,  for
\begin{align}
\omega\in\Omega:=\Omega^{(1)}\cap\Omega^{(2)},
\end{align}
we have  by \eqref{A' high prob} and \eqref{4.76}  that 
\begin{align*}
  |A' \widetilde{B'}(n,n')|\lesssim_{d,\alpha}&\kappa \sum_{n_1\in \Z^d} \frac{1}{|n-n_1|^{d-2-\varepsilon}}\cdot \frac{1}{|n_1-n'|^{d-2-\varepsilon}(|n_1|\wedge|n'|)^{8\alpha-\varepsilon}}\\
    \lesssim_{d,\alpha}&\kappa \bigg(\frac{1}{|n'|^{8\alpha-\varepsilon}}\sum_{n_1}\frac{1}{|n-n_1|^{d-2-\varepsilon}|n_1-n'|^{d-2-\varepsilon}}\\
     & + \sum_{n_1\in \Z^d}\frac{1}{|n-n_1|^{d-2-\varepsilon}|n_1|^{8\alpha-\varepsilon}|n_1-n'|^{d-2-\varepsilon}}      \bigg)\\
     \overset{{\rm br\ Lemmas\ \ref{Lemma 2.1},\ \ref{Lemma 2.2}}}{\lesssim_{d,\alpha}}& \kappa \frac{1}{|n'|^{8\alpha-\varepsilon}} \frac{1}{|n-n'|^{d-4-2\varepsilon}} +\kappa \frac{1}{|n-n'|^{d-2-\varepsilon}(|n|\wedge|n'|)^{8\alpha-2-2\varepsilon}}\\
      \lesssim_{d,\alpha}& \kappa \frac{1}{|n'|^{8\alpha-\varepsilon}} \frac{1}{|n-n'|^{d-4-2\varepsilon}}  +\kappa \frac{1}{|n-n'|^{d-2-\varepsilon}}\ ({\rm since}\ 8\alpha-2-2\varepsilon>0).
\end{align*}
From   \eqref{rearrange 1},    it follows that   for $\omega\in\Omega$ and $n,n'\in \Z^d,$
\begin{align*}
  |G(n,n')|& \leq |A'(n,n')|+|A'\widetilde{B'}(n,n')|\\
  &\lesssim \frac{1}{|n-n'|^{d-2-\varepsilon}} + \frac{1}{|n'|^{8\alpha-\varepsilon}} \frac{1}{|n-n'|^{d-4-2\varepsilon}}.
\end{align*}
Moreover, since $G$ is self-adjoint, we get  
\begin{align*}
  |G(n,n')|
  &\lesssim \frac{1}{|n-n'|^{d-2-\varepsilon}} + \frac{1}{(|n|\vee|n'|)^{8\alpha-\varepsilon}} \frac{|n-n'|^{8\alpha-\varepsilon}}{|n-n'|^{d-4+8\alpha-3\varepsilon}}\\
  &\lesssim\frac{1}{|n-n'|^{d-2-\varepsilon}} +  2^{8\alpha-\varepsilon} \frac{1}{|n-n'|^{d-2+8\alpha-3\varepsilon}} \\
      &\lesssim \frac{1}{|n-n'|^{d-2-\varepsilon}}.
\end{align*}
This concludes the proof of Theorem \ref{greenthm}. 
\end{proof}

\section{Construction of extended states: Proof of Theorem \ref{extthm}}\label{thm2sec}

In this section, we construct extended states for the renormalized operator H, thereby completing the proof of Theorem \ref{extthm}.

While we employ the  perturbation lemma  (cf. \cite[Lemma 1.2]{Bou03})  to  control  $(-\Delta+W)^{-1}$ originating from the $4$th-order renormalization, we cannot  apply  this lemma  to handle the Green's function of $-\Delta+C$ coming  from the  6th-order renormalization. This is because the operator $C$ does not have the symmetry form required in  \cite[Lemma 1.2]{Bou03}. Instead, we incorporate these 6th-order non-random terms into the decomposition $G = A'' + GB''$ via the second rearrangement.  

In the following, we first perform the rearrangement. Then, we construct the extended states via Green's function estimates together with the decoupling lemma (cf. Lemma \ref{decoupling}). 

\subsection{The second rearrangement of  \eqref{resolvent}}
We first  decompose the operators $C,v^2 W$ and $Wv^2$ as follows:
\begin{itemize}
   \item $v^2 W+W v^2 = 2M v^6 +Q_6^{(1)}$, with 
      \[Q_6^{(1)}(n_1,n_2)=[v^2_{n_2}(v^4_{n_1}-v^4_{n_2})+v^4_{n_2}(v^2_{n_1}-v^4_{n_2})]M(n_1,n_2).\]
      By Lemma \ref{difference property}, we obtain  
      \begin{align}\label{5.2}
        |Q_6^{(1)}(n_1,n_2)| & \lesssim_{d,\alpha} \kappa^6 \frac{1}{|n_1-n_2|^{3(d-2)-1}(|n_1|\wedge|n_2|)^{6\alpha+1}}. 
      \end{align}
  \item $C=\widetilde{N}v^6+Q_6^{(2)}$, with 
     \[Q_6^{(2)}(n_1,n_3)=-\frac{1}{6} P_6^{'}(n_1,n_3) +P_6^{''}(n_1,n_3)+\frac{2}{3}(v^6_{n_1}-v^6_{n_3}) \widetilde{N}(n_1,n_3),\]
     where   $\widetilde{N},P'_6$ and $P_6^{''}$  are given in   the proof of Lemma \ref{6,7 resolvent estimate} (1). By Lemma \ref{difference property}, \eqref{P6'} and \eqref{P6''},   we have 
     \begin{align}\label{5.3}
      |Q_6^{(2)}(n_1,n_2)| & \lesssim_{d,\alpha} \kappa^6 \frac{1}{|n_1-n_2|^{3(d-2)-2}(|n_1|\wedge|n_2|)^{6\alpha+1}}. 
    \end{align}
\end{itemize} 
Using the same notation as  in Subsection \ref{8thsub}, we denote 
\begin{align*}
\boxed{G_0,6}_e &=\boxed{G_0,6}+2\sigma^2 G_0 Q_6^{(1)} G_0-4 G_0 Q^{(2)}_6 G_0,\\
\boxed{G_0,7}_e &=\boxed{G_0,7}-2\sigma^2 G_0 V G_0 Q_6^{(1)}G_0 + 4G_0 V G_0 Q_6^{(2)}G_0 +4 G_0 Q_6^{(2)} G_0 V G_0,\\
\end{align*}
and
\begin{align*}
\boxed{G_0,4}_E&= \boxed{G_0,4}-G_0 W G_0 = \boxed{G_0,4}-\boxed{G_0,0} W G_0,\\
\boxed{G_0,5}_E&= \boxed{G_0,5}+G_0 V G_0 W G_0=\boxed{G_0,5}-\boxed{G_0,1} W G_0,\\
\boxed{G_0,6}_E &= \boxed{G_0,6}_e -(G_0VG_0VG_0)^{\admissible}W G_0=\boxed{G_0,6}_e-\boxed{G_0,2} W G_0,\\
  \boxed{G_0,7}_E &= \boxed{G_0,7}_e +(G_0V G_0 VG_0VG_0)^{\admissible}W G_0-\sigma^2 G_0v^2 V G_0 W G_0 \\
    & =\boxed{G_0,7}_e-\boxed{G_0,3} W G_0.
\end{align*}
With those modifications, the above terms in fact become:
\begin{itemize}
  \item \begin{align}\label{5.4}
        \boxed{G_0,4}_E = &-\sigma^2(G_0v^2VG_0 VG_0)^{(*)}-\sigma^2 (G_0 V G_0 v^2V G_0)^{\admissible} \\
          \notag & +(G_0 VG_0 VG_0 VG_0 VG_0)^{\admissible}.  
        \end{align}
  \item \begin{align}\label{5.5}
          \boxed{G_0,5}_E = & -2\sigma \rho G_0 v^4VG_0 +G_0 V D_4 G_0  +\sigma G_0 V  W G_0 -G_0 W\widetilde{G_0} V G_0 \\
          \notag & +\sigma^2 (G_0 v^2V G_0 V G_0 VG_0)^{\admissible} +\sigma^2 (G_0 V G_0 v^2 V G_0 V G_0)^{\admissible} \\
          \notag & +\sigma^2 (G_0 V G_0 V G_0 v^2V G_0)^{(*)} - (G_0 V G_0 V G_0 V G_0 V G_0 V G_0)^{\admissible}.
        \end{align}
  \item \begin{align}\label{5.6}
        \boxed{G_0,6}_E = & [\eqref{4.4}\sim \eqref{4.6}]\\
         \notag & + \big[   (G_0 W \widetilde{G_0} V G_0 V G_0)^{\admissible}  +(G_0 V \wtg W \wtg V G_0)^{\admissible} -\sigma (G_0 V G_0 V  W G_0)^{\admissible} \big]\\
         \notag & +\eqref{4.8}\\
         \notag & + 4\boxed{ G_0 \widetilde{N} v^6 G_0}-4\sigma^2\boxed{G_0 M v^6 G_0}.
        \end{align}
  \item \begin{align}\label{5.7}
        \boxed{G_0,7}_E =& [\eqref{4.11}\sim \eqref{4.14}] \\
         \notag & +\sigma^2(G_0 W \wtg v^2 V G_0-\sigma G_0 v^2 V W G_0 )\\
          \notag & +2\sigma G_0 (W v^2+v^2 W)\wtg  V G_0 - 2\sigma^3   G_0 V (W v^2+v^2 W)  G_0 +2\sigma^2 G_0 V G_0 M v^6 G_0 \\
          \notag & +[\eqref{4.17}\sim \eqref{4.18}] \\
          \notag & -\big[(G_0 W \wtg V G_0 V G_0 VG_0)^{\admissible}+ (G_0 V \wtg W \wtg V G_0 VG_0)^{\admissible} \\
          \notag &\qquad + (G_0 V G_0 V \wtg W \wtg VG_0)^{\admissible}-\sigma (G_0 V G_0 V G_0 V  WG_0)^{\admissible} \big]\\
          \notag & +[\eqref{4.20}\sim \eqref{4.21}]\\
          \notag & +4\sigma(G_0 CV G_0 + G_0 V C G_0)-4(G_0 \widetilde{N} v^6 G_0 V G_0 + G_0 V G_0 \widetilde{N} v^6 G_0 )\\
          \notag & +[\eqref{4.23}\sim \eqref{4.24}].
        \end{align}
\end{itemize}
By using the same argument  in Subsection \ref{8thsub}, we first rearrange (via repeatedly applying \eqref{D.3} and \eqref{D.7}) the ``bad terms''  involving $Q_6^{(1)}$ and $Q_6^{(2)}$ as: 
\begin{align*}
  G = &\sum_{i=0}^{5}\boxed{G_0,i} +\boxed{G,6} +\cdots \\
   =&\sum_{i=0}^{5}\boxed{G_0,i} +\boxed{G,6}_e +4 G Q_6^{(2)} G_0 -2\sigma^2  G Q_6^{(1)} G_0 -G\Delta V_2 \boxed{G_0,5}-\cdots \\
   \overset{{\rm by \ \eqref{D.3}}}{=}& \sum_{i=0}^{5}\boxed{G_0,i} +\boxed{G_0,6}_e  +4 G Q_6^{(2)} G_0 -2\sigma^2  G Q_6^{(1)} G_0  \\
   & -GV \boxed{G_0,6}_e -\uwave{\text{$G \Delta_2 V \boxed{G_0,6}_e$}} -\uwave{\text{$G \Delta_4 V \boxed{G_0,6}_e$}} -\uwave{\text{$G \Delta_6 V \boxed{G_0,6}_e$}} \\
   &   -G\Delta_2 V \boxed{G_0,5}-\cdots \\
   =&\cdots\\
 = &\sum_{i=0}^{5}\boxed{G_0,i}  +\boxed{G_0,6}_e + \boxed{G,7}_e \\
   &  +(4 G Q_6^{(2)} G_0  -2\sigma^2  G Q_6^{(1)} G_0 - 4 G Q_6^{(2)} G_0 V G_0) \\
   &  \uwave{\text{$-GV  \boxed{G_0,7} _e $}}- \uwave{\text{$G\Delta_2 V  (\boxed{G_0,7}_e + \boxed{G_0,6}_e)$}} \\
   &   \uwave{\text{$-G\Delta_4 V (\boxed{G_0,7}_e +\boxed{G_0,6}_e+ \boxed{G_0,5} +\boxed{G_0,4}) $}}  \uwave{\text{$-G\Delta_6 V (\boxed{G_0,7}_e +\boxed{G_0,6}_e +\sum_{i=2}^{5}\boxed{G_0,i}) $}}. 
\end{align*}
Next we rearrange the ``bad terms''  involving $W$. Continuing  replacing all boxed terms  with $\boxed{G_0,i}_E,\ i=4,5,6,7$ leads to  
\begin{align}
  \nonumber G =&\sum_{i=0}^{3}\boxed{G_0,i}+ \sum_{i=4}^{7}\boxed{G_0,i}_E \\
   \nonumber &  +(4 G Q_6^{(2)} G_0  -2\sigma^2  G Q_6^{(1)} G_0 - 4 G Q_6^{(2)} G_0 V G_0) \\
 \nonumber   & \uwave{\text{$-GV  \boxed{G_0,7} _E $}}- \uwave{\text{$G\Delta_2 V  (\boxed{G_0,7}_E + \boxed{G_0,6}_E)$}} \\
 \nonumber   &   \uwave{\text{$-G\Delta_4 V (\sum_{i=4}^{7}\boxed{G_0,i}_E) $}}  \uwave{\text{$-G\Delta_6 V  (\boxed{G_0,2}+\boxed{G_0,3}+\sum_{i=4}^{7}\boxed{G_0,i}_E) $}}\\
 \nonumber &+\bigg[\boxed{G_0,0}+\boxed{G_0,1}+\boxed{G_0,2}+\boxed{G_0,3} \\
 \nonumber   &  -GV\boxed{G_0,3}-G\Delta_2 V (\boxed{G_0,3}+\boxed{G_0,2})\\
   \nonumber & -G(\Delta_4 V+\Delta_6V)\cdot (\sum_{i=0}^{3}\boxed{G_0,i})\bigg] WG_0\\
 \label{rearrange 2} =&A''+GB''+GWG_0, 
\end{align}
where 
\begin{align}
 \label{A''} A^{''}= &\sum_{i=0}^{3}\boxed{G_0,i} +\sum_{i=4}^{7}\boxed{G_0,i}_E, \\
 \label{B''} B^{''}  =& (4  Q_6^{(2)} G_0  -2\sigma^2  Q_6^{(1)} G_0 - 4 Q_6^{(2)} G_0 V G_0) \\
  \notag & -GV  \boxed{G_0,7} _E-G\Delta_2 V  (\boxed{G_0,7}_E + \boxed{G_0,6}_E) \\
  \notag & -G\Delta_4 V (\sum_{i=4}^{7}\boxed{G_0,i}_E) -G\Delta_6 V(\boxed{G_0,2}+\boxed{G_0,3}+\sum_{i=4}^{7}\boxed{G_0,i}_E). 
\end{align}

\subsection{Construction of  extended states} 
In this subsection, we aim to construct the desired extended states. We first recall some arguments of Bourgain \cite{Bou03}. 
From  \eqref{rearrange 2} and Lemma 1.2 of \cite{Bou03}, we obtain
\begin{align*}
  G &= (A^{''}+GB^{''})(1-WG_0)^{-1}\\
  &=(A^{''}+GB^{''}) H_0 (H_0-W)^{-1}. 
\end{align*}
We denote 
$$H_0=-\Delta, \ H_0'=H_0-W,\ G_0'=(H_0')^{-1},$$
and hence,
\begin{equation}\label{rearrange 3}
  G=(A^{''}+GB^{''}) H_0 G_0'. 
\end{equation}
Indeed,  by Lemma 1.2 of \cite{Bou03} again, we have 
  \[|G_0'(n,n')|\lesssim_{d,\alpha} \frac{1}{|n-n'|^{d-2}}.\]
 With the rearrangement \eqref{rearrange 3}, we can construct the extended states  as follows.  
As in \cite{Bou03}, we denote $\hat{\delta_0}=\{\hat\delta_0(n)\equiv 1\}_{n\in\Z^d}\in \ell^{\infty}(\Z^d)$. 
Then  $H_0 \hat{\delta_0} = -\Delta \hat{\delta_0} =0$ and $\hat{\delta_0}$ is an extended state of $H_0$.  
 As in (4.6) of  \cite{Bou03}, starting from  $\hat{\delta_0}$ gives   the extended state   $\xi\in \ell^{\infty}(\Z^d)$  of $H_0'$, namely,  
\[H_0'\xi=0,\ \xi=\hat{\delta_0}+\mcO(\kappa)\ {\rm in}\ \ell^\infty(\Z^d).\]
It is important that $\xi$ is non-random. 
Moreover, as in  (4.11) of  \cite{Bou03},  starting from  $\xi$ gives 
\[\zeta = \xi-G(W+\widetilde{V})\xi,\ \ H\zeta =0.\]
By the resolvent identity, 
\begin{equation}\label{resolvent 2}
  G= G_0' -G(W+\widetilde{V}) G_0' \Rightarrow H= H_0' (1-G(W+\widetilde{V}))^{-1}. 
\end{equation} 
 Thus, we only need to prove that (with high probability) $\zeta =\xi+\mathcal O(\kappa)\  {\rm in}\  \ell^{\infty}(\Z^d)$.

From   \eqref{resolvent 2} and \eqref{rearrange 3}, we obtain 
\begin{align}\label{5.11}
  -G(W+\widetilde{V}) \xi =(G-G_0') H_0'  \xi 
  =(A^{''}-G_0)H_0\xi+ GB^{''}H_0\xi. 
\end{align}
We have 
\begin{thm}\label{Theorem 5.1}
  Let  $p\geq 1$. For $i=1,2,3$, we have  
  \begin{equation}\label{5.12}
      \E _p  \left |\boxed{G_0,i} H_0 \xi(n)\right | \lesssim_{d,p,\alpha} \kappa^i \frac{1}{|n|^{i\alpha}} .
  \end{equation}
 For  $i=4,5,6,7$, we have
  \begin{equation}\label{5.13}
      \E_p  \left |\boxed{G_0,i}_E H_0 \xi(n)\right | \lesssim_{d,p,\alpha}\kappa^i \frac{1}{|n|^{i\alpha}} .
  \end{equation}
\end{thm}
\begin{proof} [Proof of  Theorem \ref{Theorem 5.1}]
  When $i=1,2,3$,  we directly apply the decoupling lemma  (cf. Lemma \ref{decoupling}). For example,   we have 
  \begin{align*}
    &\ \ \  \E_p\left |  (G_0VG_0 VG_0 V\xi)^{\admissible}(n)  \right |  \\
     &=\kappa^3  \E_p\left |   \sum_{(m_1,m_2,m_3)}^{\admissible}\omega_{m_1}\omega_{m_2}\omega_{m_3} G_0(n,m_1)v_{m_1}G_0(m_1,m_2)v_{m_2}G_0(m_2,m_3)v_{m_3} \xi(m_3)  \right |\\
     &\lesssim_{p} \kappa^3 \left( \sum_{m_1,m_2,m_3\in\Z^d}\frac{1}{|n-m_1|^{2(d-2)}|m_1|^{2\alpha}|m_1-m_2|^{2(d-2)}|m_2|^{2\alpha}|m_2-m_3|^{2(d-2)}|m_3|^{2\alpha}} \right)^{\frac{1}{2}}\\
     \overset{{\rm Lemma\  \ref{Lemma 2.2}}}{\lesssim_{d,p,\alpha}}&\kappa^3 \left( \sum_{m_2,m_3\in\Z^d}\frac{1}{(|n|\wedge |m_2|)^{2\alpha}|n-m_2|^{2(d-2)}|m_2|^{2\alpha}|m_2-m_3|^{2(d-2)}|m_3|^{2\alpha}} \right)^{\frac{1}{2}}\\
     &\cdots \\
     &\lesssim_{d,p,\alpha} \kappa^3 \frac{1}{|n|^{3\alpha}}.
   \end{align*}
   Hence, we can prove \eqref{5.12} for $i=1,2,3$.

   When $i=4$,  from  \eqref{5.4}, applying Lemma \ref{decoupling} as in the case of  $i=1,2,3$  implies  that \eqref{5.13} holds for $i=4$.

   When $i=5$,  we first have 
   \begin{itemize}
    \item $|D_4(n)|\lesssim_{d,\alpha} \kappa^4\frac{1}{|n|^{4\alpha}}.$
    \item $|G_0 W \wtg (n,n')|\lesssim_{d,\alpha}\kappa^4 \frac{1}{|n-n'|^{d-2}(|n|\wedge|n'|)^{4\alpha}}$ and $|W(n,n')|\lesssim \kappa^4 \frac{1}{|n-n'|^{3(d-2)}(|n|\wedge|n'|)^{4\alpha}}$.  
   \end{itemize}
   So, using   Lemma \ref{decoupling} shows,  for example, 
   \begin{align}
    &\ \ \ \left\|  (G_0V W \xi)(n)  \right\|_p \\
    \notag&\lesssim_{d,p,\alpha} \kappa^5 \left(\sum_{n_1\in \Z^d}\frac{1}{|n-n_1|^{2(d-2)}|n_1|^{2\alpha}} (\sum_{n_2\in \Z^d}\frac{\xi(n_2)}{|n_1-n_2|^{3(d-2)}(|n_1|\wedge|n_2|)^{4\alpha}})^2\right) \\
    \notag &\lesssim_{d,p,\alpha} \kappa^5 \frac{1}{|n|^{5\alpha}}. 
   \end{align}
   Thus, we can prove \eqref{5.13} for $i=5$.

When $i=6$, by recalling  \eqref{5.6},  Lemma \ref{decoupling} and using  the two facts   in the case of  $i=5$, it suffices to deal with  $G_0 \widetilde{N} v^6\xi(n)$ and $G_0Mv^6 \xi(n)$. Since we have shown   \[|G_0 \widetilde{N}(n,n')|,|\widetilde{N} G_0(n,n')|\lesssim\frac{1}{|n-n'|^{(d+2)-}},\]
   then   by Lemma \ref{Lemma 2.1}, 
   \begin{align}
    |G_0 \widetilde{N} v^6\xi(n)|,\ |G_0 M v^6\xi(n)| & \lesssim_{d,\alpha} \sum_{n_1\in \Z^d} \frac{1}{|n-n_1|^{(d+2)-}|n_1|^{6\alpha}} \\
    \notag & {\lesssim_{d,\alpha}} \kappa^6 \frac{1}{|n|^{6\alpha}}. 
   \end{align}
   Hence, we have proven  \eqref{5.13} for $i=6$.

   When $i=7$,  we have the form of  \eqref{5.7}. Thanks to
   \begin{align*}
       &|G_0 (v^2 W+ W v^2)G_0(n,n')|,\ |G_0 (v^2 W+ W v^2) (n,n') |,\ |(v^2 W+ W v^2)G_0 (n,n')|,\\
       & |G_0C(n,n')|,\ |CG_0(n,n')|,\ |G_0 \widetilde{N}v^6 G_0(n,n')| \lesssim_{d,\alpha}\kappa^6 \frac{1}{|n-n'|^{d-2}(|n|\wedge|n'|)^{6\alpha}}, \\
   & |C(n,n')|,\ |\widetilde{N}v^6(n,n')|\lesssim_{d,\alpha} \kappa^6 \frac{1}{|n-n'|^{3(d-2)-2}(|n|\wedge|n'|)^{6\alpha}}, 
   \end{align*}
 and Lemma \ref{decoupling}, it suffices to  consider the terms $\eqref{4.21}\cdot H_0 \xi$ and $\eqref{4.24}\cdot H_0 \xi$. 
   For the term $(G_0 v^2 M_4 v^2 V M_4 v^2\xi)(n) $, we have shown in the proof of Theorem \ref{6,7 resolvent estimate}  that 
   \[|G_0 v^2 M_4(n,n')| \lesssim_{d,\alpha}\frac{1}{|n-n'|^{d-2}(|n|\wedge|n'|)^{2\alpha}}.\]
   Then applying  Lemma \ref{decoupling} yields 
   \begin{align*}
    &\ \ \ \E_p\left |  (G_0 v^2 M_4 v^2 V M_4 v^2\xi)(n)  \right | \\
    \notag &\lesssim_p \left(\sum_{n_1\in \Z^d} |G_0 v^2 M_4(n,n_1)|^2 v_{n_1}^6 \cdot \big|\sum_{n_2\in \Z^d} M_4v^2\xi(n_1,n_2) \big|^2 \right)^{\frac{1}{2}}\\
    \notag &\lesssim_{d,p,\alpha} \kappa^7 \left(\sum_{n_1\in \Z^d}\frac{1}{|n-n_1|^{2(d-2)} (|n|\wedge|n_1|)^{4\alpha}|n_1|^{10\alpha}}\right) \\
    \notag &\lesssim_{d,p,\alpha} \kappa^7 \frac{1}{|n|^{7\alpha}}\ ({\rm by\  Lemma\  \ref{Lemma 2.1}}). 
   \end{align*}
   For the term $G_0 D_7 \xi (n)$,  combining    Lemma \ref{decoupling},  Lemma \ref{Lemma 2.2} and Cauhcy-Schwarz inequality  gives  
   \begin{align*}
    \E_p\left | G_0 D_7 \xi (n)\right | &= \E_p\left | \sum_{n_1,n_2 \in\Z^d} \omega_{n_2} G_0(n,n_1) v_{n_1}^4 \wtg(n_1,n_2)^6 v^3_{n_2} \xi(n_2) \right |\\
         & \lesssim_{p} \left(   \sum_{n_2\in \Z^d} v_{n_2}^6 \big(   \sum_{n_1\in \Z^d}G_0(n,n_1) v_{n_1}^4 \wtg(n_1,n_2)^6   \big)^2 \right)^{\frac{1}{2}}\\
         & \lesssim_{p} \kappa^7 \left(   \sum_{n_2\in \Z^d} \frac{1}{|n_2|^{6\alpha}} \big(   \sum_{n_1\in \Z^d}\frac{1}{|n-n_1|^{d-2}|n_1|^{4\alpha}|n_1-n_2|^{6(d-2)}}  \big)^2 \right)^{\frac{1}{2}}\\
         & \lesssim_{p} \kappa^7 \biggl(   \sum_{n_2\in \Z^d} \frac{1}{|n_2|^{6\alpha}} \big(   \sum_{n_1\in \Z^d}\frac{1}{|n-n_1|^{2(d-2)}|n_1|^{8\alpha}|n-n_2|^{6(d-2)}} \big) \cdot \big(   \sum_{n_1\in \Z^d}\frac{1}{|n_1-n_2|^{6(d-2)}}  \big) \biggl)^{\frac12}\\
         & \lesssim_{d,p,\alpha} \kappa^7 \left( \sum_{n_2\in \Z^d} \frac{1}{|n_2|^{6\alpha}} \cdot \frac{1}{|n-n_2|^{2(d-2)}(|n|\wedge|n_2|)^{8\alpha}} \right)^{\frac{1}{2}}\\
         & \lesssim_{d,p,\alpha} \kappa^7 \frac{1}{|n|^{7\alpha}}.
   \end{align*}
   Finally, for the term $G_0V S \xi(n)$ and $G_0 S^{\top}V \xi (n)$,  similar to the above proof, for example, we have   
   \begin{align*}
    &\E_p\left | G_0 V S \xi (n)  \right | =\E_p \left |  \sum_{n_1,n_2,n_3\in \Z^d}  G_0(n,n_1) \omega_{n_1} v^3_{n_1} v^2_{n_2} v^2_{n_3}  \wtg(n_1,n_3)^2  \wtg(n_1,n_2)^3 \wtg(n_2,n_3)  \xi(n_3)         \right |\\
    &\lesssim_p \left(\sum_{n_1\in \Z^d} G_0 (n,n_1)^2 v^6_{n_1} \big( \sum_{n_2,n_3\in \Z^d} v^2_{n_2}v^2_{n_3}\wtg(n_2,n_3)\wtg(n_1,n_3)^2\wtg(n_1,n_2)^3 \big)^2 \right)^{\frac{1}{2}}\\
    &\lesssim_p \left(\sum_{n_1\in \Z^d} G_0 (n,n_1)^2 v^6_{n_1} \big[ \sum_{n_2\in \Z^d} \big(\sum_{n_3\in \Z^d}v^2_{n_3}\wtg(n_2,n_3)\wtg(n_1,n_3)^2 \big) v^2_{n_2}\wtg(n_1,n_2)^3 \big]^2 \right)^{\frac{1}{2}}\\
    &\lesssim_p \kappa^2 \left(\sum_{n_1\in \Z^d} G_0 (n,n_1)^2 v^6_{n_1} \big[ \sum_{n_2\in \Z^d} \big(\sum_{n_3\in \Z^d}\frac{1}{|n_2-n_3|^{d-2}|n_3|^{2\alpha}|n_3-n_1|^{2(d-2)}} \big) v^2_{n_2}\wtg(n_1,n_2)^3 \big]^2 \right)^{\frac{1}{2}}\\
    &\lesssim_{d,p,\alpha} \kappa^7 \left(\sum_{n_1\in \Z^d} \frac{1}{|n-n_1|^{2(d-2)}|n_1|^{6\alpha}} \big[ \sum_{n_2\in \Z^d} \frac{1}{|n_2-n_1|^{d-2}(|n_1|\wedge|n_2|)^{2\alpha}} \cdot \frac{1}{|n_2|^{2\alpha}|n_2-n_1|^{3(d-2)}} \big]^2 \right)^{\frac{1}{2}}\\
    &\lesssim_{d,p,\alpha} \kappa^7 \left(\sum_{n_1\in \Z^d} \frac{1}{|n-n_1|^{2(d-2)}|n_1|^{6\alpha}} \cdot \frac{1}{|n_1|^{8\alpha}} \right)^{\frac{1}{2}}\\
    & \lesssim_{d,p,\alpha} \kappa^7 \frac{1}{|n|^{7\alpha}}.
  \end{align*}
  Summarizing all the above estimates leads to  \eqref{5.13} for $i=7$.
  \end{proof}

We are ready to prove Theorem \ref{extthm}. 

\begin{proof}[Proof of Theorem \ref{extthm}]

From Theorem \ref{Theorem 5.1}, it follows  that 
\begin{equation}\label{A'' moment}
 \E_p\left | (A^{''}-G_0)H_0 \xi (n)\right | \lesssim_{d,p,\alpha}\sum_{i=1}^{7} \kappa^i \frac{1}{|n|^{i\alpha}}\lesssim_{d,p,\alpha} \kappa\frac{1}{|n|^{\alpha}}. 
\end{equation}
Denote 
\begin{align*}
  B^{''} & = (4  Q_6^{(2)} G_0  -2\sigma^2  Q_6^{(1)} G_0 - 4 Q_6^{(2)} G_0 V G_0)  +\widetilde{B}^{''}. 
\end{align*}
Then applying Theorem \ref{Theorem 5.1} again gives desired estimate on  $GB^{''}H_0 \hat{\delta_0}$ in \eqref{5.11}.  So, we have   
\begin{equation}\label{wtB'' moment}
  \E_p\left | \widetilde{B}^{''}H_0 \xi (n)\right | \lesssim_{d,p,\alpha} \kappa^8 \frac{1}{|n|^{8\alpha}}.
\end{equation}
It remains  to control the additional terms generated from the rearrangement. Recalling  the estimates  \eqref{5.2} and \eqref{5.3} on the non-symmetrical  differences  $Q_6^{(1)}$ and $Q_6^{(2)}$, we get by Lemma \ref{Lemma 2.1} that
\begin{align*}
    |Q_6^{(1)} G_0 H_0 \xi (n)|& =|Q^{(1)}_6 \hat{\delta_0}(n) | \\
    &\lesssim_{d,\alpha}\kappa^6 \sum_{n_1\in \Z^d}\frac{1}{|n-n_1|^{3(d-2)-1}(|n|\wedge|n_1|)^{6\alpha+1}} \\
    &\lesssim_{d,\alpha}\kappa^6 \frac{1}{|n|^{6\alpha+1}}.
\end{align*}
Similarly,   
\begin{align*}
  &|Q_6^{(2)} G_0 H_0 \hat{\delta_0}(n)|\lesssim_{d,\alpha}\kappa^6 \frac{1}{|n|^{6\alpha+1}}. 
\end{align*}
In  the above estimates, it requires  $d\geq 5$, which implies $3(d-2)-2> d$. Note that we have  \[|Q_6^{(2)} G_0(n,n')| \lesssim_{d,\alpha}\kappa^6 \frac{1}{|n-n'|^{d-2}(|n|\wedge|n'|)^{6\alpha+1}}.\]
Combining the above estimates  yields  
\begin{align*}
  \E_p\left |Q_6^{(2)} G_0 V\xi (n)\right | &\lesssim_{d,p,\alpha}\kappa^6 \left( \sum_{n_1\in \Z^d} \frac{1}{|n-n_1|^{2(d-2)}(|n|\wedge|n_1|)^{2(6\alpha+1)}|n_1|^{2\alpha}} \right)^{\frac{1}{2}} \\
   &\lesssim_{d,p,\alpha} \kappa^6 \frac{1}{|n|^{7\alpha+1}}.
\end{align*}
Putting all the above estimates together shows 
\begin{equation}\label{B'' moment}
  \left\| B^{''}H_0 \xi (n)\right\|_p \lesssim_{d,p,\alpha} \kappa^6 \frac{1}{|n|^{\min\{8\alpha,6\alpha+1\}}}. 
\end{equation}

Now, keep in mind  that $\frac{1}{4}<\alpha\leq \frac{1}{3}$. Similar to the  proof  of  $\eqref{4.69}\sim\eqref{A' high prob}$,  we can apply  Chebyshev's inequality by choosing   $0<100\varepsilon<\min\{8\alpha-2,6\alpha-1\}=8\alpha-2$. From moment estimates  \eqref{A'' moment} and \eqref{B'' moment}, it follows  that with  probability $1-\mcO(\kappa^{\frac{p}{2}})$ (mainly coming from   $(A^{''}-G_0)H_0\xi(n)$),  
\[  | (A^{''}-G_0)H_0 \xi (n)|\lesssim_{d,\alpha} \kappa^{\frac{1}{2}}\frac{1}{|n|^{\alpha-\varepsilon}} \  {\rm for}\  \forall  n\in \Z^d,\]
\[  | B^{''}H_0 \xi (n) | \leq \kappa^2 \frac{1}{|n|^{8\alpha-\varepsilon}}\ {\rm for}\  \forall n\in \Z^d.\]
Applying Theorem \ref{greenthm} shows that with high probability ($1-\mcO(\kappa^{{p}})$),
\[|G(n,n')|\lesssim \frac{1}{|n-n'|^{d-2-\varepsilon}}  \ {\rm for}\  \forall  n,n'\in \Z^d.\]
Hence,
\begin{align*}
  |GB^{''}H_0 \xi (n)| & \lesssim \kappa^2 \sum_{n_1\in Z^d} \frac{1}{|n-n_1|^{d-2-\varepsilon}|n_1|^{8\alpha-\varepsilon}}\\
     & \leq \kappa \frac{1}{|n|^{8\alpha-2-2\varepsilon}}.
\end{align*}
Finally, from  $\alpha-\varepsilon>0$ and $8\alpha-2-2\varepsilon>0$, it follows that with  probability $1-\mathcal O(\kappa^{\frac p2}),$
\[|-G(W+\widetilde{V}) \xi|=\mcO(\sqrt{\kappa})\  {\rm in}\  \ell^{\infty}(\Z^d),\]
which implies 
\[\zeta=\hat\delta_0+\mcO(\sqrt{\kappa})\  {\rm in}\  \ell^{\infty}(\Z^d),\ H\zeta=0.\]

We finish the proof of Theorem \ref{extthm}.
\end{proof}

\begin{rmk}\label{finrmk}
  Finally, we emphasize the presence of some  new phenomena and obstacles  compared with \cite{Bou03} when expanding the resolvent to higher order terms. If one wishes to relax the condition $\alpha>\frac{1}{4}$ to $\alpha>0$, these issues should  be addressed.  
  
  \begin{itemize}
    \item[(1)] The  presence of 7th-order  remaining terms \eqref{4.21} and \eqref{4.24}  indicates  that in higher-order expansions, not all terms can be expressed as summations over admissible tuples. Thus, one should  improve  Lemma \ref{decoupling} further in order to  handle  these  random but not admissible   terms.
    
    \item[(2)] From  \eqref{B'' moment},  it  requires $6\alpha+1>2$ (i.e., $\alpha>\frac{1}{6}$) to ensure convergence. This may prevent us from relaxing the condition $\alpha>\frac14$ to $\alpha>0$. Although the key lemma of \cite{Bou03} (cf. Lemma 1.2]) might seem applicable, the operator $C$ in the 6th-order remaining  terms  does not match the form $cMd+dMc$ as required in that lemma. Therefore,  Lemma 1.2 in \cite{Bou03} cannot be  directly applied. Moreover, our  proof of Theorem \ref{6,7 resolvent estimate} shows that a non-symmetrical  difference operator  $P_6^{''}$ always exists, yielding only a power-law decay rate of $6\alpha+1$ instead of the $6\alpha+2$ one as in the symmetrical differences case.
  \end{itemize} 
\end{rmk}

\appendix{}
\section{Computation of the  6th-order renormalization via graphs}\label{App6th}
In this section, we will compute the 6th-order renormalization via graphs representation. 

From now on,  when we use the notation $n_1,n_2,n_3,\cdots$ in some tuple, we always assume that $n_i\neq n_j$ for $i\neq j$ in $\Z^d$. And if we use  $m_1,m_2,\cdots$, the relationship of $m_i,m_j$ ($i\neq j$)  may  not  be determined.

Suppose that we have found  the 4th-order renormalized  potential
\[V^{(4)}=V+\sigma v^2-\rho v^4\]
as in \cite{Bou03}.  Substitute $\widetilde{V}=V^{(4)}$ into  \eqref{8 order expansion}. Then the terms with orders less than 5 can be found in the Subsection \ref{Bousub}. 

Now consider the 6th-order terms in \eqref{8 order expansion}: 
\begin{align}
\label{A1}  & -\sigma\rho (G_0v^2G_0 v^4 G_0+G_0v^4 G_0 v^2 G_0) \\
\label{A2}  &+\rho (G_0 v^4 G_0 VG_0VG_0+ G_0 V G_0 v^4 G_0V G_0+G_0VG_0VG_0 v^4 G_0)\\
\label{A3}  &-\sigma^3 G_0 v^2 G_0 v^2 G_0 v^2 G_0 \\
\label{A4} & +\sigma^2 ( G_0 v^2 G_0 v^2 G_0VG_0VG_0 +G_0 v^2 G_0 V G_0 v^2 G_0VG_0+ G_0 v^2 G_0 V G_0VG_0v^2G_0 \\
 \notag & \qquad +G_0 V G_0 v^2 G_0 v^2G_0VG_0+G_0 V G_0 v^2 G_0VG_0v^2 G_0+G_0 V G_0 V G_0v^2 G_0v^2 G_0 ) \\
 \label{A5} &-\sigma (G_0 v^2 G_0 V G_0VG_0VG_0 V G_0+G_0 V G_0 v^2 G_0VG_0VG_0 V G_0 \\
  \notag & \qquad+G_0 V G_0 V G_0 v^2 G_0VG_0 V G_0 +G_0 V G_0 V G_0VG_0 v^2 G_0 V G_0 \\
  \notag& \qquad +G_0 V G_0 V G_0VG_0VG_0 v^2 G_0 )\\
  \label{A6} &+G_0 V G_0 V G_0VG_0VG_0 V G_0 V G_0.
\end{align}
We will associate each term  in $\eqref{A1}\sim \eqref{A6}$ with a  graph. For a $s$-tuple $(m_1,m_2,\cdots,m_s)\in (\Z^d)^s$, we define its  {\bf characteristic graph} to be $(\mathcal V, \mathcal E)$, with $\mathcal V=\{1,2,\cdots,s\}\subset\Z$, and 
\[\mathcal E\subset \{(i,i+1):\ i=1,\cdots ,s-1\}.\]
Additionally, we label the edge $(i,i+1)\in \mathcal E$ with a  \textbf{solid} line if $m_i=m_{i+1}$, and with a \textbf{dotted} line if $m_i\neq m_{i+1}$. For example, the characteristic graph of the tuple $(n_1,n_1,n_2,n_3,n_3,n_1,n_1)$ is \\
\begin{figure}[htbp]
  \centering

  \tikzset{every picture/.style={line width=0.75pt}} 

  \begin{tikzpicture}[x=0.75pt,y=0.75pt,yscale=-0.7,xscale=0.7]
  
  \draw  [fill={rgb, 255:red, 208; green, 2; blue, 27 }  ,fill opacity=1 ] (8.89,154) .. controls (8.89,151.32) and (10.82,149.14) .. (13.21,149.14) .. controls (15.59,149.14) and (17.52,151.32) .. (17.52,154) .. controls (17.52,156.68) and (15.59,158.86) .. (13.21,158.86) .. controls (10.82,158.86) and (8.89,156.68) .. (8.89,154) -- cycle ;
  \draw  [fill={rgb, 255:red, 208; green, 2; blue, 27 }  ,fill opacity=1 ] (113.79,152.86) .. controls (113.79,150.17) and (115.72,148) .. (118.11,148) .. controls (120.49,148) and (122.42,150.17) .. (122.42,152.86) .. controls (122.42,155.54) and (120.49,157.71) .. (118.11,157.71) .. controls (115.72,157.71) and (113.79,155.54) .. (113.79,152.86) -- cycle ;
  \draw  [fill={rgb, 255:red, 208; green, 2; blue, 27 }  ,fill opacity=1 ] (218.69,152.86) .. controls (218.69,150.17) and (220.62,148) .. (223,148) .. controls (225.39,148) and (227.32,150.17) .. (227.32,152.86) .. controls (227.32,155.54) and (225.39,157.71) .. (223,157.71) .. controls (220.62,157.71) and (218.69,155.54) .. (218.69,152.86) -- cycle ;
  \draw  [fill={rgb, 255:red, 208; green, 2; blue, 27 }  ,fill opacity=1 ] (323.58,151.71) .. controls (323.58,149.03) and (325.52,146.86) .. (327.9,146.86) .. controls (330.29,146.86) and (332.22,149.03) .. (332.22,151.71) .. controls (332.22,154.4) and (330.29,156.57) .. (327.9,156.57) .. controls (325.52,156.57) and (323.58,154.4) .. (323.58,151.71) -- cycle ;
  \draw  [fill={rgb, 255:red, 208; green, 2; blue, 27 }  ,fill opacity=1 ] (428.48,151.71) .. controls (428.48,149.03) and (430.42,146.86) .. (432.8,146.86) .. controls (435.18,146.86) and (437.12,149.03) .. (437.12,151.71) .. controls (437.12,154.4) and (435.18,156.57) .. (432.8,156.57) .. controls (430.42,156.57) and (428.48,154.4) .. (428.48,151.71) -- cycle ;
  \draw  [fill={rgb, 255:red, 208; green, 2; blue, 27 }  ,fill opacity=1 ] (529.06,150.57) .. controls (529.06,147.89) and (531,145.71) .. (533.38,145.71) .. controls (535.77,145.71) and (537.7,147.89) .. (537.7,150.57) .. controls (537.7,153.25) and (535.77,155.43) .. (533.38,155.43) .. controls (531,155.43) and (529.06,153.25) .. (529.06,150.57) -- cycle ;
  \draw    (13.21,154) -- (118.11,152.86) ;
  \draw    (327.9,151.71) -- (432.8,151.71) ;
  \draw  [fill={rgb, 255:red, 208; green, 2; blue, 27 }  ,fill opacity=1 ] (635.48,149.71) .. controls (635.48,147.03) and (637.42,144.86) .. (639.8,144.86) .. controls (642.18,144.86) and (644.12,147.03) .. (644.12,149.71) .. controls (644.12,152.4) and (642.18,154.57) .. (639.8,154.57) .. controls (637.42,154.57) and (635.48,152.4) .. (635.48,149.71) -- cycle ;
  \draw    (534.9,149.71) -- (639.8,149.71) ;
  \draw  [dash pattern={on 0.84pt off 2.51pt}]  (118.11,152.86) -- (223,152.86) ;
  \draw  [dash pattern={on 0.84pt off 2.51pt}]  (223,152.86) -- (327.9,151.71) ;
  \draw  [dash pattern={on 0.84pt off 2.51pt}]  (432.8,151.71) -- (533.38,150.57) ;

  \end{tikzpicture}

\end{figure}
\\
What's more, if in a tuple, one cannot determine whether $m_i=m_{i+1}$ or not,  then we discard $(i,i+1)$ from the edge set $\mathcal E$ (i.e.,  there is no edge between $i$ and $i+1$) and say there is a \textbf{vacuum} edge between $i$ and $i+1$. In fact, if there is no edge between $i$ and $i+1$, one may think  that there can be both a solid line ($m_i=m_{i+1}$) and a dotted  line ($m_i\neq m_{i+1}$) between $i$ and $i+1$, and hence they cause a counteraction.

We say a characteristic graph is \textbf{complete}, if $(i,i+1)\in \mathcal E$ for all $1\leq i\leq s-1$. Namely,  $(\mathcal V, \mathcal E)$ is complete if and only if  the relationship of  all adjacent pairs  is determined in the corresponding tuple. Indeed, each incomplete graph can be viewed as a union of some complete graphs, with the vacuum edge replaced by solid or dotted  line.

In a complete characteristic graph, a {\bf connected  component} is a connected  segment  made up of only vertices and solid edges. For example, the above characteristic graph of $(n_1,n_1,n_2,n_3,n_3,n_1,n_1)$ has four connected  components. The number of vertices in a connected  component is called the \textbf{length} of this  connected  component. 

Now, let's consider the terms $\eqref{A1}\sim \eqref{A6}$. For example, the first term in \eqref{A1}, $G_0 v^2 G_0 v^4 G_0$  corresponds  to the summation tuple $(m_1,m_1,m_2,m_2,m_2,m_2)$ and we don't know the relationship between $m_1$ and $m_2$. So its characteristic graph is \\
\begin{figure}[htbp]
  \centering

  \tikzset{every picture/.style={line width=0.75pt}} 

  \begin{tikzpicture}[x=0.75pt,y=0.75pt,yscale=-0.7,xscale=0.7]
  
  \draw  [fill={rgb, 255:red, 208; green, 2; blue, 27 }  ,fill opacity=1 ] (60.89,154) .. controls (60.89,151.32) and (62.82,149.14) .. (65.21,149.14) .. controls (67.59,149.14) and (69.52,151.32) .. (69.52,154) .. controls (69.52,156.68) and (67.59,158.86) .. (65.21,158.86) .. controls (62.82,158.86) and (60.89,156.68) .. (60.89,154) -- cycle ;
  \draw  [fill={rgb, 255:red, 208; green, 2; blue, 27 }  ,fill opacity=1 ] (165.79,152.86) .. controls (165.79,150.17) and (167.72,148) .. (170.11,148) .. controls (172.49,148) and (174.42,150.17) .. (174.42,152.86) .. controls (174.42,155.54) and (172.49,157.71) .. (170.11,157.71) .. controls (167.72,157.71) and (165.79,155.54) .. (165.79,152.86) -- cycle ;
  \draw  [fill={rgb, 255:red, 208; green, 2; blue, 27 }  ,fill opacity=1 ] (270.69,152.86) .. controls (270.69,150.17) and (272.62,148) .. (275,148) .. controls (277.39,148) and (279.32,150.17) .. (279.32,152.86) .. controls (279.32,155.54) and (277.39,157.71) .. (275,157.71) .. controls (272.62,157.71) and (270.69,155.54) .. (270.69,152.86) -- cycle ;
  \draw  [fill={rgb, 255:red, 208; green, 2; blue, 27 }  ,fill opacity=1 ] (375.58,151.71) .. controls (375.58,149.03) and (377.52,146.86) .. (379.9,146.86) .. controls (382.29,146.86) and (384.22,149.03) .. (384.22,151.71) .. controls (384.22,154.4) and (382.29,156.57) .. (379.9,156.57) .. controls (377.52,156.57) and (375.58,154.4) .. (375.58,151.71) -- cycle ;
  \draw  [fill={rgb, 255:red, 208; green, 2; blue, 27 }  ,fill opacity=1 ] (480.48,151.71) .. controls (480.48,149.03) and (482.42,146.86) .. (484.8,146.86) .. controls (487.18,146.86) and (489.12,149.03) .. (489.12,151.71) .. controls (489.12,154.4) and (487.18,156.57) .. (484.8,156.57) .. controls (482.42,156.57) and (480.48,154.4) .. (480.48,151.71) -- cycle ;
  \draw  [fill={rgb, 255:red, 208; green, 2; blue, 27 }  ,fill opacity=1 ] (581.06,150.57) .. controls (581.06,147.89) and (583,145.71) .. (585.38,145.71) .. controls (587.77,145.71) and (589.7,147.89) .. (589.7,150.57) .. controls (589.7,153.25) and (587.77,155.43) .. (585.38,155.43) .. controls (583,155.43) and (581.06,153.25) .. (581.06,150.57) -- cycle ;
  \draw    (65.21,154) -- (170.11,152.86) ;
  \draw    (379.9,151.71) -- (484.8,151.71) ;
  \draw    (275,152.86) -- (379.9,151.71) ;
  \draw    (484.8,151.71) -- (585.38,150.57) ;

  \end{tikzpicture}

\end{figure}\\
and this graph can be separated into two complete characteristic graphs\\
\begin{figure}[htbp]
  \centering

\tikzset{every picture/.style={line width=0.75pt}} 

\begin{tikzpicture}[x=0.75pt,y=0.75pt,yscale=-0.7,xscale=0.7]

\draw  [fill={rgb, 255:red, 208; green, 2; blue, 27 }  ,fill opacity=1 ] (13.89,152.38) .. controls (13.89,150.22) and (14.89,148.47) .. (16.13,148.47) .. controls (17.37,148.47) and (18.38,150.22) .. (18.38,152.38) .. controls (18.38,154.54) and (17.37,156.29) .. (16.13,156.29) .. controls (14.89,156.29) and (13.89,154.54) .. (13.89,152.38) -- cycle ;
\draw  [fill={rgb, 255:red, 208; green, 2; blue, 27 }  ,fill opacity=1 ] (68.41,151.46) .. controls (68.41,149.3) and (69.41,147.55) .. (70.65,147.55) .. controls (71.89,147.55) and (72.89,149.3) .. (72.89,151.46) .. controls (72.89,153.62) and (71.89,155.37) .. (70.65,155.37) .. controls (69.41,155.37) and (68.41,153.62) .. (68.41,151.46) -- cycle ;
\draw  [fill={rgb, 255:red, 208; green, 2; blue, 27 }  ,fill opacity=1 ] (122.92,151.46) .. controls (122.92,149.3) and (123.93,147.55) .. (125.17,147.55) .. controls (126.4,147.55) and (127.41,149.3) .. (127.41,151.46) .. controls (127.41,153.62) and (126.4,155.37) .. (125.17,155.37) .. controls (123.93,155.37) and (122.92,153.62) .. (122.92,151.46) -- cycle ;
\draw  [fill={rgb, 255:red, 208; green, 2; blue, 27 }  ,fill opacity=1 ] (177.44,150.54) .. controls (177.44,148.38) and (178.44,146.63) .. (179.68,146.63) .. controls (180.92,146.63) and (181.93,148.38) .. (181.93,150.54) .. controls (181.93,152.7) and (180.92,154.45) .. (179.68,154.45) .. controls (178.44,154.45) and (177.44,152.7) .. (177.44,150.54) -- cycle ;
\draw  [fill={rgb, 255:red, 208; green, 2; blue, 27 }  ,fill opacity=1 ] (231.95,150.54) .. controls (231.95,148.38) and (232.96,146.63) .. (234.2,146.63) .. controls (235.44,146.63) and (236.44,148.38) .. (236.44,150.54) .. controls (236.44,152.7) and (235.44,154.45) .. (234.2,154.45) .. controls (232.96,154.45) and (231.95,152.7) .. (231.95,150.54) -- cycle ;
\draw  [fill={rgb, 255:red, 208; green, 2; blue, 27 }  ,fill opacity=1 ] (284.23,149.62) .. controls (284.23,147.46) and (285.23,145.71) .. (286.47,145.71) .. controls (287.71,145.71) and (288.71,147.46) .. (288.71,149.62) .. controls (288.71,151.78) and (287.71,153.53) .. (286.47,153.53) .. controls (285.23,153.53) and (284.23,151.78) .. (284.23,149.62) -- cycle ;
\draw    (16.13,152.38) -- (70.65,151.46) ;
\draw    (179.68,150.54) -- (234.2,150.54) ;
\draw    (125.17,151.46) -- (179.68,150.54) ;
\draw    (234.2,150.54) -- (286.47,149.62) ;
\draw    (70.65,151.46) -- (122.92,151.46) ;
\draw  [fill={rgb, 255:red, 208; green, 2; blue, 27 }  ,fill opacity=1 ] (369.89,148.38) .. controls (369.89,146.22) and (370.89,144.47) .. (372.13,144.47) .. controls (373.37,144.47) and (374.38,146.22) .. (374.38,148.38) .. controls (374.38,150.54) and (373.37,152.29) .. (372.13,152.29) .. controls (370.89,152.29) and (369.89,150.54) .. (369.89,148.38) -- cycle ;
\draw  [fill={rgb, 255:red, 208; green, 2; blue, 27 }  ,fill opacity=1 ] (424.41,147.46) .. controls (424.41,145.3) and (425.41,143.55) .. (426.65,143.55) .. controls (427.89,143.55) and (428.89,145.3) .. (428.89,147.46) .. controls (428.89,149.62) and (427.89,151.37) .. (426.65,151.37) .. controls (425.41,151.37) and (424.41,149.62) .. (424.41,147.46) -- cycle ;
\draw  [fill={rgb, 255:red, 208; green, 2; blue, 27 }  ,fill opacity=1 ] (478.92,147.46) .. controls (478.92,145.3) and (479.93,143.55) .. (481.17,143.55) .. controls (482.4,143.55) and (483.41,145.3) .. (483.41,147.46) .. controls (483.41,149.62) and (482.4,151.37) .. (481.17,151.37) .. controls (479.93,151.37) and (478.92,149.62) .. (478.92,147.46) -- cycle ;
\draw  [fill={rgb, 255:red, 208; green, 2; blue, 27 }  ,fill opacity=1 ] (533.44,146.54) .. controls (533.44,144.38) and (534.44,142.63) .. (535.68,142.63) .. controls (536.92,142.63) and (537.93,144.38) .. (537.93,146.54) .. controls (537.93,148.7) and (536.92,150.45) .. (535.68,150.45) .. controls (534.44,150.45) and (533.44,148.7) .. (533.44,146.54) -- cycle ;
\draw  [fill={rgb, 255:red, 208; green, 2; blue, 27 }  ,fill opacity=1 ] (587.95,146.54) .. controls (587.95,144.38) and (588.96,142.63) .. (590.2,142.63) .. controls (591.44,142.63) and (592.44,144.38) .. (592.44,146.54) .. controls (592.44,148.7) and (591.44,150.45) .. (590.2,150.45) .. controls (588.96,150.45) and (587.95,148.7) .. (587.95,146.54) -- cycle ;
\draw  [fill={rgb, 255:red, 208; green, 2; blue, 27 }  ,fill opacity=1 ] (640.23,145.62) .. controls (640.23,143.46) and (641.23,141.71) .. (642.47,141.71) .. controls (643.71,141.71) and (644.71,143.46) .. (644.71,145.62) .. controls (644.71,147.78) and (643.71,149.53) .. (642.47,149.53) .. controls (641.23,149.53) and (640.23,147.78) .. (640.23,145.62) -- cycle ;
\draw    (372.13,148.38) -- (426.65,147.46) ;
\draw    (535.68,146.54) -- (590.2,146.54) ;
\draw    (481.17,147.46) -- (535.68,146.54) ;
\draw    (590.2,146.54) -- (642.47,145.62) ;
\draw  [dash pattern={on 0.84pt off 2.51pt}]  (426.65,147.46) -- (478.92,147.46) ;

\draw (318,140) node [anchor=north west][inner sep=0.75pt]   [align=left] {and};

\end{tikzpicture}

\end{figure}

Rewrite the summations $\eqref{A1}\sim \eqref{A6}$ with their characteristic graphs. Then we get the  formal  graphs representation 
\begin{align*}
 \label{A1 graph} \tag{A.1'}  - &\sigma\rho (\raisebox{-0.6 ex}{\includegraphics[width=4cm]{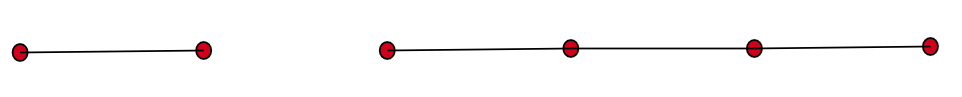}}+\raisebox{-0.6 ex}{\includegraphics[width=4cm]{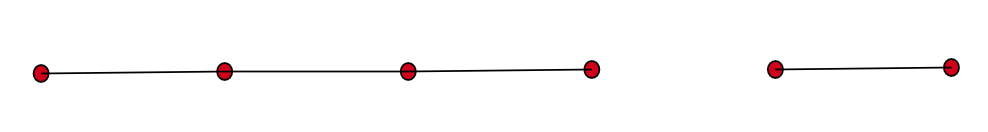}}) \\
 \label{A2 graph} \tag{A.2'} + &\rho (\raisebox{-0.6 ex}{\includegraphics[width=4cm]{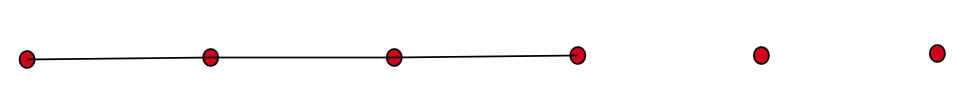}}+\raisebox{-0.6 ex}{\includegraphics[width=4cm]{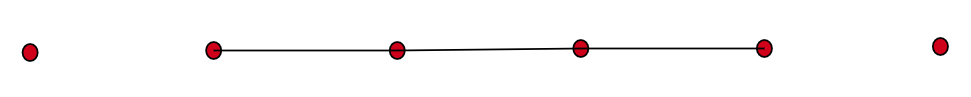}}+\raisebox{-0.6 ex}{\includegraphics[width=4cm]{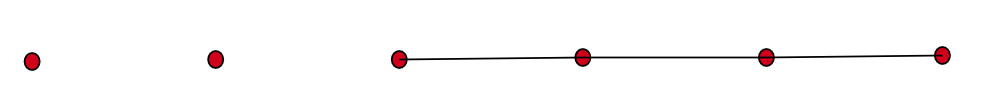}}) \\
 \label{A3 graph} \tag{A,3'} -&\sigma^3 (\raisebox{-0.6 ex}{\includegraphics[width=4cm]{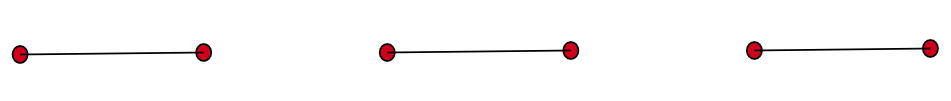}})\\
 \label{A4 graph} \tag{A.4'} +&\sigma^2 (\raisebox{-0.6 ex}{\includegraphics[width=4cm]{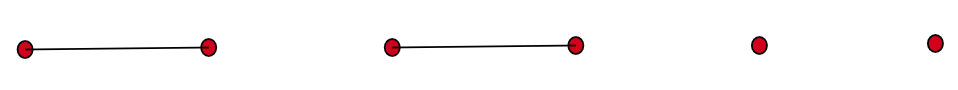}}+\raisebox{-0.6 ex}{\includegraphics[width=4cm]{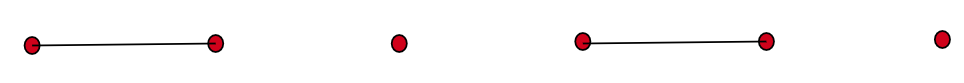}}+ \raisebox{-0.6 ex}{\includegraphics[width=4cm]{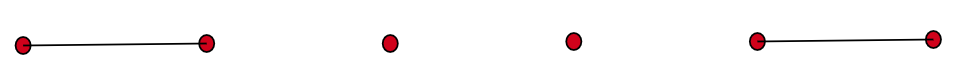}}\\
      &\qquad +\raisebox{-0.6 ex}{\includegraphics[width=4cm]{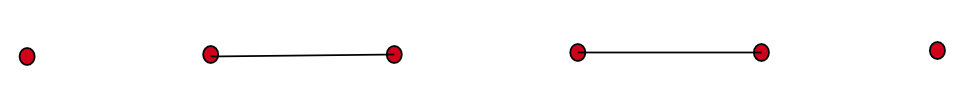}}  +\raisebox{-0.6 ex}{\includegraphics[width=4cm]{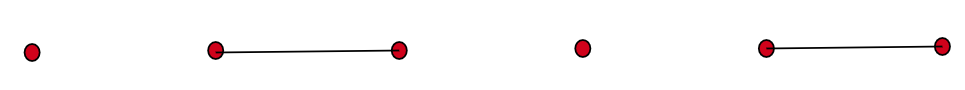}}+\raisebox{-0.6 ex}{\includegraphics[width=4cm]{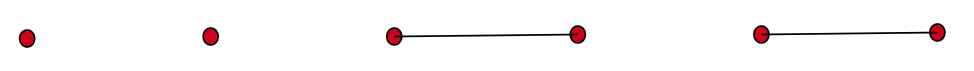}})\\
  \label{A5 graph} \tag{A.5'} -&\sigma(\raisebox{-0.6 ex}{\includegraphics[width=4cm]{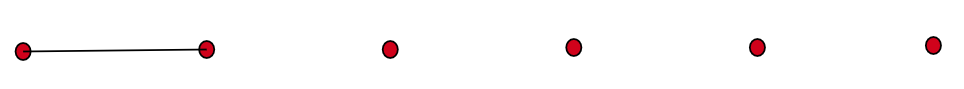}} +\raisebox{-0.6 ex}{\includegraphics[width=4cm]{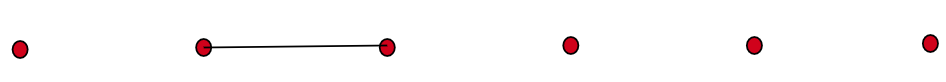}}\\
    &   \qquad +\raisebox{-0.6 ex}{\includegraphics[width=4cm]{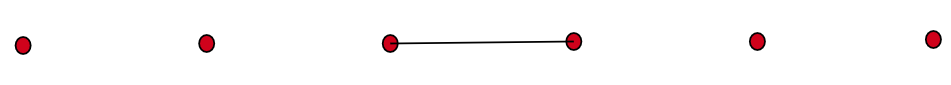}}+\raisebox{-0.6 ex}{\includegraphics[width=4cm]{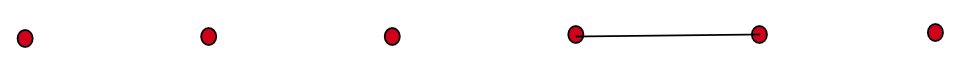}}+\raisebox{-0.6 ex}{\includegraphics[width=4cm]{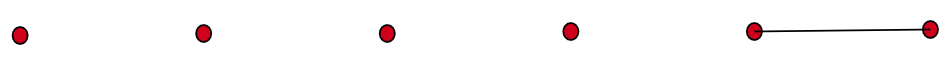}})\\
    \label{A6 graph}  \tag{A.6'} +&(\raisebox{-0.6 ex}{\includegraphics[width=4cm]{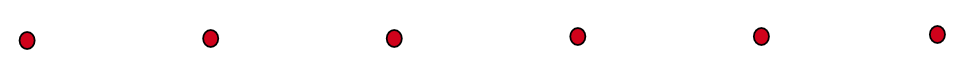}}).
\end{align*}\\

From now on, let's investigate the 6th-order terms with complete characteristic graphs.
\subsection{Terms without randomness} First, we consider those 6th-order terms  in $\eqref{A1}\sim \eqref{A6}$ that have  no randomness, namely, those terms with their 6-tuple summations  being  cancelled.   Then all vertexes  will appear in an  even number of  times in the summations. We discuss:\smallskip

\noindent\textcircled{1}{\it The cancelled summation tuple  has only  one vertex }.  In this case, it must be $(n_1,n_1,n_1,n_1,n_1,n_1)$, which corresponds to  $G_0 v^6 G_0$ with its complete characteristic graph\\
\begin{figure}[htbp]
  \centering

\tikzset{every picture/.style={line width=0.75pt}} 

\begin{tikzpicture}[x=0.75pt,y=0.75pt,yscale=-0.5,xscale=0.5]

\draw  [fill={rgb, 255:red, 208; green, 2; blue, 27 }  ,fill opacity=1 ] (60.89,154) .. controls (60.89,151.32) and (62.82,149.14) .. (65.21,149.14) .. controls (67.59,149.14) and (69.52,151.32) .. (69.52,154) .. controls (69.52,156.68) and (67.59,158.86) .. (65.21,158.86) .. controls (62.82,158.86) and (60.89,156.68) .. (60.89,154) -- cycle ;
\draw  [fill={rgb, 255:red, 208; green, 2; blue, 27 }  ,fill opacity=1 ] (165.79,152.86) .. controls (165.79,150.17) and (167.72,148) .. (170.11,148) .. controls (172.49,148) and (174.42,150.17) .. (174.42,152.86) .. controls (174.42,155.54) and (172.49,157.71) .. (170.11,157.71) .. controls (167.72,157.71) and (165.79,155.54) .. (165.79,152.86) -- cycle ;
\draw  [fill={rgb, 255:red, 208; green, 2; blue, 27 }  ,fill opacity=1 ] (270.69,152.86) .. controls (270.69,150.17) and (272.62,148) .. (275,148) .. controls (277.39,148) and (279.32,150.17) .. (279.32,152.86) .. controls (279.32,155.54) and (277.39,157.71) .. (275,157.71) .. controls (272.62,157.71) and (270.69,155.54) .. (270.69,152.86) -- cycle ;
\draw  [fill={rgb, 255:red, 208; green, 2; blue, 27 }  ,fill opacity=1 ] (375.58,151.71) .. controls (375.58,149.03) and (377.52,146.86) .. (379.9,146.86) .. controls (382.29,146.86) and (384.22,149.03) .. (384.22,151.71) .. controls (384.22,154.4) and (382.29,156.57) .. (379.9,156.57) .. controls (377.52,156.57) and (375.58,154.4) .. (375.58,151.71) -- cycle ;
\draw  [fill={rgb, 255:red, 208; green, 2; blue, 27 }  ,fill opacity=1 ] (480.48,151.71) .. controls (480.48,149.03) and (482.42,146.86) .. (484.8,146.86) .. controls (487.18,146.86) and (489.12,149.03) .. (489.12,151.71) .. controls (489.12,154.4) and (487.18,156.57) .. (484.8,156.57) .. controls (482.42,156.57) and (480.48,154.4) .. (480.48,151.71) -- cycle ;
\draw  [fill={rgb, 255:red, 208; green, 2; blue, 27 }  ,fill opacity=1 ] (581.06,150.57) .. controls (581.06,147.89) and (583,145.71) .. (585.38,145.71) .. controls (587.77,145.71) and (589.7,147.89) .. (589.7,150.57) .. controls (589.7,153.25) and (587.77,155.43) .. (585.38,155.43) .. controls (583,155.43) and (581.06,153.25) .. (581.06,150.57) -- cycle ;
\draw    (65.21,154) -- (170.11,152.86) ;
\draw    (170.11,154) -- (275,152.86) ;
\draw    (275,152.86) -- (379.9,151.71) ;
\draw    (379.9,152.86) -- (484.8,151.71) ;
\draw    (484.8,151.71) -- (585.38,150.57) ;

\end{tikzpicture}

\caption{The characteristic graph of $G_0 v^6 G_0$.}
\label{111111}
\end{figure}
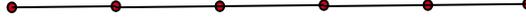

Now let's figure out  what kind of incomplete graphs  in $\eqref{A1 graph}\sim \eqref{A6 graph}$ will contain the above complete graph (cf. FIGURE \ref{111111}). Two basic rules  are 
\begin{itemize}
  \item[{\bf Rule 1}:] A complete graph (denoted by $\mcG_{complete}$)  is contained in an incomplete graph (denoted by $\mcG_{incomplete}$)  if and only if the   solid  (and dotted) edges set   of $\mcG_{incomplete}$ is a subset of the solid (and dotted) edges set  of $\mcG_{complete}$. 
  \item[{\bf Rule 2}:] Every solid edge  corresponds  to  a $\sigma=G_0(n,n)=G_0(0,0)$ in the summation, and every dotted  edge  corresponds to a $\wtg=G_0-\sigma$. Moreover, since  we can view vacuum edge as the coexist of both solid and dotted  edges, the vacuum edge corresponds exactly to $G_0$. Hence, if one wants  to replace a vacuum edge with a solid one (resp. a dotted edge), the corresponding $G_0$ should  to be replaced with  $\sigma$ (resp. $\wtg$) in the term.
\end{itemize}
With  the above two rules, one can compute  the coefficients  of $G_0 v^6 G_0$ in $\eqref{A1 graph}\sim \eqref{A6 graph}$ as 
\begin{equation}
  \label{caculation of coefficient of 111111 by graph} \underbrace{-2\sigma^2 \rho}_{from \ \eqref{A1 graph},\ } \underbrace{+3\sigma^2\rho}_{from \ \eqref{A2 graph},\ }  \underbrace{-\sigma^5}_{from \ \eqref{A3 graph},\ } \underbrace{+6\sigma^5}_{from \ \eqref{A4 graph},\ } \underbrace{-5\sigma^5}_{from \ \eqref{A5 graph},\ } \underbrace{+\sigma^5}_{from \ \eqref{A6 graph}}=\sigma^5+\rho \sigma^2.
\end{equation}
However, this term  (after discarding the $G_0$ at the beginning and the end, because the renormalization focuses on  term $G_0\widetilde{V}G_0$) is ``{\bf diagonal}''. And it is {\it uncontrollable} since  for $\frac 14<\alpha\leq \frac13,$
\begin{align*}
  |G_0 v^6 G_0(n,n')| &\lesssim \sum_{n_1\in \Z^d}\frac{1}{|n-n_1|^{d-2}|n_1|^{6\alpha}|n_1-n'|^{d-2}}\\
        &\lesssim_{d,\alpha} \frac{1}{|n-n'|^{d-4}(|n|\wedge |n'|)^{6\alpha}}. 
\end{align*}
This estimate is not the  desired bound of   $|n-n'|^{-(d-2)}(|n|\wedge|n'|)^{-6\alpha}$. So, this term must be  renormalized in the potential: 
\textcolor{blue}{
  \begin{flushright}
    $\Rightarrow {\rm \textbf{Renormalization:} } \ (\sigma^5+\rho\sigma^2) v^6.$
  \end{flushright}
}
\smallskip

\noindent \textcircled{2}{\it The cancelled summation tuple  has exactly two distinct vertices}. Then, since  the randomness is cancelled, all possible cases  are listed below (cf. the next page):
\newpage
\begin{center}
  \begin{xltabular}[h]{\textwidth}{|>{\centering\arraybackslash}X |>{\centering\arraybackslash}X |>{\centering\arraybackslash}X | >{\centering\arraybackslash}X|}
   \hline
   \textbf{Number of connected components}  & \textbf{Length of connected components} & \textbf{Tuple} &\textbf{Coefficient} \\
    \hline
     2 & 2,4 & $(n_1,n_1,n_2,n_2,n_2,n_2)$ $(n_2,n_2,n_2,n_2,n_1,n_1)$ & 0\\
     \hline
     \multirow{3}{*}{3} & 2,2,2 & $(n_2,n_2,n_1,n_1,n_2,n_2)$ & 0 \\
     \cline{2-4}
     \multirow{3}{*}{} & 1,2,3 & $(n_2,n_1,n_1,n_2,n_2,n_2)$ $(n_2,n_2,n_2,n_1,n_1,n_2)$ & 0  \\
     \cline{2-4}
     \multirow{3}{*}{} & 1,1,4 & $(n_1,n_2,n_2,n_2,n_2,n_1)$ & $\rho-\sigma^3$  \\
      \hline
     \multirow{2}{*}{4} & 1,1,1,3 & $(n_1,n_2,n_1,n_2,n_2,n_2)$ $(n_2,n_1,n_2,n_2,n_2,n_1)$ $(n_1,n_2,n_2,n_2,n_1,n_2)$ $(n_2,n_2,n_2,n_1,n_2,n_1)$  & $-\sigma^2$ \\
     \cline{2-4}
     \multirow{2}{*}{} & 1,1,2,2 & $(n_1,n_2,n_2,n_1,n_2,n_2)$ $(n_2,n_2,n_1,n_2,n_2,n_1)$ & 0  \\
     \hline
     5 & 1,1,1,1,2 & $(n_2,n_1,n_2,n_1,n_2,n_2)$ $(n_2,n_1,n_2,n_2,n_1,n_2)$ $(n_2,n_2,n_1,n_2,n_1,n_2)$ & 0\\
     \hline
   \end{xltabular}
\end{center}
The coefficients listed in the table above are  calculated by the same way as in  \eqref{caculation of coefficient of 111111 by graph}. For example, consider  the tuple $(n_1,n_2,n_1,n_2,n_2,n_2)$, whose characteristic graph is 
\begin{figure}[htbp]
  \centering

  \tikzset{every picture/.style={line width=0.75pt}} 

  \begin{tikzpicture}[x=0.75pt,y=0.75pt,yscale=-0.7,xscale=0.7]
  
  \draw  [fill={rgb, 255:red, 208; green, 2; blue, 27 }  ,fill opacity=1 ] (60.89,154) .. controls (60.89,151.32) and (62.82,149.14) .. (65.21,149.14) .. controls (67.59,149.14) and (69.52,151.32) .. (69.52,154) .. controls (69.52,156.68) and (67.59,158.86) .. (65.21,158.86) .. controls (62.82,158.86) and (60.89,156.68) .. (60.89,154) -- cycle ;
  \draw  [fill={rgb, 255:red, 208; green, 2; blue, 27 }  ,fill opacity=1 ] (165.79,152.86) .. controls (165.79,150.17) and (167.72,148) .. (170.11,148) .. controls (172.49,148) and (174.42,150.17) .. (174.42,152.86) .. controls (174.42,155.54) and (172.49,157.71) .. (170.11,157.71) .. controls (167.72,157.71) and (165.79,155.54) .. (165.79,152.86) -- cycle ;
  \draw  [fill={rgb, 255:red, 208; green, 2; blue, 27 }  ,fill opacity=1 ] (270.69,152.86) .. controls (270.69,150.17) and (272.62,148) .. (275,148) .. controls (277.39,148) and (279.32,150.17) .. (279.32,152.86) .. controls (279.32,155.54) and (277.39,157.71) .. (275,157.71) .. controls (272.62,157.71) and (270.69,155.54) .. (270.69,152.86) -- cycle ;
  \draw  [fill={rgb, 255:red, 208; green, 2; blue, 27 }  ,fill opacity=1 ] (375.58,151.71) .. controls (375.58,149.03) and (377.52,146.86) .. (379.9,146.86) .. controls (382.29,146.86) and (384.22,149.03) .. (384.22,151.71) .. controls (384.22,154.4) and (382.29,156.57) .. (379.9,156.57) .. controls (377.52,156.57) and (375.58,154.4) .. (375.58,151.71) -- cycle ;
  \draw  [fill={rgb, 255:red, 208; green, 2; blue, 27 }  ,fill opacity=1 ] (480.48,151.71) .. controls (480.48,149.03) and (482.42,146.86) .. (484.8,146.86) .. controls (487.18,146.86) and (489.12,149.03) .. (489.12,151.71) .. controls (489.12,154.4) and (487.18,156.57) .. (484.8,156.57) .. controls (482.42,156.57) and (480.48,154.4) .. (480.48,151.71) -- cycle ;
  \draw  [fill={rgb, 255:red, 208; green, 2; blue, 27 }  ,fill opacity=1 ] (581.06,150.57) .. controls (581.06,147.89) and (583,145.71) .. (585.38,145.71) .. controls (587.77,145.71) and (589.7,147.89) .. (589.7,150.57) .. controls (589.7,153.25) and (587.77,155.43) .. (585.38,155.43) .. controls (583,155.43) and (581.06,153.25) .. (581.06,150.57) -- cycle ;
  \draw  [dash pattern={on 0.84pt off 2.51pt}]  (65.21,154) -- (170.11,152.86) ;
  \draw  [dash pattern={on 0.84pt off 2.51pt}]  (170.11,154) -- (275,152.86) ;
  \draw  [dash pattern={on 0.84pt off 2.51pt}]  (275,152.86) -- (379.9,151.71) ;
  \draw    (379.9,152.86) -- (484.8,151.71) ;
  \draw    (484.8,151.71) -- (585.38,150.57) ;
   
  \end{tikzpicture}
  
\end{figure}
\\
It has 4 connected components of which  the length vector is  $(1,1,1, 3)$.  Similar to the computation of  \eqref{caculation of coefficient of 111111 by graph}, the coefficient of this summation tuple is 
\begin{equation}
  \label{caculation of coefficient of 121222 by graph} \underbrace{-2\sigma^2 }_{from \ \eqref{A5 graph},\ } \underbrace{+\sigma^2}_{from \ \eqref{A6 graph}} =-\sigma^2. 
\end{equation}
For this  tuple  $(n_1,n_2,n_1,n_2,n_2,n_2)$  (with coefficient $-\sigma^2$), its corresponding term is $G_0 v^2 M_4 v^4 G_0$ or $G_0 v^4 M_4 v^2 G_0$. By the argument of \cite[(3.5)--(3.10)]{Bou03},  it needs  to decompose 
\[M_4=M+(\sigma^3-\rho),\]
since the  control of   $MG_0$ via the convolution regularization technique requires  $\hat M(0)=0$. Hence,  we need to renormalize a  $(\sigma^3-\rho)v^6$ in the potential, so  totally 
\textcolor{blue}{
  \begin{flushright}
    $\Rightarrow {\rm \textbf{Renormalization:} } \ -4\sigma^2(\sigma^3-\rho) v^6. $
  \end{flushright}
}
\noindent Moreover, such operation causes  the non-random term occurring in the 6th order remaining terms:   
\textcolor{red}{
  \begin{flushright}
    $\Rightarrow {\rm \textbf{Non-random term:} } \ -2\sigma^2(\boxed{G_0 v^2 W G_0}+\boxed{G_0 W v^2 G_0}).$
  \end{flushright}
}
On the other hand, from the table  listed above, the tuple $(n_1,n_2,n_2,n_2,n_2,n_1)$  also has  non-zero coefficient $(\sigma^3-\rho)$ and keeps  remaining. It corresponds to  the term $G_0 R_6^{(1)} G_0$ with a diagonal operator
\begin{equation*}
  R_6^{(1)}(n_1)=v\wtg v^4 \wtg v(n_1,n_1)=v^2_{n_1}\sum_{n_2}\wtg(n_1,n_2)^2 v^4_{n_2}.
\end{equation*}
Hence, the renormalization in potential is 
\textcolor{blue}{
  \begin{flushright}
    $\Rightarrow {\rm \textbf{Renormalization:} } \ (\rho-\sigma^3) R_6^{(1)}.$
  \end{flushright}
}
\ \\

\noindent \textcircled{3}{\it The cancelled summation tuple  has exactly three distinct vertices}. Since  the randomness is cancelled,  one  observes  that   each  complete characteristic graph in this case  has at most  3 solid edges, and the length of connected  components  is at most 2.  Similar to the  computation of  \eqref{caculation of coefficient of 111111 by graph}, only the tuple with complete graph\\
\begin{figure}[htbp]
  \centering

\tikzset{every picture/.style={line width=0.75pt}} 

\begin{tikzpicture}[x=0.75pt,y=0.75pt,yscale=-0.7,xscale=0.7]

\draw  [fill={rgb, 255:red, 208; green, 2; blue, 27 }  ,fill opacity=1 ] (60.89,154) .. controls (60.89,151.32) and (62.82,149.14) .. (65.21,149.14) .. controls (67.59,149.14) and (69.52,151.32) .. (69.52,154) .. controls (69.52,156.68) and (67.59,158.86) .. (65.21,158.86) .. controls (62.82,158.86) and (60.89,156.68) .. (60.89,154) -- cycle ;
\draw  [fill={rgb, 255:red, 208; green, 2; blue, 27 }  ,fill opacity=1 ] (165.79,152.86) .. controls (165.79,150.17) and (167.72,148) .. (170.11,148) .. controls (172.49,148) and (174.42,150.17) .. (174.42,152.86) .. controls (174.42,155.54) and (172.49,157.71) .. (170.11,157.71) .. controls (167.72,157.71) and (165.79,155.54) .. (165.79,152.86) -- cycle ;
\draw  [fill={rgb, 255:red, 208; green, 2; blue, 27 }  ,fill opacity=1 ] (270.69,152.86) .. controls (270.69,150.17) and (272.62,148) .. (275,148) .. controls (277.39,148) and (279.32,150.17) .. (279.32,152.86) .. controls (279.32,155.54) and (277.39,157.71) .. (275,157.71) .. controls (272.62,157.71) and (270.69,155.54) .. (270.69,152.86) -- cycle ;
\draw  [fill={rgb, 255:red, 208; green, 2; blue, 27 }  ,fill opacity=1 ] (375.58,151.71) .. controls (375.58,149.03) and (377.52,146.86) .. (379.9,146.86) .. controls (382.29,146.86) and (384.22,149.03) .. (384.22,151.71) .. controls (384.22,154.4) and (382.29,156.57) .. (379.9,156.57) .. controls (377.52,156.57) and (375.58,154.4) .. (375.58,151.71) -- cycle ;
\draw  [fill={rgb, 255:red, 208; green, 2; blue, 27 }  ,fill opacity=1 ] (480.48,151.71) .. controls (480.48,149.03) and (482.42,146.86) .. (484.8,146.86) .. controls (487.18,146.86) and (489.12,149.03) .. (489.12,151.71) .. controls (489.12,154.4) and (487.18,156.57) .. (484.8,156.57) .. controls (482.42,156.57) and (480.48,154.4) .. (480.48,151.71) -- cycle ;
\draw  [fill={rgb, 255:red, 208; green, 2; blue, 27 }  ,fill opacity=1 ] (581.06,150.57) .. controls (581.06,147.89) and (583,145.71) .. (585.38,145.71) .. controls (587.77,145.71) and (589.7,147.89) .. (589.7,150.57) .. controls (589.7,153.25) and (587.77,155.43) .. (585.38,155.43) .. controls (583,155.43) and (581.06,153.25) .. (581.06,150.57) -- cycle ;
\draw  [dash pattern={on 0.84pt off 2.51pt}]  (65.21,154) -- (170.11,152.86) ;
\draw  [dash pattern={on 0.84pt off 2.51pt}]  (170.11,154) -- (275,152.86) ;
\draw  [dash pattern={on 0.84pt off 2.51pt}]  (275,152.86) -- (379.9,151.71) ;
\draw  [dash pattern={on 0.84pt off 2.51pt}]  (379.9,152.86) -- (484.8,151.71) ;
\draw  [dash pattern={on 0.84pt off 2.51pt}]  (484.8,151.71) -- (585.38,150.57) ;

\end{tikzpicture}

\end{figure} 

\noindent has non-zero coefficient $1$ (from \eqref{A6 graph}). In other words,  each adjacent pair contains distinct vertices. Tuples satisfying the above restrictions  must be 
\begin{itemize}
  \item $(n_1,n_2,n_3,n_2,n_3,n_1)$. It corresponds to  the diagonal operator 
       \begin{equation*}
        R_6^{(2)}(n_1)=v \wtg W_4 \wtg v(n_1,n_1) =v^2_{n_1}\sum_{n_2,n_3\in \Z^d} v^2_{n_2}v^2_{n_3}\wtg(n_1,n_2)\wtg(n_2,n_3)^3\wtg(n_3,n_1),
       \end{equation*} 
which leads to  the renormalization in potential: 
\textcolor{blue}{
  \begin{flushright}
    $\Rightarrow {\rm \textbf{Renormalization:} } \ R_6^{(2)}.$
  \end{flushright}
}
  \item $(n_1,n_2,n_3,n_1,n_2,n_3),(n_1,n_2,n_3,n_1,n_3,n_2),(n_1,n_2,n_1,n_3,n_2,n_3)$ and\\
      $(n_1,n_2,n_3,n_2,n_1,n_3)$.  Those tuples correspond to  the term $G_0 C_6 G_0$, where $C_6$ is a non-diagonal operator 
      \begin{equation}\label{A.10}
              C_6(n_1,n_3)=v^2_{n_1}v^2_{n_3}\wtg(n_1,n_3) \sum_{n_2}\wtg(n_1,n_2)^2 v^2_{n_2} \wtg(n_2,n_3)^2.
      \end{equation}
      Note  that $C_6$ is a non-convolutional  operator. If  discarding  all $v_{\cdot}$,   we will get a convolution operator 
      \[N(n_1,n_3)=\wtg(n_1,n_3) \sum_{n_2}\wtg(n_1,n_2)^2  \wtg(n_2,n_3)^2,\]
       which will enable us to control  $G_0(C_6-\eta v^6)G_0$ well. Thus, we can  decompose 
      \[C_6=(C_6-\eta v^6)+\eta v^6=C+\eta v^6,\]
      which leads to the  renormalization in potential 
      \textcolor{blue}{
  \begin{flushright}
    $\Rightarrow {\rm \textbf{Renormalization:} } \ 4\eta v^6,$
  \end{flushright}
    }
 and the non-random term occurring in the 6th-order remaining terms  
      \textcolor{red}{
  \begin{flushright}
    $\Rightarrow {\rm \textbf{Non-random term:} } \ 4\boxed{ G_0 C G_0}.$
  \end{flushright}
}

\end{itemize}

Combining  all the  above renormalizations in potential  (marked in blue color) and the non-random remaining terms  (marked in red color)  yields the 6th-order renormalization potential: 
\begin{align}\label{offset 1}
  V^{(6)}_{\omega} &= V^{(4)}_{\omega} +((\sigma^5+\rho \sigma^2)-4\sigma^2(\sigma^3-\rho)+4\eta)v^6+(\rho-\sigma^3)R_6^{(1)}+R_6^{(2)}\\
\notag &= V^{(4)}_{\omega} +(4\eta-3\sigma^5+5\rho\sigma^2)v^6 + R_6,
\end{align} 
which matches with \eqref{6-order renormalization of potential}. And the non-randomness term in the 6th-order remaining terms  is 
\[4\boxed{ G_0 C G_0}-2\sigma^2(\boxed{G_0 v^2 W G_0}+\boxed{G_0 W v^2 G_0}),\]
which also matches with \eqref{4.9}.

\subsection{Random part in the 6th-order remaining terms}\label{6thsub2}
In this subsection,  we will calculate the terms $\eqref{4.4}\sim \eqref{4.8}$.  Unlike the calculation of non-random terms, which focuses on the distribution of repeated vertices, here we primarily focus on the information about the connected components of complete characteristic graphs.

We now characterize each complete characteristic graph with a sequence $<a,b,\cdots>$, by writing down the lengths of all connected components in order. For example, the sequence $<3,1,2,1>$ stands for the following complete characteristic graph:
\begin{figure}[htbp]
  \centering

  \tikzset{every picture/.style={line width=0.75pt}} 

  \begin{tikzpicture}[x=0.75pt,y=0.75pt,yscale=-0.7,xscale=0.7]
  
  \draw  [fill={rgb, 255:red, 208; green, 2; blue, 27 }  ,fill opacity=1 ] (9.89,146.86) .. controls (9.89,144.17) and (11.82,142) .. (14.21,142) .. controls (16.59,142) and (18.52,144.17) .. (18.52,146.86) .. controls (18.52,149.54) and (16.59,151.71) .. (14.21,151.71) .. controls (11.82,151.71) and (9.89,149.54) .. (9.89,146.86) -- cycle ;
  \draw  [fill={rgb, 255:red, 208; green, 2; blue, 27 }  ,fill opacity=1 ] (114.79,145.71) .. controls (114.79,143.03) and (116.72,140.86) .. (119.11,140.86) .. controls (121.49,140.86) and (123.42,143.03) .. (123.42,145.71) .. controls (123.42,148.4) and (121.49,150.57) .. (119.11,150.57) .. controls (116.72,150.57) and (114.79,148.4) .. (114.79,145.71) -- cycle ;
  \draw  [fill={rgb, 255:red, 208; green, 2; blue, 27 }  ,fill opacity=1 ] (219.69,145.86) .. controls (219.69,143.17) and (221.62,141) .. (224,141) .. controls (226.39,141) and (228.32,143.17) .. (228.32,145.86) .. controls (228.32,148.54) and (226.39,150.71) .. (224,150.71) .. controls (221.62,150.71) and (219.69,148.54) .. (219.69,145.86) -- cycle ;
  \draw  [fill={rgb, 255:red, 208; green, 2; blue, 27 }  ,fill opacity=1 ] (324.58,144.57) .. controls (324.58,141.89) and (326.52,139.71) .. (328.9,139.71) .. controls (331.29,139.71) and (333.22,141.89) .. (333.22,144.57) .. controls (333.22,147.25) and (331.29,149.43) .. (328.9,149.43) .. controls (326.52,149.43) and (324.58,147.25) .. (324.58,144.57) -- cycle ;
  \draw  [fill={rgb, 255:red, 208; green, 2; blue, 27 }  ,fill opacity=1 ] (429.48,144.57) .. controls (429.48,141.89) and (431.42,139.71) .. (433.8,139.71) .. controls (436.18,139.71) and (438.12,141.89) .. (438.12,144.57) .. controls (438.12,147.25) and (436.18,149.43) .. (433.8,149.43) .. controls (431.42,149.43) and (429.48,147.25) .. (429.48,144.57) -- cycle ;
  \draw  [fill={rgb, 255:red, 208; green, 2; blue, 27 }  ,fill opacity=1 ] (530.06,143.57) .. controls (530.06,140.89) and (532,138.71) .. (534.38,138.71) .. controls (536.77,138.71) and (538.7,140.89) .. (538.7,143.57) .. controls (538.7,146.25) and (536.77,148.43) .. (534.38,148.43) .. controls (532,148.43) and (530.06,146.25) .. (530.06,143.57) -- cycle ;
  \draw    (14.21,146.86) -- (119.11,147) ;
  \draw    (119.11,147) -- (224,145.86) ;
  \draw  [dash pattern={on 0.84pt off 2.51pt}]  (224,145.86) -- (328.9,144.71) ;
  \draw  [dash pattern={on 0.84pt off 2.51pt}]  (328.9,144.57) -- (438.12,144.57) ;
  \draw    (433.8,144.71) -- (534.38,143.57) ;
  \draw  [fill={rgb, 255:red, 208; green, 2; blue, 27 }  ,fill opacity=1 ] (630.64,142.29) .. controls (630.64,139.6) and (632.58,137.43) .. (634.96,137.43) .. controls (637.35,137.43) and (639.28,139.6) .. (639.28,142.29) .. controls (639.28,144.97) and (637.35,147.14) .. (634.96,147.14) .. controls (632.58,147.14) and (630.64,144.97) .. (630.64,142.29) -- cycle ;
  \draw  [dash pattern={on 0.84pt off 2.51pt}]  (534.38,143.57) -- (634.96,142.29) ;

  \end{tikzpicture}
 
\end{figure}
\ \\ 
Obviously, if a summation tuple (with a complete characteristic graph) has randomness (i.e.,  does not cancel),  then there must be some connected  component (of its complete characteristic graph)  having  the  {\it odd} number  length.  By this fact, we  can  list
  \begin{center}
    \begin{xltabular}[h]{\textwidth}{|>{\centering\arraybackslash}X |>{\centering\arraybackslash}X | >{\centering\arraybackslash}X|}
     \hline
     \textbf{Number of connected  components}  & \textbf{Characteristic graphs} & \textbf{Coefficient} \\
      \hline
      2 & $<1,5>,<5,1>$ & $2\rho \sigma$ \\
      \hline
      \textcolor{orange}{3} & \textcolor{orange}{$<4,1,1>,<1,4,1>$, $<1,1,4>$} & \textcolor{orange}{$\rho-\sigma^3$} \\
      \hline
      2 & $<3,3>$ & $\sigma^4$ \\
      \hline
      3 & $<3,2,1>,<3,1,2>$, $<2,1,3>,<2,3,1>$, $<1,3,2>,<1,2,3>$ & 0\\
      \hline
      4 & $<3,1,1,1>,<1,3,1,1>$, $<1,1,3,1>,<1,1,1,3>$ & $-\sigma^2$\\
      \hline
      4 & $<2,2,1,1>,\cdots$  & 0 \\
      \hline
      5 & $<2,1,1,1,1>,\cdots$ & 0 \\
      \hline
      6 & $<1,1,1,1,1,1>$ & 1\\
      \hline
     \end{xltabular}
\end{center}
\ \\ 
Also, the coefficients are calculated similar to that of  \eqref{caculation of coefficient of 111111 by graph}. 
\begin{rmk}
All possible complete characteristic graphs can be exhausted as follows. We have already known  that at least one connected  component has an odd number  length.
\begin{itemize}
  \item[(1)] First, consider the graphs with at least a 5-length connected  component. The only possibilities  are  $<5,1>$ and its permutation $<1,5>$; 
  
  \item[(2)] Second, consider the graphs with no 5-length connected  component, but with at least a 4-length connected  component. The only possibilities  are $<4,1,1>$ and its permutations; 
  
  \item[(3)] Third, consider the graphs with no 4-length connected  component, but with at least a 3-length connected  component. The graphs consist of:     $<3,3>$,     $<3,2,1>$ with  its permutations,  and $<3,1,1,1>$ with  its permutations; 
  
  \item[(4)] Fourth, consider the graphs with no 3-length connected  component, but  with at least a 2-length connected  component. The graphs  consist of:  $<2,2,1,1>$ with  its permutations,  and $<2,1,1,1,1>$  with  its permutations; 
  
  \item[(5)] Finally, consider the graphs with only 1-length connected  components. It is just $<1,1,1,1,1,1>$.
\end{itemize} 
Such  exhaustions  argument also works  well for the 7th-order terms, but becomes more complicated.
\end{rmk}

Next,  we discuss the random parts   in   terms with non-zero coefficients listed  in the above table: 
\begin{itemize}
  \item $<1,5>,<5,1>$:  They are automatically random and correspond  to 
        \begin{align*} 
          2\rho\sigma &(G_0 V\wtg v^4VG_0+G_0 v^4V\wtg VG_0)\\
          &= 2\sigma\rho \big((G_0 v^4 VG_0 V G_0)^{\admissible} + (G_0  VG_0  v^4 V G_0)^{\admissible}  \big), 
        \end{align*}
        which is the first term in \eqref{4.4}.
        
  \item $<3,3>$: This is automatically random and corresponds to 
        \[\sigma^4 G_0 v^2V \wtg v^2VG_0=\sigma^4(G_0 v^2V G_0 v^2 V G_0)^{\admissible},\]
        which is the second term in \eqref{4.4}.
        
  \item $<3,1,1,1>,<1,3,1,1><1,1,3,1>,<1,1,1,3>$: The random part is exactly \eqref{4.5}. For example, the random part of $<3,1,1,1>$ is just the tuple $(n_2,n_2,n_2,n_1,n_2,n_1)$, which has been discussed in \eqref{caculation of coefficient of 121222 by graph}.
  
  \item $<1,1,1,1,1,1>$,  and $<4,1,1>$ with its permutations: This case is the most challenging because we need to determine the connections between  ``tuples with adjacent elements different''  and ``admissible'' ones. In other words, the problem  is how to rewrite
        \[ G_0 V\wtg V\wtg V\wtg V\wtg V\wtg VG_0 = (\cdots)^{\admissible}+(\cdots)^{\admissible}+\cdots.\]
        Such things  also cause the main obstacle in the calculations of  the 7th-order terms.\\ 
         Now,  since  we have already ensured that the 6-tuple does  not cancel,   the condition of  ``with adjacent elements different''   can guarantee all 1-subtuples, 2-subtuples, 3-subtuples, and 5-subtuples do not cancel  (note that odd (number)-tuple never cancels). Therefore, we only need to consider what tuples  could  carry the cancelled 4-tuples. 
        What's more, the cancelled 4-tuples  with adjacent elements different can only be $(n_1,n_2,n_1,n_2)$. Hence (we remark   that the tuples below  are all with adjacent elements different), we have 
        \begin{align}\label{<1,1,1,1,1,1>}
          <1,1,1,1,1,1> =&<\underbrace{1,1,1,1}_{ \admissible },1,1>+<\underbrace{1,1,1,1}_{cancel},1,1>\\
            \notag =&<(1,1,1,1)^{\admissible},1,1>+(n_1,n_2,n_1,n_2,X,X)\\
            \notag =&<(1,\underbrace{1,1,1)^{\admissible},1}_{\admissible},1>+<(1,\underbrace{1,1,1)^{\admissible},1}_{cancel},1>\\
              \notag& +(n_1,n_2,n_1,n_2,X,X)\\
            \label{explanation 1} =&<(1,1,1,1,1)^{\admissible},1>+<1,\underbrace{1,1,1,1}_{cancel},1>\\
            \notag & -(n_1,n_2,n_1,n_2,n_1,X)+(n_1,n_2,n_1,n_2,X,X)\\
            \notag =&<(1,1,\underbrace{1,1,1)^{\admissible},1}_{\admissible}>+<(1,1,\underbrace{1,1,1)^{\admissible},1}_{cancel}>\\
            \notag &+<1,\underbrace{1,1,1,1}_{cancel},1>-(n_1,n_2,n_1,n_2,n_1,X)+(n_1,n_2,n_1,n_2,X,X)\\
            \label{explanation 2} =&<(1,1,1,1,1,1)^{\admissible}>+(X,X,n_1,n_2,n_1,n_2)-(X,n_2,n_1,n_2,n_1,n_2)\\
            \notag &+(X,n_2,n_1,n_2,n_1,X)-(n_1,n_2,n_1,n_2,n_1,X)+(n_1,n_2,n_1,n_2,X,X). 
        \end{align}
        Here, we provide some explanations for the above computations.  The notation ``$X$"  indicates that, ensuring the tuple has different adjacent elements, the place can be any vertex. For the equality \eqref{explanation 1}, consider a tuple (with different adjacent elements) whose first 4-subtuple does not cancel, but the second 4-subtuple does. This can be understood as the tuple, whose second 4-subtuple cancels, minus the tuple whose first and second 4-subtuple both cancel. By this argument, we have 
        \begin{equation}\label{A.16}
          <(1,\underbrace{1,1,1)^{\admissible},1}_{cancel},1>=<1,\underbrace{1,1,1,1}_{cancel},1>-(n_1,n_2,n_1,n_2,n_1,X). 
        \end{equation}
        For the equality \eqref{explanation 2}, consider the tuple whose first and second 4-subtuples  do not cancel, but the third 4-subtuple does. Such a tuple must be of the  form 
        $$(X,\tilde{X},n_1,n_2,n_1, n_2).$$
        Since  the second 4-subtuple does not cancel and the adjacent elements are different,  the second $\tilde{X}$ can only be $n_3$ and hence the first 4-subtuple automatically does not cancel.   
        So it is equivalent to consider the tuple whose second 4-subtuple do not cancel, but the third 4-subtuple does. Then using the similar argument as in  \eqref{A.16} will give  us
        \[<(1,1,\underbrace{1,1,1)^{\admissible},1}_{cancel}>=(X,X,n_1,n_2,n_1,n_2)-(X,n_2,n_1,n_2,n_1,n_2).\]
        Finally, rewriting  \eqref{<1,1,1,1,1,1>} into  the operator summation form yields 
        \begin{align*}
          &\ \ \ G_0 V\wtg V\wtg V\wtg V\wtg V\wtg VG_0 \\
           =& \big[   (G_0 W_4 \widetilde{G_0} V G_0 V G_0)^{\admissible}  +(G_0 V \wtg W_4 \wtg V G_0)^{\admissible}\\
             & +(G_0 V G_0 V \wtg W_4 G_0)^{\admissible}      \big]+\eqref{4.6}+\eqref{4.8}. 
        \end{align*}
        Note that the first term on the RHS of the above equality  is not  \eqref{4.7}  (i.e., with $W$ replaced by $W_4$.) {\bf Fortunately, recall that we have not yet  considered   the graphs $<4,1,1>,<1,4,1>$ and $<1,1,4>$ with  the coefficient $\rho-\sigma^3$}. Indeed,  those terms take the form 
        \begin{equation}\label{A.17}
              (\rho-\sigma^3)\big[   (G_0  v^4 \widetilde{G_0} V G_0 V G_0)^{\admissible}  +(G_0 V \wtg v^4 \wtg V G_0)^{\admissible} +(G_0 V G_0 V \wtg v^4 G_0)^{\admissible}      \big]. 
        \end{equation}
        Hence,
        \begin{align}\label{offset 2}
          \big[   (G_0 W_4 \widetilde{G_0} V G_0 V G_0)^{\admissible}  & +(G_0 V \wtg W_4 \wtg V G_0)^{\admissible} +(G_0 V G_0 V \wtg W_4 G_0)^{\admissible}      \big]\\
        \notag     &+\eqref{A.17} =\eqref{4.7}. 
        \end{align}
        We finally get  \eqref{4.7} successfully!    Such  {\bf offset}  is amazing  and will also  appear in   the computations  of 7th-order remaining terms.
\end{itemize}
By summarizing all the above discussions, we prove that the random part of the 6th-order remaining terms  is given by  $\eqref{4.4}\sim \eqref{4.8}$.

\section{Computation of  the 7th-order remaining terms}\label{App7th}

The computations of the 7th-order remaining terms $\eqref{4.11}\sim \eqref{4.24}$ are  much more complicated, but follow a similar procedure as in the  6th-order case.  

In the 7th-order case,  all terms  are random.  Therefore, the 7th-order terms  do not require additional renormalizations on the potential.  Similar to Subsection \ref{6thsub2}, we list

 \begin{center}
    \begin{xltabular}[h]{\textwidth}{|>{\centering\arraybackslash}X |>{\centering\arraybackslash}X | >{\centering\arraybackslash}X|}
     \hline
     \textbf{Number of connected components}  & \textbf{Characteristic graphs/ terms} & \textbf{Coefficient} \\
      \hline
      \qquad & $\mathcolor{orange}{G_0 R_6\wtg V G_0, G_0 V\wtg R_6 G_0}$ & $\mathcolor{orange}{1}$ \\
      \hline
      \qquad &  $G_0 R_6 V G_0$  & $2\sigma$ \\
      \hline
      1 & $<7>$ & $8\eta\sigma-7\sigma^6+12\sigma^3\rho$ \\
      \hline
      \color{orange}{2} & $\mathcolor{orange}{<6,1>,<1,6>}$ & $\mathcolor{orange}{4\eta +4\sigma^2\rho-4\sigma^5}$ \\
      \hline
      2 & $<5,2>,<2,5>$ & 0\\
      \hline
      3 & $<5,1,1,>,<1,5,1>$, $<1,1,5>$ & $-2\sigma\rho$\\
      \hline
      \color{orange}{2} & $\mathcolor{orange}{<4,3>,<3,4>}$  & $\mathcolor{orange}{\sigma^2(\rho-\sigma^3)}$ \\
      \hline
      3 & $<4,2,1>,\cdots$ & 0 \\
      \hline
      \color{orange}{4} & $\mathcolor{orange}{<4,1,1,1>,<1,4,1,1>,}$  $\mathcolor{orange}{<1,1,4,1>,<1,1,1,4>}$ & $\mathcolor{orange}{\sigma^3-\rho}$ \\
      \hline
      3 & $<3,3,1>,<3,1,3>$, $<1,3,3>$ & $-\sigma^4$\\
      \hline
      3 & $<3,2,2>,\cdots$ & 0 \\
      \hline
      4 & $<3,2,1,1>,\cdots$ & 0 \\
      \hline
      5 & $<3,1,1,1,1>,\cdots$ & $\sigma^2$ \\
      \hline
      \qquad & graphs  have  only connected  components of length   2 or 1, and at least one  2-length connected  component & 0\\
      \hline
      7 & $<1,1,1,1,1,1,1>$ & -1 \\
      \hline
     \end{xltabular}
\end{center}
 
Next, we discuss the random part in  terms with non-zero coefficients  in the above table. We remark  that, by our definition, the vertices in two adjacent connected  components are different. The following items only consider \textbf{the graphs with  connected  components  of  an odd number}. 
\begin{itemize}
  \item $G_0 R_6 V G_0$: It is automatically random, and  is just the first term in \eqref{4.11}.
  \item $<7>$: It is automatically random, and corresponds to 
        \[(8\eta\sigma-7\sigma^6+12\sigma^3\rho)G_0 v^6 VG_0, \]
        which is the second term in \eqref{4.11}.
  \item $<5,1,1>,<1,5,1>,<1,1,5>$: They  are automatically random, admissible, and correspond exactly  to  \eqref{4.12}.
  \item $<3,3,1>,<3,1,3>,<1,3,3>$: They are automatically random, admissible, and correspond  exactly to  \eqref{4.13}.
  \item $<3,1,1,1,1>$ with  its permutations: In this case,  we also need do the same analysis as in  \eqref{<1,1,1,1,1,1>}. For example, figure out how to connect $G_0v^2 V\wtg V\wtg V\wtg V\wtg VG_0$ with summations  on  admissible tuples. For simplicity,  we first consider 
        \begin{align}\label{<3,1,1,1,1>}
          <1,1,1,1,1>  =& <(1,1,1,1)^{\admissible},1>+<\underbrace{1,1,1,1}_{cancel},1> \\
          \notag  =& <(1,1,1,1,1)^{\admissible}> + <(1,\underbrace{1,1,1)^{\admissible},1}_{cancel}>+ (n_1,n_2,n_1,n_2,X)\\
          \notag  =&<(1,1,1,1,1)^{\admissible}> + <1,\underbrace{1,1,1,1}_{cancel}>-(n_1,n_2,n_1,n_2,n_1) \\
          \notag & + (n_1,n_2,n_1,n_2,X)\\
          \notag =&<(1,1,1,1,1)^{\admissible}> + (n_1,n_2,n_1,n_2,X)+(X,n_1,n_2,n_1,n_2) \\
          \notag & -(n_1,n_2,n_1,n_2,n_1). 
        \end{align}
        Since the considered  graphs  are $<3,1,1,1,1>$ with  its permutations, we can  replace  one of the 1-length connected  component in \eqref{<3,1,1,1,1>}  with a 3-length connected  component. 
        
        If we replace each connected  component  of  $<(1,1,1,1,1)^{\admissible}>$   (remember  the coefficient is $\sigma^2$), we get  
        \begin{align*}
            \sigma^2 \big( <(3,1,1,1,1)^{\admissible}> + & <(1,3,1,1,1)^{\admissible}> + <(1,1,3,1,1)^{\admissible}> \\
               &+ <(1,1,1,3,1)^{\admissible}> + <(1,1,1,1,3)^{\admissible}> \big),  
        \end{align*}
        which  corresponds  exactly  to \eqref{4.14}.

        If we replace  each  connected  component in $(n_1,n_2,n_1,n_2,X)+(X,n_1,n_2,n_1,n_2)$, we obtain  
        \begin{align}\label{B.2}
          \sigma^2\left[ (n_1,n_2,n_1,n_2,X=X=X)+(X=X=X,n_1,n_2,n_1,n_2)  \right],
        \end{align}
        and
        \begin{align}\label{B.3}
          \sigma^2 \biggl[ &\big( (n_1,n_1,n_1,n_2,n_1,n_2,X)+ (n_1,n_2,n_1,n_1,n_1,n_2,X)   \big) \\
           \notag  &+  \big( (n_1,n_2,n_2,n_2,n_1,n_2,X)+ (n_1,n_2,n_1,n_2,n_2,n_2,X)   \big) \\
           \notag & + \big( (X,n_1,n_1,n_1,n_2,n_1,n_2)+ (X,n_1,n_2,n_1,n_1,n_1,n_2)   \big) \\
           \notag & +\big( (X,n_1,n_2,n_2,n_2,n_1,n_2)+ (X,n_1,n_2,n_1,n_2,n_2,n_2)   \big) \biggl]. 
        \end{align}
        Then we  observe that \eqref{B.2} corresponds to
        \begin{align}\label{B.2 correspond}
          &\ \ \ \sigma^2  (G_0 W_4 \wtg v^2V G_0 +G_0 v^2 V \wtg W_4 G_0) \\
          \notag &=\eqref{4.15} + \mathcolor{orange}{\sigma^2(\sigma^3-\rho)\cdot (G_0 v^4 \wtg v^2V G_0 +G_0 v^2 V \wtg v^4 G_0)}\\
          \notag &=\eqref{4.15} + \mathcolor{orange}{\sigma^2(\sigma^3-\rho)(<4,3>+<3,4>)}, 
        \end{align}
        and \eqref{B.3} corresponds to 
        \begin{align}\label{B.3 correspond}
          &\ \ \  2\sigma^2 (G_0 W_4 v^2 \wtg VG_0+ G_0 v^2 W_4 \wtg V G_0 +G_0 V \wtg v^2 W_4 G_0+ G_0 V \wtg W_4 v^2 G_0) \\
              \notag     &= \eqref{4.16} + \mathcolor{orange}{4\sigma^2(\sigma^3-\rho) (G_0 v^6 \wtg V G_0+G_0 V \wtg v^6 G_0) }\\
              \notag     &= \eqref{4.16} + \mathcolor{orange}{4\sigma^2(\sigma^3-\rho) (<6,1>+<1,6>)}. 
        \end{align}
        If we replace  each connected  component in $-(n_1,n_2,n_1,n_2,n_1)$, we get  
        \begin{align}\label{B.6}
          -\sigma^2 \big[ (n_1,n_1,n_1,n_2,n_1,n_2,n_1) &+(n_1,n_2,n_1,n_1,n_1,n_2,n_1) \\
           \notag &+(n_1,n_2,n_1,n_2,n_1,n_1,n_1) \big]
        \end{align}
        and 
        \begin{align}\label{B.7}
          -\sigma^2\big[  (n_1,n_2,n_2,n_2,n_1,n_2,n_1)+(n_1,n_2,n_1,n_2,n_2,n_2,n_1)        \big]. 
        \end{align}
        Then we can observe that \eqref{B.6} corresponds  exactly  to  the first term in \eqref{4.17}, and \eqref{B.7}   to   the second one   in \eqref{4.17}.  

        Finally,  combining computations concerning  $<3,1,1,1,1>$  with  its permutations yields  
        \begin{align*}
            &[\eqref{4.14} \sim \eqref{4.17} ]\\
            &  \ \ \ + \mathcolor{orange}{4\sigma^2(\sigma^3-\rho) (<6,1>+<1,6>)} +\mathcolor{orange}{\sigma^2(\sigma^3-\rho)(<4,3>+<3,4>)}. 
        \end{align*}
  \item $<1,1,1,1,1,1,1>$:  We try to rewrite the 7-tuples   with different adjacent elements into summations  on  admissible tuples. That is, to investigate 
        \[G_0 V \wtg V\wtg V\wtg V\wtg V\wtg V\wtg VG_0.\]
        First, we can  decompose 
        \begin{align*}
          <1,1,1,1,1,1,1> &= <(1,1,1,1,1,1,1)^{\admissible}> + \sum_{non \ admissible,  \atop with\  adjacent \ elements \ different}\cdots\\
          &=<(1,1,1,1,1,1,1)^{\admissible}>+A+B\\
          &\Rightarrow -\eqref{4.18}+A+B, 
        \end{align*}
        where   $A$ contains  tuples having cancelled 4-subtuples, and  $B$ contains  tuples having cancelled 6-subtuples, but no cancelled 4-subtuple. 

We first deal with $A$, and then $B$: \smallskip 

       \noindent \textbf{(Term A)}\\
       We decompose 
        \begin{align*}
          &A = <\underbrace{1,1,1,1}_{cancel},1,1,1> + < \mathcolor{red}{1,}\underbrace{\mathcolor{red}{1,1,1},1}_{cancel},1,1> + < \mathcolor{red}{1,1,}\underbrace{\mathcolor{red}{1,1,1},1}_{cancel},1> \\
            &\qquad + < \mathcolor{red}{1,1,1,}\underbrace{\mathcolor{red}{1,1,1},1}_{cancel}> \\
            &= (n_1,n_2,n_1,n_2,X,X,X) + [(X,n_1,n_2,n_1,n_2,X,X)-(n_2,n_1,n_2,n_1,n_2,X,X)] \\
            &\qquad +[(X,X,n_1,n_2,n_1,n_2,X)-(X,n_2,n_1,n_2,n_1,n_2,X)] \\
            &\qquad +[(X,X,X,n_1,n_2,n_1,n_2)-(X,X,n_2,n_1,n_2,n_1,n_2)+(n_2,n_1,n_2,n_1,n_2,n_1,n_2)\\
            &\qquad \qquad \qquad \qquad - < \underbrace{1,1,1,}_{cancel}1\underbrace{,1,1,1}_{cancel}>]\\
            &=\\
            \label{part 1}\tag{A-part 1}    &[(X,X,X,n_1,n_2,n_1,n_2)+(X,X,n_1,n_2,n_1,n_2,X)+\\
            &\qquad (X,n_1,n_2,n_1,n_2,X,X)+  (n_1,n_2,n_1,n_2,X,X,X) ]\\
            \label{part 2}\tag{A-part 2}&\qquad -[(n_2,n_1,n_2,n_1,n_2,X,X)+(X,n_2,n_1,n_2,n_1,n_2,X)+(X,X,n_2,n_1,n_2,n_1,n_2)]\\
            \label{part 3}\tag{A-part 3}&\qquad + (n_2,n_1,n_2,n_1,n_2,n_1,n_2)-< \underbrace{1,1,1,}_{cancel}1\underbrace{,1,1,1}_{cancel}>. 
        \end{align*}
        Here the red part of the graph means that there is no cancelled 4-subtuple.\\
        Now for \eqref{part 1}, we have 
        \begin{align}
         \label{B.8} (X,X,X,n_1,n_2,n_1,n_2)& = (X,X,X,\underbrace{n_1,n_2,n_1,n_2}_{W}) + \mathcolor{orange}{(\sigma^3-\rho)\cdot <1,1,1,4>}\\
        \notag  &\Rightarrow (G_0 VG_0 VG_0 V \wtg W_4 G_0)^{\admissible} +  \mathcolor{orange}{(\sigma^3-\rho)\cdot <1,1,1,4>}
        \end{align}
        and 
        \begin{align}
         \label{B.9} &(X,X,n_1,n_2,n_1,n_2,X) = (X,X,\underbrace{n_1,n_2,n_1,n_2}_{W},X) + \mathcolor{orange}{(\sigma^3-\rho)\cdot <1,1,4,1>}\\
         \notag &\Rightarrow (G_0 V G_0 V \wtg W \wtg V G_0)^{\admissible}+\mathcolor{orange}{(X,n_3,\underbrace{n_2,n_1,n_2,n_1}_{W},n_3)}+\mathcolor{orange}{(\sigma^3-\rho)\cdot <1,1,4,1>}\\
         \notag & =(G_0 V G_0 V \wtg W \wtg V G_0)^{\admissible}+\mathcolor{orange}{G_0 V \wtg R_6 G_0}+\mathcolor{orange}{(\sigma^3-\rho)\cdot <1,1,4,1>}. 
        \end{align}
        Similar argument also applies to  $(n_1,n_2,n_1,n_2,X,X,X)$ and $(X,n_1,n_2,n_1,n_2,X,X)$. Then we obtain
        \begin{align}
        \label{B.10}  \eqref{part 1} &\Rightarrow -\eqref{4.19} + \mathcolor{orange}{(G_0 V \wtg R_6 G_0+G_0 R_6 \wtg V G_0)}\\
             \notag                         &\qquad +\mathcolor{orange}{(\sigma^3-\rho)\cdot (<1,1,1,4>+<1,1,4,1>+\cdots)}
        \end{align}
For \eqref{part 2}, we have 
        \[(n_1,n_2,n_1,n_2,n_1,X,X)\Rightarrow (G_0 VD_4 G_0 V G_0VG_0)^{\admissible},\]
        so,
        \begin{align}
          \label{B.11} \eqref{part 2}\Rightarrow -\eqref{4.20}. 
        \end{align}
        For \eqref{part 3}, we have 
        \begin{align}
          \label{B.12} \eqref{part 3}\Rightarrow -\eqref{4.21}.
        \end{align}
 Hence, combining   \eqref{B.10}, \eqref{B.11} and \eqref{B.12} shows 
        \begin{align}\label{A generate}
          A \Rightarrow -[\eqref{4.19}\sim\eqref{4.21}]&+\mathcolor{orange}{(G_0 V \wtg R_6 G_0+G_0 R_6 \wtg V G_0)}\\
          \notag &+\mathcolor{orange}{(\sigma^3-\rho)\cdot (<1,1,1,4>+<1,1,4,1>+\cdots)}. 
        \end{align}
        
       \textbf{(Term B)}\\
        Remember that $B$ contains tuples with different adjacent elements,  having no cancelled 4-subtuple and  with at least one cancelled 6-subtuple. Hence, the cancelled 6-subtuples  it containing must be of the  form 
        \[(n_1,n_2,n_3,n_1,n_2,n_3),\ (n_1,n_3,n_2,n_1,n_2,n_3),\]
        \[(n_1,n_2,n_3,n_2,n_1,n_3),\ (n_1,n_2,n_1,n_3,n_2,n_3).\]
        The case $(n_1,n_2,n_3,n_2,n_3,n_1)$ can be  excluded, because it contains the  cancelled 4-subtuple $(n_2,n_3,n_2,n_3)$. Now we decompose (keeping in mind that $B$ contains  no cancelled 4-subtuple)
        \begin{align*}
          B& =<\underbrace{1,1,1,1,1,1}_{cancel},1>+<\mathcolor{green}{1,}\underbrace{\mathcolor{green}{1,1,1,1,1},1}_{cancel}>\\
           &= <\underbrace{1,1,1,1,1,1}_{cancel},1>+<1,\underbrace{1,1,1,1,1,1}_{cancel}> \\
           &\qquad -[(n_3,n_1,n_2,n_3,n_1,n_2,n_3)+(n_3,n_1,n_2,n_3,n_1,n_3,n_2)\\
           &\qquad \qquad +(n_3,n_1,n_2,n_3,n_2,n_1,n_3)+(n_3,n_1,n_2,n_1,n_3,n_2,n_3)]\\
           &=(\mathbf{all \ tuples \ below \ contain  \ no \ cancelled  \ 4-subtuple})\\
         \label{B-part 1}\tag{B-part 1} & \quad \bigg\{ [(n_1,n_2,n_3,n_1,n_2,n_3,X)+(X,n_1,n_2,n_3,n_1,n_2,n_3)] \\
          &\qquad \qquad + [(n_1,n_3,n_2,n_1,n_2,n_3,X)+(X,n_1,n_3,n_2,n_1,n_2,n_3)]\\
          &\qquad \qquad + [(n_1,n_2,n_3,n_2,n_1,n_3,X)+(X,n_1,n_2,n_3,n_2,n_1,n_3)]\\
          &\qquad \qquad + [(n_1,n_2,n_1,n_3,n_2,n_3,X)+(X,n_1,n_2,n_1,n_3,n_2,n_3)]\bigg\}\\
          \label{B-part 2}\tag{B-part 2} & \quad -[(n_3,n_1,n_2,n_3,n_1,n_2,n_3)+(n_3,n_1,n_2,n_3,n_1,n_3,n_2)\\
          &\qquad \qquad +(n_3,n_1,n_2,n_3,n_2,n_1,n_3)+(n_3,n_1,n_2,n_1,n_3,n_2,n_3)],\\
        \end{align*}
        where the green part of the graph means that there is no cancelled 6-subtuple.\\
        For \eqref{B-part 1}, recall that $C_6=C+\eta v^6$ (cf. \eqref{A.10}). Hence, we have 
        \[(n_1,n_2,n_3,n_1,n_2,n_3,X)\Rightarrow G_0 C_6 \wtg V G_0=G_0 C \wtg V G_0+\mathcolor{orange}{\eta(<6,1>)}.\]
        Such an argument also applies to  $(n_1,n_3,n_2,n_1,n_2,n_3,X),(n_1,n_2,n_3,n_2,n_1,n_3,X)$. However,  for $(n_1,n_2,n_1,n_3,n_2,n_3,X)$, since  it in fact contains  no cancelled 4-subtuple (i.e., $X\neq n_2$), we have 
        \begin{align*}
            (n_1,n_2,n_1,n_3,n_2,n_3,X) &\Rightarrow G_0 C_6 \wtg V G_0-(n_1,n_2,n_1,n_3,n_2,n_3,n_2)\\
               &=G_0 C \wtg V G_0-G_0 S^{\top} V G_0+\mathcolor{orange}{\eta(<6,1>)}. 
        \end{align*}
        From the above analysis,  it follows that 
        \begin{equation}\label{B.14}
          \eqref{B-part 1}\Rightarrow -[\eqref{4.22}+\eqref{4.24}]+\mathcolor{orange}{4\eta(<6,1>+<1,6>)}. 
        \end{equation}
         For \eqref{B-part 2}, we have  
        \begin{equation}\label{B.15}
          \eqref{B-part 2}\Rightarrow -\eqref{4.23}. 
        \end{equation}
       Hence, combining  \eqref{B.14}, \eqref{B.15} and \eqref{B.12} yields 
        \begin{align}\label{B generate}
          B &\Rightarrow -[\eqref{4.22}\sim \eqref{4.24}]+\mathcolor{orange}{4\eta(<6,1>+<1,6>)}. 
        \end{align}
        Finally, recall that $<1,1,1,1,1,1,1>\Rightarrow -\eqref{4.18}+A+B$, and there is the  coefficient $-1$ in front of the  graph $<1,1,1,1,1,1,1>$. By \eqref{A generate} and \eqref{B generate}, we get
        \begin{align*}
            -&<1,1,1,1,1,1,1> \Rightarrow [\eqref{4.18}\sim\eqref{4.24}] -\mathcolor{orange}{(G_0 V \wtg R_6 G_0+G_0 R_6 \wtg V G_0)}\\
            &\qquad -\mathcolor{orange}{(\sigma^3-\rho)\cdot (<1,1,1,4>+<1,1,4,1>+\cdots)}-\mathcolor{orange}{4\eta(<6,1>+<1,6>)}. 
        \end{align*}
   \end{itemize}
   \smallskip

  Now summarizing all analysis in the above items, we get that, for the 7th-order remaining terms, ``graphs  with  only odd number of   connected  components'' implies 
  \begin{align*}
    \label{offset 3}\tag{B.17}
    &\qquad  [\eqref{4.11}\sim\eqref{4.24}] + \mathcolor{orange}{4\sigma^2(\sigma^3-\rho) (<6,1>+<1,6>)} +\mathcolor{orange}{\sigma^2(\sigma^3-\rho)(<4,3>+<3,4>)} \\
    &\qquad -\mathcolor{orange}{(\sigma^3-\rho)\cdot (<1,1,1,4>+<1,1,4,1>+\cdots)}-\mathcolor{orange}{4\eta(<6,1>+<1,6>)}\\
    &\qquad -\mathcolor{orange}{(G_0 V \wtg R_6 G_0+G_0 R_6 \wtg V G_0)}\\
    &=[ \eqref{4.11}\sim\eqref{4.24}] -\mathcolor{orange}{(4\eta-4\sigma^5+4\sigma^2\rho) (<6,1>+<1,6>)} - \mathcolor{orange}{\sigma^2(\rho-\sigma^3)(<4,3>+<3,4>)} \\
    &\qquad -\mathcolor{orange}{(\sigma^3-\rho)\cdot (<1,1,1,4>+<1,1,4,1>+\cdots)} -\mathcolor{orange}{(G_0 V \wtg R_6 G_0+G_0 R_6 \wtg V G_0)}. 
  \end{align*}
  The part marked in  orange color is the singular one obtained  from graphs with an  odd number of  connected  components. Again, as we have seen in \eqref{offset 2}, \textbf{amazingly they offset with the remaining graphs with even number of  connected  components in the table!}  So, we finally prove that the 7th-order remaining term is exactly $\eqref{4.11}\sim \eqref{4.24}$.

\section{Proofs of  technical lemmas}\label{tecApp}

In this section, we provide detailed proofs of some technical lemmas. 
\begin{proof}[Proof of Lemma \ref{Lemma 2.1}]
Without loss of generality,  we assume $a\leq b$, so $b\neq d$.  The summation can be decomposed as  
  \begin{align}\label{2.4}
      (\sum_{n_1=0}+&\sum_{n_1=m}+ \sum_{n_1\neq 0,m \atop |n_1|> 2|m|}+\sum_{n_1\neq 0,m \atop \frac12 |m|<|n_1|\leq 2|m|}  +\sum_{n_1\neq 0,m \atop |n_1|\leq \frac{1}{2}|m|}) \frac{1}{|n_1|^a |m-n_1|^b} \\
   \notag  &\lesssim \frac{1}{|m|^a}+\frac{1}{|m|^b}+ ( \sum_{n_1\neq 0,m \atop |n_1|> 2|m|}+\sum_{n_1\neq 0,m \atop \frac12 |m|<|n_1|\leq 2|m|}  +\sum_{n_1\neq 0,m \atop |n_1|\leq \frac{1}{2}|m|}) \frac{1}{|n_1|^a |m-n_1|^b}. 
  \end{align}
  
  First, note that if $|n_1|>2|m|$, then $ |m-n_1|>\frac12 |n_1|$ and we have 
  \begin{align}\label{2.5}
    \sum_{n_1\neq 0,m \atop |n_1|>2|m|}\frac{1}{|n_1|^a|m-n_1|^b} &\lesssim_a\sum_{n_1\in \Z^d \atop |n_1|>2|m|}\frac{1}{|n_1|^{a+b}}\lesssim_{a,d}\sum_{L\in \Z_+ \atop L>2|m|}\frac{1}{L^{a+b+1-d}}\\
    \notag  &\lesssim_{a,b,d}\frac{1}{|m|^{a+b-d}},
  \end{align}
  where the last inequality needs  $a+b>d$.  Similarly,  from $\frac12|m|<|n_1|\leq 2|m|\Rightarrow |n_1-m|\leq 3|m|$,   it follows that  
  \begin{align*}
    \sum_{n_1\neq 0,m \atop \frac{1}{2}|m|<|n_1|\leq 2|m|}\frac{1}{|n_1|^a|m-n_1|^b} &\lesssim_a |m|^{-a}\sum_{n_1\in \Z^d \atop |n_1-m|\leq 3|m|}\frac{1}{|n_1-m|^{b}}\lesssim_{a,d} |m|^{-a}\sum_{L\in \Z_+ \atop L\leq 3|m|}\frac{1}{L^{b+1-d}}. 
  \end{align*}
Thus,
  \begin{itemize}
    \item  if  $b< d$, we have 
            \[\sum_{L\in \Z_+ \atop L\leq 3|m|}\frac{1}{L^{b+1-d}} \lesssim_{b,d} |m|^{d-b}.\]
    \item if  $b=d$, we have 
            \[\sum_{L\in \Z_+ \atop L\leq 3|m|}\frac{1}{L} \lesssim \log |m|.\]
    \item  if  $b> d$, we have 
            \[\sum_{L\in \Z_+ \atop L\leq 3|m|}\frac{1}{L^{b+1-d}} \leq \sum_{L\in \Z_+}\frac{1}{L^{b+1-d}}\lesssim_{b,d} 1.\]
  \end{itemize}
Since  $b\neq d$ and by taking account of estimates in  the above cases, we get  
  \begin{align}\label{2.6}
    \sum_{n_1\neq 0,m \atop \frac{1}{2}|m|<|n_1|\leq 2|m|}\frac{1}{|n_1|^a|m-n_1|^b} \lesssim_{a,d,b}\frac{1}{|m|^{\min\{a,a+b-d\}}}. 
  \end{align}

  Next, note that $|n_1|\leq \frac{1}{2}|m|\Rightarrow |m-n_1|\geq \frac{1}{2}|m|$ and hence,
  \begin{align*}
    \sum_{n_1\neq 0,m \atop |n_1|\leq \frac{1}{2}|m|}\frac{1}{|n_1|^a|m-n_1|^b} \lesssim_{b}|m|^{-b}\sum_{n_1\in \Z^d \atop |n_1|\leq \frac12|m|}\frac{1}{|n_1|^{a}}. 
  \end{align*}
  If $a=b$, using same argument as above shows 
  \begin{align*}
    \sum_{n_1\neq 0,m \atop |n_1|\leq \frac{1}{2}|m|}\frac{1}{|n_1|^a|m-n_1|^b}  \lesssim_{a,d,b}\frac{1}{|m|^{\min\{b,a+b-d\}}}. 
  \end{align*}
  If $a<b$, the only difference is the  case of $a=d$, where we still have  
  \begin{align*}
    \sum_{n_1\neq 0,m \atop |n_1|\leq \frac{1}{2}|m|}\frac{1}{|n_1|^a|m-n_1|^b} \lesssim_{a,d,b}\frac{\log |m|}{|m|^b}\lesssim_{a,d,b} \frac{1}{|m|^{b-}}\lesssim_{a,d,b} \frac{1}{|m|^{a}}.  
  \end{align*}
  Hence, we have for $a\leq b,$
  \begin{align}\label{2.7}
    \sum_{n_1\neq 0,m \atop |n_1|\leq \frac{1}{2}|m|}\frac{1}{|n_1|^a|m-n_1|^b}  \lesssim_{a,d,b}\frac{1}{|m|^{\min\{a,b,a+b-d\}}}. 
  \end{align}

Finally,  combining all estimates  \eqref{2.4}--\eqref{2.7} together  concludes   the proof of Lemma \ref{Lemma 2.1}. 
\end{proof}

\begin{proof}[Proof of Lemma \ref{Lemma 2.2}]
Consider  first  the special cases  of  $n_1=n$ or $n_1=n'$ or $n_1=0$ in the summation.  Indeed, we have 
  \begin{itemize}
    \item if $n_1=n$, then  by $a\leq  b,$ 
          \[\frac{1}{|n|^{\varepsilon}|n-n'|^b}\leq \frac{1}{|n-n'|^a(|n|\wedge |n'|)^{\varepsilon}}.\]
    \item if $n_1=n'$, then  
          \[\frac{1}{|n'|^{\varepsilon}|n-n'|^a}\leq \frac{1}{|n-n'|^a(|n|\wedge |n'|)^{\varepsilon}}.  \]
    \item if $n_1=0$, by $a\leq b$, 
          \begin{align*}
            \frac{1}{|n|^a|n'|^b} \leq \frac{1}{(|n|\cdot|n'|)^a}=\frac{1}{(|n|\vee |n'|)^a(|n|\wedge|n'|)^a}.
          \end{align*}
         In this case,  since  $|n|+|n'|\geq |n-n'|$,  we have  $|n|\vee|n'|\geq \frac{1}{2}|n-n'|$. This implies  
          \begin{align*}
            \frac{1}{|n|^a|n'|^b} \lesssim_a \frac{1}{|n-n'|^a(|n|\wedge|n'|)^{a}}. 
          \end{align*}
  \end{itemize}

  Next,  we aways assume $n_1\neq 0,n,n'$. 
  From $|n-n_1|+|n_1-n'|\geq |n-n'|$, it follows that  either $|n_1-n|\geq \frac{1}{2}|n-n'|$ or $|n_1-n'|\geq \frac12|n-n'|$.  Then we can decompose the summation as 
  \begin{equation}\label{2.10}
    \sum_{n_1\neq 0,n,n'}= \sum_{n_1\neq 0,n,n' \atop |n_1-n|\geq \frac{1}{2}|n-n'|}+\sum_{n_1\neq 0,n,n' \atop |n_1-n'|\geq \frac{1}{2}|n-n'|}. 
  \end{equation}
  Applying Lemma \ref{Lemma 2.1} with $b+\varepsilon>d,b\neq d$ yields 
  \begin{align}\label{2.11}
    \sum_{n_1\neq 0,n,n' \atop |n_1-n|\geq \frac12 |n-n'|}\frac{1}{|n-n_1|^a |n_1|^{\varepsilon} |n_1-n'|^b} &\lesssim_a |n-n'|^{-a}\sum_{n_1\neq 0,n'} \frac{1}{|n_1|^{\varepsilon}|n_1-n'|^b}\\
    \notag   &\lesssim_{a,b,\varepsilon,d}|n-n'|^{-a}|n'|^{-\min\{\varepsilon,b+\varepsilon-d\}}. 
  \end{align}
 So, it remains to deal with the summation satisfying $ |n_1-n'|\geq \frac12 |n-n'|.$  In fact, we have  
  \begin{align}
  \nonumber \ \ \  &\sum _{n_1\neq 0,n,n' \atop |n_1-n'|\geq \frac12 |n-n'|}\frac{1}{|n-n_1|^a |n_1|^{\varepsilon} |n_1-n'|^{b}} \\
   \label{2.12}& \lesssim_a |n-n'|^{-a}\sum_{n_1\neq 0,n,n'} \frac{1}{|n-n_1|^a|n_1|^{\varepsilon}|n_1-n'|^{b-a}}. 
  \end{align}
  We divide the discussion  into the following two cases. 
  \begin{itemize}
    \item[{\bf Case 1}:]  $b>d$. In this case, we decompose
       \begin{equation}\label{2.13}
        \sum_{n_1\neq 0,n,n'} \frac{1}{|n-n_1|^a|n_1|^{\varepsilon}|n_1-n'|^{b-a}} =(\sum_{n_1\neq 0,n,n' \atop |n_1|\leq \frac12(|n|\wedge|n'|)} + \sum_{n\neq 0,n,n' \atop |n_1|> \frac12 (|n|\wedge |n'|)}) \cdots.
       \end{equation}
  On one hand, by Lemma \ref{Lemma 2.1}, $b>d$ and Remark  \ref{rmk 2.1}  (1), we have
  \begin{align}\label{2.14}
    \sum_{n_1\neq 0,n,n' \atop |n_1|> \frac12(|n|\wedge|n'|)} \frac{1}{|n-n_1|^a|n_1|^{\varepsilon}|n_1-n'|^{b-a}} &\lesssim_{\varepsilon} (|n|\wedge |n'|)^{-\varepsilon}\sum_{n_1\neq n,n'}\frac{1}{|n-n_1|^a|n_1-n'|^{b-a}}\\
     \notag  &\lesssim_{a,b,\varepsilon,d} (|n|\wedge|n'|)^{-\varepsilon}\frac{1}{|n-n'|^{\min\{a,b-a,b-d\}-}}\\
     \notag  &\lesssim_{a,b,\varepsilon,d} (|n|\wedge|n'|)^{-\varepsilon}. 
  \end{align}
  On the other hand, from  $|n_1|\leq \frac{1}{2}(|n|\wedge|n'|)$, it follows that  $|n-n_1|\geq \frac12 |n|$ and $|n'-n_1|\geq \frac{1}{2}|n'|$. As a result, we obtain  
  \begin{align}\label{2.15}
    \sum_{n_1\neq 0,n,n' \atop |n_1|\leq \frac12(|n|\wedge|n'|)} \frac{1}{|n-n_1|^a|n_1|^{\varepsilon}|n_1-n'|^{b-a}} &\lesssim_{a,b} |n|^{-a}|n'|^{-(b-a)}\sum_{0<|n_1|\leq \frac12(|n|\wedge|n'|)}\frac{1}{|n_1|^{\varepsilon}}\\
    \notag   &\lesssim_{a,b,\varepsilon,d} \frac{1}{(|n|\wedge|n'|)^{\varepsilon+b-d}}\\
    \notag   &\lesssim_{a,b,\varepsilon,d} (|n|\wedge|n'|)^{-\varepsilon},
  \end{align}
  where  the second inequality above  depends on $\varepsilon<d$.  Finally, combining  \eqref{2.13}--\eqref{2.15} together  leads to 
  \begin{equation}\label{2.16}
    \sum_{n_1\neq 0,n,n'} \frac{1}{|n-n_1|^a|n_1|^{\varepsilon}|n_1-n'|^{b-a}} \lesssim_{a,b,\varepsilon,d} (|n|\wedge|n'|)^{-\varepsilon} \ \ {\rm if } \ \ b>d. 
  \end{equation}
  \item[{\bf Case 2}:] $a\leq b<d$. In this case,  if  $a=b$, then the estimate   is totally the  same as that of  \eqref{2.12}.  Thus, we only need to consider the case of  $a<b$.  The main difficulty here is  that  we do not necessarily have $a+\varepsilon>d$,  so we  cannot apply Lemma \ref{Lemma 2.1} directly. Fortunately,  since  $a+\varepsilon+(b-a)>d, \varepsilon<d$ and $b-a<b<d$, we can find   $p,q>0 $ such that $\frac{1}{p}+\frac{1}{q}=1$ and 
    \[d-\varepsilon<pa<d,\ d-\varepsilon<q(b-a)<d.\]
    In fact,  $p=\frac ba$ suffices for the purpose.  Then applying H\"older inequality and Lemma \ref{Lemma 2.1} shows 
    \begin{align}\label{2.17}
      \sum_{n_1\neq 0,n,n'} & \frac{1}{|n-n_1|^a|n_1|^{\varepsilon}|n_1-n'|^{b-a}} \\
      \notag &\leq (\sum_{n_1\neq 0,n,} \frac{1}{|n-n_1|^{pa}|n_1|^{\varepsilon}})^{\frac{1}{p}} (\sum_{n_1\neq 0,n'}  \frac{1}{|n_1|^\varepsilon|n_1-n'|^{q(b-a)}})^{\frac{1}{q}} \\
      \notag &\lesssim_{a,b,\varepsilon,d}|n|^{-\frac{pa+\varepsilon-d}{p}}|n'|^{-\frac{\varepsilon+q(b-a)-d}{q}}\\
      \notag &\lesssim_{a,b,\varepsilon,d} (|n|\wedge |n'|)^{-(b+\varepsilon-d)} \ \ {\rm if } \ \ b<d. 
    \end{align}
\end{itemize}
Thus, combining estimates  \eqref{2.16}, \eqref{2.17} and \eqref{2.12} in the above two cases together implies 
\begin{equation}\label{2.18}
  \sum _{n_1\neq 0,n,n' \atop |n_1-n'|\geq \frac12 |n-n'|}\frac{1}{|n-n_1|^a |n_1|^{\varepsilon} |n_1-n'|^{b}} \lesssim_{a,b,\varepsilon,d}  |n-n'|^{-a}(|n|\wedge|n'|)^{-\min\{\varepsilon,a, b+\varepsilon-d\}}. 
\end{equation}

Finally,  the proof of Lemma \ref{Lemma 2.2} follows by combining  estimates \eqref{2.10}, \eqref{2.11} and \eqref{2.18} together.
\end{proof}

\begin{proof}[Proof of Lemma \ref{difference property}]
 Without loss of generality, we can assume $|n_1|\geq |n_2|$. Note that 
    \[ ||n_1|^{-\alpha}-|n_2|^{-\alpha}| =\frac{||n_1|^{\alpha}-|n_2|^{\alpha}|}{|n_1|^{\alpha}|n_2|^{\alpha}}.\]
  When $\alpha\geq 1$, it's easy to check that $(\alpha-1)t+\frac{1}{t^{\alpha-1}}\geq \alpha,t\geq 1$. That is,
    \[t^{\alpha}-1\leq \alpha (t-1)t^{\alpha-1}.\]
    By taking $t=\frac{|n_1|}{|n_2|}$,  we get 
    \[|n_1|^{\alpha}-|n_2|^{\alpha}\leq \alpha (|n_1|-|n_2|)|n_1|^{\alpha-1}.\]
    Hence,
    \begin{align*}
      ||n_1|^{-\alpha}-|n_2|^{-\alpha}|&\leq \alpha \frac{|n_1|-|n_2|}{|n_1||n_2|^{\alpha} }\leq\frac{|n_1-n_2|}{|n_1||n_2|^{\alpha} } \\
         &\leq \frac{|n_1-n_2|}{\frac{|n_1|+|n_2|}{2}|n_2|^{\alpha} }\lesssim_{\alpha}\frac{|n_1-n_2|}{(|n_1|+|n_2|)\cdot (|n_1|\wedge|n_2|)^{\alpha}}.
    \end{align*}
 When $0<\alpha<1$, one can check that $t^{\alpha}-1\leq 10 \frac{t-1}{(t+1)^{1-\alpha}},t\geq 1$. We again take $t=\frac{|n_1|}{|n_2|}$ and  get 
    \[|n_1|^{\alpha}-|n_2|^{\alpha}\leq 10\frac{|n_1|-|n_2|}{(|n_1|+|n_2|)^{1-\alpha}}.\]
    Hence,
    \begin{align*}
      ||n_1|^{-\alpha}-|n_2|^{-\alpha}|&\leq 10 \frac{|n_1|-|n_2|}{(|n_1|+|n_2|)^{1-\alpha}|n_1|^{\alpha}|n_2|^{\alpha} } \\
         &\leq 10 \frac{|n_1|-|n_2|}{(|n_1|+|n_2|)^{1-\alpha}(\frac{|n_1|+|n_2|}{2})^{\alpha}|n_2|^{\alpha} }\\
         &\lesssim_{\alpha}\frac{|n_1-n_2|}{(|n_1|+|n_2|)\cdot (|n_1|\wedge|n_2|)^{\alpha}}.
    \end{align*}

\end{proof}

\section{Proof of Lemma \ref{FGNlem}}\label{GNineq}

We begin with  the  fractional Gagliardo-Nirenberg inequality  proven in  \cite{BM18}.
\begin{lem}[\cite{BM18}]\label{GNlem}
  Let \(\Omega \subset \mathbb{R}^d\) be either the whole space, a half-space or a bounded Lipschitz domain. Let \(1 \leq p, \, p_1, \, p_2 \leq +\infty\) be three positive extended real quantities and let \(s, \, s_1, \, s_2\) be non-negative real numbers. Furthermore, let \(\theta \in (0, 1)\) and assume that
\[
s_1 \leq s_2, \  s = \theta s_1 + (1 - \theta)s_2, \  \frac{1}{p} = \frac{\theta}{p_1} + \frac{1 - \theta}{p_2}
\]
hold. Then
\[
\|u\|_{W^{s,p}(\Omega)} \leq C\|u\|^{\theta}_{W^{s_1,p_1}(\Omega)}\|u\|^{1-\theta}_{W^{s_2,p_2}(\Omega)}
\]
for any \(u \in W^{s_1,p_1}(\Omega) \cap W^{s_2,p_2}(\Omega)\) if and only if at least one of
\[
\begin{cases}
s_2 \in \mathbb{N} \text{ and } s_2 \geq 1,\\
p_2 = 1, \\
0 < s_2 - s_1 \leq 1 - \frac{1}{p_1},
\end{cases}
\]
is false. The constant \(C > 0\) depends on the parameters \(p, \, p_1, \, p_2, \, s, \, s_1, \, s_2, \, \theta\), on the domain \(\Omega\), but not on \(u\). Here $W^{s,p}(\Omega)$ can be  both the Bessel potential space and the Sobolev-Slobodeckij space.

\end{lem}

\begin{proof}[Proof of  Lemma \ref{FGNlem}]
  (1)   Note first  that when $\beta\in \Z_+^d$, we have near $\xi=0,$
  \[\partial^{\beta}\hat{G_0}(\xi)=\mcO(\frac{1}{\|\xi\|^{2+|\beta|_1}}).\]
  Hence, for $0\leq k\leq d-3, k\in\Z$, we have that $p(2+k)<d \Rightarrow \nm \hat{G_0}\nm_ {W^{k,p}}<\infty$.
  
When $|\beta|_1\notin \Z_+,\ |\beta|_1<d-3$, we set  $k=\lfloor |\beta|_1\rfloor$ (i.e., the smallest integer larger than  $ |\beta|_1$) and hence,
  \[k-1<|\beta|_1<k,\ p(2+|\beta|_1)<d.\]
  By Lemma \ref{GNlem}, we have
  \begin{equation}\label{BM1}
       \|\partial^{\beta} \hat{G_0}\|_{L^p} \lesssim \|\hat{G_0}\|_{W^{k-1,p_1}}^{\theta}\cdot \| \hat{G_0}\|_{W^{k,p_2}}^{1-\theta}, 
  \end{equation} 
  where $|\beta|_1=(k-1)\theta+k(1-\theta),\frac{1}{p}=\frac{\theta}{p_1}+\frac{1-\theta}{p_2}\Rightarrow\theta=k-|\beta|_1$. Thus, if we can find  $p_1,p_2$ such that 
  \[
  \begin{cases}
    p_1(k-1+2)<d, \\
    p_2 (k+2)<d, \\
    \frac{1}{p}=\frac{k-|\beta|_1}{p_1}+\frac{1-k+|\beta|_1}{p_2}. 
  \end{cases}
  \]
  Then by Lemma \ref{GNlem}, we can prove that $\|\partial^{\beta} \hat{G_0}\|_{L^p}<\infty$. The system can be rewritten as 
  \[
  \begin{cases}
    \frac{1}{p_1}=(-\frac{1-k+|\beta|_1}{p_2}+\frac{1}{p})/ (k-|\beta|_1),\\
    \frac{k+1}{d}<\frac{1}{p_1}<(-(1-k+|\beta|_1)\frac{k+2}{d}+\frac{1}{p})/ (k-|\beta|_1).
  \end{cases}
  \]
  The system has a solution if and only if 
  \[    \frac{k+1}{d}<(-(1-k+|\beta|_1)\frac{k+2}{d}+\frac{1}{p})/ (k-|\beta|_1)  \Leftrightarrow p(2+|\beta|_1)<d.\]
 Hence, we have proven  the result for $|\beta|_1\leq d-3$.

  The  harder case is  $d-3<|\beta|_1<d-2$, since it may be   $\|\hat{G_0}\|_{W^{d-2,1}}=+\infty$.  We  need to make the dyadic   decomposition. Define the  non-negative $\phi(\xi)\in C^{\infty}_0(\R^d)$ as 
  \[\phi(\xi)=\begin{cases}
    (1+e^{\frac{1}{1-\|\xi\|}+\frac{1}{2-\|\xi\|}})^{-1},\ 1\leq \|\xi\|\leq 2, \\
    1-(1+e^{\frac{1}{1-2\|\xi\|}+\frac{1}{2-2\|\xi\|}})^{-1},\ \frac{1}{2}\leq \|\xi\|\leq 1, \\
    0, \ {\rm otherwise}. 
  \end{cases}\]
  Then $\sum_{j=-\infty}^{\infty}\phi(2^j\xi)\equiv 1$ and $\supp\  \phi=[\frac{1}{2},2]$ (i.e., the support of $\phi$). For $\alpha\in\Z_+^d$ with $|\alpha|_1=d-3, $
  \[\partial^{\alpha}\hat{G_0}(\xi)=\sum_{j=1}^{\infty}\phi(2^j\xi)\partial^{\alpha}\hat{G_0}(\xi)\ {\rm for}\  \forall \xi\in [0,1]^d.  \]
  Denote $g=\partial^{\alpha}\hat{G_0}$ and $g^{(1)}=\partial g$. Now assume 
  \[|\beta|_1\in (0,1),\ (d-3+|\beta|_1+2)p<d.\]
  Then  
  \begin{align*}
    \| \partial^{\beta}g \|_{L^p} &=\|\sum_{j=1}^{\infty} \partial^{\beta} (\phi(2^j\xi)g(\xi))\|_{L^p}  \\
      &\leq  \sum_{j=1}^{\infty} \| \partial^{\beta} (\phi(2^j\xi)g(\xi))\|_{L^p}. 
  \end{align*}
  By Lemma \ref{GNlem} again, we have
  \begin{align*}
    \| \partial^{\beta} (\phi(2^j\xi)g(\xi))\|_{L^p} \lesssim \| \phi(2^j\xi)g(\xi)\|_{L^{p_1}}^{\theta} \cdot \| \partial (\phi(2^j\xi)g(\xi))\|_{L^1}^{1-\theta},
  \end{align*}
  where  we take 
  \[
  \begin{cases}
    |\beta|_1=0\cdot \theta +1\cdot (1-\theta)\Rightarrow \theta=1-|\beta|_1, \\
    \frac{1}{p}=\frac{\theta}{p_1}+\frac{1-\theta}{1}\Rightarrow \frac{1}{p_1}=(\frac{1}{p}-|\beta|_1)/ (1-|\beta|_1). 
  \end{cases}\]
  Now, as $g=\mcO(\frac{1}{\|\xi\|^{d-1}})$, we get 
  \begin{align*}
    \| \phi(2^j\xi)g(\xi)\|_{L^{p_1}} &\lesssim \left(\int \phi^{p_1}(2^j\xi) \frac{1}{\|\xi\|^{(d-1)p_1}}{\rm d}\xi\right)^{\frac{1}{p_1}} \\
          & \lesssim \left \|\phi(\xi)\frac{1}{\|\xi\|^{d-2}}\right \|_{L^{p_1}} \cdot [2^{\frac{1}{p_1}(d-(d-1)p_1)}]^{-j}. 
  \end{align*}
  On the other hand, due to $g^{(1)}=\mcO(\frac{1}{\|\xi\|^{d}})$, using the  similar argument shows
  \begin{align*}
   \sup_{j} \| \partial (\phi(2^j\xi)g(\xi))\|_{L^1} & \leq \sup_{j}(\| \phi(2^j\xi)g^{(1)}(\xi)\|_{L^1} +\| 2^j (\partial\phi) (2^j\xi)g(\xi)\|_{L^1}) \\
     &<+\infty.
  \end{align*}
  Combining  all above  estimates gives 
  \begin{equation*}
  \| \partial^{\beta}g \|_{L^p} \lesssim \sum_{j=1}^{\infty} [2^{\frac{1}{p_1}(d-(d-1)p_1)}]^{-j}. 
  \end{equation*}
  To ensure the convergence of the above series, it requires
  \[d-(d-1)p_1>0 \Leftrightarrow (d-1+|\beta|_1)p<d.\]
  This proves  the result for $d-3<|\beta|_1< d-2$.\medskip

  \noindent (2) By Lemma \ref{GNlem} again, we obtain 
  \[\| \partial^{\beta}(\hat{G_0}-\sigma)  \|_{L^p}<\infty \ {\rm if} \ |\beta|_1<d-2,\ p\geq 1, \ (2+|\beta|_1)p<d. \]
  Now assume $\alpha, \beta \in \Z_+^d$ with $|\alpha|_1+|\beta|_1=k\leq 2(d-2)-1,k\in \Z_+$. Then  
  \[\|\partial^{\alpha+\beta}(f^2)\|_{L^1} \leq \sum_{|\alpha|_1 +|\beta|_1 =k} \|\partial^\alpha  f\cdot  \partial^{\beta} f \|_{L^1}.  \]
  Hence, from  H\"older's inequality and Young's  inequality, it follows that 
  \begin{align*}
    \|\partial^\alpha  f\cdot  \partial^{\beta} f \|_{L^1}  \leq& \|\partial^\alpha f\|_{L^{r_{\alpha}}}\|\partial^\beta f \|_{L^{r_{\beta}}} \\
          =& \|\partial^{\alpha_1}(\hat{G_0}-\sigma) * \partial^{\alpha_2}(\hat{G_0}-\sigma)\|_{L^{r_{\alpha}}}  \|\partial^{\beta_1}(\hat{G_0}-\sigma) *\partial^{\beta_2} (\hat{G_0}-\sigma) \|_{L^{r_{\beta}}} \\
          \leq& \|\partial^{\alpha_1}(\hat{G_0}-\sigma)\|_{L^{p_{\alpha}}} \cdot \|\partial^{\alpha_2} (\hat{G_0}-\sigma)\|_{L^{q_{\alpha}}} \\
          &\cdot \|\partial^{\beta_1}(\hat{G_0}-\sigma)\|_{L^{p_{\beta}}} \cdot \|\partial^{\alpha_2}(\hat{G_0}-\sigma)\|_{L^{q_{\beta}}} \\
          <&\infty,
  \end{align*}
  where $(\alpha_1,\alpha_2,\beta_1,\beta_2,r_\alpha,r_\beta,p_\alpha,p_\beta,q_\alpha,q_\beta)$ must satisfy  (we remark that $\alpha_1,\alpha_2,\beta_1,\beta_2$ need not be integer vectors)  
  \[\begin{cases}
    |\alpha|_1+|\beta|_1 =k \leq 2d-5\ (d\geq 5), \\
    \alpha_1+ \alpha_2=\alpha, \\
    \beta_1+ \beta_2=\beta, \\
    \frac{1}{r_\alpha}+\frac{1}{r_\beta}=1, \\
    \frac{1}{r_\alpha}+1 =\frac{1}{p_\alpha}+\frac{1}{q_\alpha},\\
    \frac{1}{r_\beta}+1 =\frac{1}{p_\beta}+\frac{1}{q_\beta},
  \end{cases}\]
  and \[
   \begin{cases}
    p_\alpha(2+|\alpha_1|_1)<d, \\
    q_\alpha(2+|\alpha_2|_1)<d, \\
    p_\beta(2+|\beta_1|_1)<d, \\
    q_\beta(2+|\beta_2|_1)<d. 
   \end{cases} 
  \]
  Considering the  symmetry, we can  take
  \begin{align*}
  \alpha_1&=\alpha_2=\frac{\alpha}{2}, \ \beta_1=\beta_2=\frac{\beta}{2},\\
  p_\alpha&=q_\alpha=\frac{2r_\alpha}{r_\alpha+1},\  p_\beta=q_\beta=\frac{2r_\beta}{r_\beta+1}=\frac{2r_\alpha}{2r_\alpha-1}.
  \end{align*}
  So, we need to find $r_\alpha>1$ such that 
  \[\frac{2r_\alpha}{r_\alpha+1}(2+\frac{|\alpha|_1}{2})<d,\ \frac{2r_\alpha}{2r_\alpha-1}(2+\frac{|\beta|_1}{2})<d.\]
  Denote $c=\frac{2r_\alpha}{r_\alpha+1}\in(1,2)$. We need to  find $c$ such that 
  \[c(2+\frac{|\alpha|_1}{2})<d,\ \frac{2c}{3c-2}(2+\frac{|\beta|_1}{2})<d \]
  which is equivalent to find $c$ such that 
  \[1<c<2,\ \frac{2d}{3d-4-|\beta|_1}<c<\frac{2d}{4+|\alpha|_1}.\]
  Such a $c$ exists if and only if
  \[\begin{cases}
    \frac{2d}{3d-4-|\beta|_1}<2, \\
    \frac{2d}{4+|\alpha|_1} >1, \\
    \frac{2d}{3d-4-|\beta|_1}<\frac{2d}{4+|\alpha|_1},
  \end{cases}
  \qquad \Longleftrightarrow \qquad
  \begin{cases}
    |\beta|_1<2(d-2),\\
    |\alpha|_1<2(d-2),\\
    8+k<3d.
  \end{cases}\]
  This can  be ensured by $k=|\alpha|_1+|\beta|_1\leq 2(d-2)-1<3d-8$  since  $d\geq 5$.  We have proven (by carefully selecting  parameters as above)  that 
  \begin{equation}\label{BM3}
    \| f^2 \|_{W^{k, 1}}<\infty  \ {\rm for} \ k\in \Z_+,\ k\leq 2(d-2)-1. 
  \end{equation}

  Finally, for  the non-integer $|\beta|_1\leq 2(d-1)-1$, if we take $k=\lfloor|\beta|_1\rfloor$,  using Lemma \ref{GNlem} again shows that
  \begin{equation}\label{BM4}
    \| \partial^{\beta} f^2\|_{L^1}\lesssim \| f^2\|_{W^{k-1, 1}}^{\theta}\cdot\|  f^2\|_{W^{k, 1}}^{1-\theta}<\infty.  
  \end{equation}
  Combining  \eqref{BM3} and \eqref{BM4}  concludes the proof of Lemma \ref{FGNlem}.
\end{proof}

\section*{Acknowledgement}
This work  is  supported by  the National Key R\&D Program
of China under Grant 2023YFA1008801.
 Y. Shi is supported by the NSFC(12271380).   Z. Zhang is  supported by the NSFC (12288101).   
The authors would like to thank Prof. T. Spencer  for sharing with them the references \cite{Spe93,Elg09}.

\section*{Data Availability}
		The manuscript has no associated data.
		\section*{Declarations}
		{\bf Conflicts of interest} \ The authors  state  that there is no conflict of interest.


\end{document}